\documentclass[times,sort&compress,3p]{elsarticle}
\usepackage[labelfont=bf]{caption}

\usepackage{amsmath,amsfonts,amssymb,amsthm,booktabs,color,epsfig,graphicx,hyperref,url}

\theoremstyle{plain}
\newtheorem{theorem}{Theorem}

\newtheorem{proposition}{Proposition}
\newtheorem{lemma}{Lemma}

\theoremstyle{definition}
\newtheorem{definition}{Definition}
\newtheorem{remark}{Remark}
\newtheorem{example}{Example}	

\makeatletter
\def\ps@pprintTitle{%
  \let\@oddhead\@empty
  \let\@evenhead\@empty
  \let\@oddfoot\@empty
  \let\@evenfoot\@oddfoot
}
\makeatother

\usepackage{bbm,graphicx,xurl,dsfont,arydshln,eucal,multirow}

\newcommand\smallO{
  \mathchoice
    {{\scriptstyle\mathcal{O}}}
    {{\scriptstyle\mathcal{O}}}
    {{\scriptscriptstyle\mathcal{O}}}
    {\scalebox{.7}{$\scriptscriptstyle\mathcal{O}$}}
  }
  
\makeatletter
\newcommand*\bigcdot{\mathpalette\bigcdot@{.9}}
\newcommand*\bigcdot@[2]{\mathbin{\vcenter{\hbox{\scalebox{#2}{$\m@th#1\bullet$}}}}}
\makeatother

\usepackage{float}
\usepackage{placeins} 
\usepackage{color}

\begin{document}

\begin{frontmatter}

\title{High-dimensional copula-based Wasserstein dependence.}

 \author[1]{Steven De Keyser}
 \author[1]{Ir\`{e}ne Gijbels\corref{mycorrespondingauthor}}

 \address[1]{Department of Mathematics, KU Leuven, Celestijnenlaan 200B, B-3001 Leuven (Heverlee), Belgium}

 \cortext[mycorrespondingauthor]{Corresponding author. Email address:  \url{irene.gijbels@kuleuven.be}}

\begin{abstract}
We generalize $2$-Wasserstein dependence coefficients to measure dependence between a finite number of random vectors. This generalization includes theoretical properties, and in particular focuses on an interpretation of maximal dependence and an asymptotic normality result for a proposed semi-parametric estimator under a Gaussian copula assumption. In addition, we discuss general axioms for dependence
measures between multiple random vectors, other plausible normalizations, and various examples. Afterwards, we look into plug-in estimators based on penalized empirical covariance matrices in order to deal with high dimensionality issues and take possible marginal independencies into account by inducing (block) sparsity. The latter ideas are investigated via a simulation study, considering other dependence coefficients as well. We illustrate the use of the developed methods in two real data applications.
\end{abstract}

\begin{keyword} 
Copula \sep
Normal scores rank correlation \sep
Regularization \sep 
Sparsity \sep
Wasserstein dependence\\[1.12 ex]
\end{keyword}
\end{frontmatter}
\noindent

\section{Introduction}\label{sec:1}

A prominent line of research in statistics considers measuring dependence between groups of variables. In case of two groups, greatly celebrated is the canonical correlation analysis of \cite{Hotelling1936} relying on the Pearson correlation coefficient. To step away from the assumption of Gaussianity, concordance measures as studied in, e.g., \cite{Gijbels2021,Nelsen1996,Schmid2007} among many others, are also used for quantifying general monotonic associations between two random vectors in \cite{Grothe2014}. In \cite{Hofert2019}, measures of association computed
from collapsed random variables are used to measure the dependence between
random vectors.  Fundamental is copula theory (e.g., \cite{Nelsen2006, Sklar1959}), allowing to split multivariate distributions into marginal distributions on the one hand, and a dependence structure described by the copula on the other hand. Especially when the marginals are continuous, the preference often goes to copula-based dependence measures since then, by Sklar's theorem \cite{Sklar1959}, the copula is unique, and hence margin-free dependence  can unequivocally be quantified. 

Statistical independence between random vectors holds if and only if the true underlying copula is the product of the marginal copulas (where a one dimensional copula is basically a uniform distribution on $[0,1]$), yielding zero dependence. However, the dependence measures mentioned above do not detect all types of departure from independence, meaning that they might vanish while the independence product copula is misspecified. Since the work of \cite{Szekely2007}, there has been a growing interest for dependence measures that completely characterize independence. Some recent developments are, e.g., the Hoeffding's Phi-Square measure of \cite{Medovikov2017}, the $\Phi$-dependence measures of \cite{Gijbels2023} (of which the Hellinger correlation \cite{Geenens2022} and essential dependence \cite{Zhang2020} are particular cases), the coefficient of Chatterjee \cite{Chatterjee2021,Azadkia2021,Fuchs2022,Ansari2023}, and the $2$-Wasserstein coefficients of \cite{Mordant2022}.

The aim of this article is to elaborate further on the optimal transport measures of \cite{Mordant2022}. First, the focus will be on extending their dependence coefficients from two to finitely many random vectors. We do this from a copula point of view. This includes generalizing the results of \cite{Mordant2022}, and also verifying the axioms stated in \cite{Gijbels2023} (see also \ref{App A}). Some additional examples, focusing on, e.g., the impact of the normalization, are provided as well. Afterwards, we dive into the Bures-Wasserstein dependence making a Gaussian copula assumption. This yields dependence measures that are attractive, and more amenable for estimation. The results are a far from straightforward extension of results of \cite{Mordant2022} to the case of a finite number of random vectors, and require significant mathematical care.

The proposed semi-parametric estimation approach of the Bures-Wasserstein coefficients relies on the sample matrix of normal scores rank correlations (see, e.g., \cite{Hajek1967}). Since we extend the theory to a general finite amount of groups of variables, high dimensional cases with a large number of variables compared to the sample size are of study interest as well. Acclaimed penalization techniques are known to significantly improve (inverse) covariance matrix estimation (see, e.g., \cite{Lam2009} and references therein). We utilize these ideas in our Gaussian copula context in order to correct for high dimensionality bias and possibly enforce sparsity at the individual level or group level to take plausible marginal independencies into account.

The outline of this paper is the following. Section \ref{sec:2} explains how optimal transport theory is combined with copula theory in order to arrive at a dependence measure between multiple groups of random variables that completely characterizes independence and is invariant with respect to the univariate marginal distributions. The verification of the properties postulated in \cite{Gijbels2023} for such dependence measures is also part of this section. The Gaussian copula approach is discussed in Section \ref{sec:3}, with special attention to the meaning of maximal dependence, and asymptotic normality of the proposed semi-parametric estimator. Afterwards, different regularization techniques for Gaussian copula covariance matrices are discussed in Section \ref{sec: 4}. Next to an empirical illustration of the asymptotic normality result, simulations are performed in Section \ref{sec:5} to investigate the performance of these regularization techniques on various dependence coefficients for random vectors. Two real data applications are discussed in Section \ref{sec:6}. All proofs are deferred to the Appendix. Any experiments reported can be reproduced via the source code available at \url{https://github.com/StevenDeKeyser98/VecDep}. Additional figures are included in the Supplementary Material.

\section{General $2$-Wasserstein dependence}\label{sec:2}
We consider a $q$-dimensional random vector $\mathbf{X} =  (\mathbf{X}_{1}, \dots, \mathbf{X}_{k})$ defined on $(\mathbb{R}^{q},\mathcal{B}(\mathbb{R}^{q}),\lambda^{q})$ consisting of $k$ marginal random vectors $\mathbf{X}_{i} = (X_{i1},\dots,X_{id_{i}})$ for $i = 1,\dots,k$ having $d_{i}$ continuous univariate marginal random variables $X_{ij}$ for $j = 1,\dots,d_{i}$. The numbers $d_{1}, \dots, d_{k} \in \mathbb{Z}_{> 0}$ are such that $q = d_{1} + \cdots + d_{k}$, and $\lambda^{q}$ denotes the $q$-dimensional Lebesgue measure defined on $\mathcal{B}(\mathbb{R}^{q})$, the Borel sigma-algebra on $\mathbb{R}^{q}$. Let $\mathbb{P}: \mathcal{B}(\mathbb{R}^{q})  \rightarrow \mathbb{R}$ be a probability measure. Our aim is to measure the dependence between $\mathbf{X}_{1}, \dots, \mathbf{X}_{k}$. For $\mathbb{I} = [0,1]$ and $\mathcal{P}(\mathbb{I}^{q})$ the set of Borel probability measures on $\mathbb{I}^{q}$, the random vector $\mathbf{X}$ is assigned a copula probability law $\mu_{C} \in \mathcal{P}(\mathbb{I}^{q})$ having corresponding marginal copula laws $\mu_{C_{i}} \in \mathcal{P}(\mathbb{I}^{d_{i}})$ of $\mathbf{X}_{i}$ for $i = 1,\dots,k$. Note that in case $d_{i} = 1$, the law $\mu_{C_{i}}$ is that of a uniform distribution on $\mathbb{I}$. We denote $\Gamma(\mu_{C_{1}},\dots,\mu_{C_{k}})$ for the set of measures $\gamma \in \mathcal{P}(\mathbb{I}^{q})$ such that $$\mu_{C_{i}}(B_{i}) =  \gamma \left (\mathbb{I}^{d_{1}} \times \cdots \times \mathbb{I}^{d_{i-1}} \times B_{i} \times \mathbb{I}^{d_{i+1}} \times \cdots \times \mathbb{I}^{d_{k}} \right )$$ for all $B_{i} \in \mathcal{B}(\mathbb{I}^{d_{i}})$ and $i = 1, \dots, k$. This is a natural generalization of the set of coupling measures. Obviously, $\mu_{C} \in \Gamma(\mu_{C_{1}},\dots,\mu_{C_{k}})$. Quantifying the intensity of relation between $\mathbf{X}_{1},\dots,\mathbf{X}_{k}$ can be done by measuring the difference between $\mu_{C}$ and the product measure $\mu_{C_{1}} \times \cdots \times \mu_{C_{k}}$. We use the $2$-Wasserstein distance, whose square is, for certain measures $\pi, \widetilde{\pi} \in \mathcal{P}(\mathbb{I}^{q})$ given by
\begin{equation}\label{eq: 2W distance}
    W_{2}^{2}(\pi,\widetilde{\pi})  = \underset{\gamma \in \Gamma(\pi,\widetilde{\pi})}{\inf}  \int_{\mathbb{I}^{2q}} ||\mathbf{u}-\mathbf{v}||^{2} d\gamma(\mathbf{u},\mathbf{v}) 
= \underset{\substack{\mathbf{U} \sim \pi \\ \mathbf{V} \sim \widetilde{\pi}}}{\inf} \mathbb{E} \left  (||\mathbf{U}-\mathbf{V}||^{2} \right  ). 
\end{equation}
The interpretation of the metric \eqref{eq: 2W distance} is optimal transport, see, e.g., \cite{Panaretos2019}.  It is the minimal effort (cost) required to transform the mass of $\pi$ into the mass of $\widetilde{\pi}$, i.e., for every $(\mathbf{u},\mathbf{v})$ transport an  infinitesimal amount of mass $d\gamma(\mathbf{u},\mathbf{v})$ from $\mathbf{u}$ to $\mathbf{v}$ at a distance cost of $||\mathbf{u}-\mathbf{v}||^{2}$. Aggregating the mass $\gamma(\{\mathbf{u}\} \times \mathbb{I}^{q})$ that leaves $\mathbf{u}$ gives $\pi(\mathbf{u})$ and the total mass $\gamma(\mathbb{I}^{q} \times \{\mathbf{v}\})$ that reaches $\mathbf{v}$ equals $\widetilde{\pi}(\mathbf{v})$.
For certain non-degenerate (i.e. not Dirac delta distributions) reference laws $\nu_{i} \in \mathcal{P}(\mathbb{I}^{d_{i}})$ for $i = 1, \dots, k$, we now define 
\begin{equation*}
\begin{split}
T_{d_{1},\dots,d_{k}}(\mu_{C};\nu_{1},\dots,\nu_{k}) & = W_{2}^{2}(\mu_{C}, \nu_{1} \times \cdots \times \nu_{k}) - W_{2}^{2}(\mu_{C_{1}} \times \cdots \times \mu_{C_{k}}, \nu_{1} \times \cdots \times \nu_{k}) \\
    & = W_{2}^{2}(\mu_{C}, \nu_{1} \times \cdots \times \nu_{k}) - \sum_{i=1}^{k} W_{2}^{2}(\mu_{C_{i}},\nu_{i}),
\end{split}
\end{equation*}
where the second identity is known to be true (see, e.g., \cite{Panaretos2019}). We call a subset $A \subset \mathcal{P}(\mathbb{I}^{q})$ $W_{2}$-compact if every sequence in the metric space $(A,W_{2})$ has a convergent subsequence with limit in $A$. Lemma \ref{lem1} gives the main properties of $T_{d_{1},\dots,d_{k}}$. Proofs of Lemma \ref{lem1}, and all  other theoretical results of Section \ref{sec:2}, can be found in \ref{App B}.
\begin{lemma}\label{lem1}
It holds that 
\begin{itemize}
    \item[(a)] $T_{d_{1},\dots,d_{k}}(\mu_{C};\nu_{1},\dots,\nu_{k}) \geq 0$
    \item[(b)] $T_{d_{1},\dots,d_{k}}(\mu_{C_{1}} \times \cdots \times \mu_{C_{k}};\nu_{1},\dots,\nu_{k}) = 0$
    \item[(c)] If either $\mu_{C_{i}} = \nu_{i}$ for all $i  = 1,\dots,k$ or $\nu_{i}$ is absolutely continuous (with respect to $\lambda^{d_{i}}$) for all $i = 1,\dots,k$, then $T_{d_{1},\dots,d_{k}}(\mu_{C};\nu_{1},\dots,\nu_{k}) = 0$ implies $\mu_{C} = \mu_{C_{1}} \times \cdots \times \mu_{C_{k}}$
    \item[(d)] The set $\Gamma(\mu_{C_{1}},\dots,\mu_{C_{k}})$  is $W_{2}$-compact in $\mathcal{P}(\mathbb{I}^{q})$ and the mapping $T_{d_{1},\dots,d_{k}}(\hspace{0.1cm} \cdot \hspace{0.1cm} ;\nu_{1},\dots,\nu_{k}): (\Gamma(\mu_{C_{1}},\dots,\mu_{C_{k}}), W_{2}) \rightarrow (\mathbb{R},|\cdot|)$ is continuous. 
\end{itemize}
\end{lemma}
\vspace{0.3cm}

Its interpretation, together with its mathematical properties, make $T_{d_{1},\dots,d_{k}}(\mu_{C};\nu_{1},\dots,\nu_{k})$ an appealing measure of dependence between $k$ random vectors $\mathbf{X}_{1},\dots,\mathbf{X}_{k}$. In what follows, we assume that $\nu_{i}$ is absolutely continuous for $i = 1,\dots,k$, and let $G_{C} \subset \Gamma(\mu_{C_{1}},\dots,\mu_{C_{k}})$ be a compact set such that $\mu_{C} \in G_{C}$. General axioms for dependence measures between multiple random vectors are formulated in \cite{Gijbels2023}, see also \ref{App A}. Lemma \ref{lem1} offers aid in proving them here, see Proposition \ref{prop1}.

\begin{proposition}\label{prop1}
Consider Axioms (A1)-(A8) given in \ref{App A}, and a normalized version of $T_{d_{1},\dots,d_{k}}(\mu_{C};\nu_{1}, \dots, \nu_{k})$ given by
\begin{equation}\label{eq: dep opt transp}
    \mathcal{D}(\mu_{C};\nu_{1},\dots,\nu_{k}) = \frac{T_{d_{1},\dots,d_{k}}(\mu_{C};\nu_{1},\dots,\nu_{k})}{\sup_{\pi \in G_{C}} T_{d_{1},\dots,d_{k}}(\pi;\nu_{1},\dots,\nu_{k})}.
\end{equation}
Then, $\mathcal{D}$ satisfies (A1)-(A3) and (A5)-(A7). Axioms (A4) and (A8) are satisfied by the non-normalized version $T_{d_{1},\dots,d_{k}}$. 
\end{proposition} 
\vspace{0.3cm}

The supremum in \eqref{eq: dep opt transp} is attained when $G_{C}$ is $W_{2}$-compact (e.g., $G_{C} = \Gamma(\mu_{C_{1}},\dots,\mu_{C_{k}})$) because of (d) in Lemma \ref{lem1}. It represents the case of maximal dependence, which characterization (and hence the overall behaviour of \eqref{eq: dep opt transp}), largely depends on $G_{C}$. There is still freedom in choosing the normalization by picking the set $G_{C}$. It might impose additional constraints (in addition to having marginals $\mu_{C_{i}}$) on the $\pi \in \Gamma(\mu_{C_{1}},\dots,\mu_{C_{k}})$  that characterizes maximal dependence (for example $\pi$ should be in the same copula family as $\mu_{C}$). Strictly speaking, we can have a different $G_{C}$ for every copula $C$, even if the marginals are the same.

Regarding axioms (A4) and (A8), we make the following remark.

\begin{remark}\label{r1}
While \eqref{eq: dep opt transp} does not satisfy axiom (A4) in general (that is for every possible choice of $G_{C}$), there might still be some specific choices for $G_{C}$ such that (A4) is satisfied. We illustrate this further in Example \ref{ex1}.

Also, when considering $C_{n} \to C$ uniformly as $n \to \infty$, axiom (A8) can be satisfied by \eqref{eq: dep opt transp} under some extra constraints. The numerator converges if $G_{C}$ is chosen such that $C \mapsto \sup_{\pi \in G_{C}} T_{d_{1},\dots,d_{k}}(\pi;\nu_{1},\dots,\nu_{k})$ is continuous (considering the uniform metric on the space of copulas). In Section \ref{sec:3}, we see that this holds in the class of Gaussian copulas when taking $G_{C} = \Gamma(\mu_{C_{1}},\dots,\mu_{C_{k}})$.
\end{remark}
\vspace{0.3cm}

We now arrive at a natural generalization of the Wasserstein dependence coefficients of \cite{Mordant2022}, which come from \eqref{eq: dep opt transp} with two particular choices of reference measures. 
\begin{definition}\label{dfn: wdcoef}
For $m \in \mathbb{Z}_{>0}$, let $\gamma_{m}$ denote the measure of an $m$-variate Gaussian copula with identity correlation matrix (i.e., the $m$-variate independence copula). 
For $q = d_{1} + \cdots + d_{k}$ with $d_{1},\dots,d_{k} \in \mathbb{Z}_{>0}$ and for $\mu_{C} \in G_{C} \subset \Gamma(\mu_{C_{1}},\dots,\mu_{C_{k}})$, where $G_{C}$ is $W_{2}$-compact and non-degenerate $\mu_{C_{i}} \in \mathcal{P}(\mathbb{I}^{d_{i}})$ for $i = 1,\dots,k$, define 
\begin{equation*}
\begin{split}
    \mathcal{D}_{1}^{d_{1},\dots,d_{k}}(\mu_{C}) & = \mathcal{D}(\mu_{C}; \gamma_{d_{1}},\dots,\gamma_{d_{k}}) = \frac{W_{2}^{2}(\mu_{C},\gamma_{q}) - \sum_{i=1}^{k}W_{2}^{2}(\mu_{C_{i}},\gamma_{d_{i}})}{\sup_{\pi \in G_{C}}W_{2}^{2}(\pi,\gamma_{q}) - \sum_{i=1}^{k}W_{2}^{2}(\mu_{C_{i}},\gamma_{d_{i}})} \\
    \mathcal{D}_{2}^{d_{1},\dots,d_{k}}(\mu_{C}) & = \mathcal{D}(\mu_{C}; \mu_{C_{1}},\dots,\mu_{C_{k}}) = \frac{W_{2}^{2}(\mu_{C},\mu_{C_{1}} \times \cdots \times \mu_{C_{k}})}{\sup_{\pi \in G_{C}} W_{2}^{2}(\pi,\mu_{C_{1}} \times \cdots \times \mu_{C_{k}})}.
\end{split}
\end{equation*}
If the context is clear, we just write $\mathcal{D}_{r}(\mu_{C})$ for $r \in \{1,2\}$, or also $\mathcal{D}_{r}(\mathbf{X}_{1};\cdots;\mathbf{X}_{k})$ to emphasize that we are measuring the dependence between $k$ random vectors (having joint copula $C$). 
\end{definition}
\vspace{0.3cm}

Let us consider an example illustrating that $\mathcal{D}$ does not necessarily satisfy Axiom (A4) in general.

\begin{example}\label{ex1}
Consider a random vector $(X_{1},X_{2},X_{3})$ having a trivariate Gaussian copula $\mu_{C}$ with correlation matrix 
\begin{equation*}
    \mathbf{R} = \begin{pmatrix}
    1 & \rho & 0 \\ \rho & 1 & 0 \\ 0 & 0 & 1
    \end{pmatrix}, \hspace{0.4cm} \text{where} \hspace{0.2cm} -1 \leq \rho \leq 1.
\end{equation*}
Let $\mu_{C_{i}}$ be the marginal copula measure of $X_{i}$ for $i = 1,2,3$, which in this case is in fact the Lebesgue measure corresponding to a $\mathcal{U}[0,1]$ distribution.
The product measure $\mu_{C_{1}} \times \mu_{C_{2}} \times \mu_{C_{3}}$ is the three dimensional independence copula, being equal to the trivariate Gaussian copula with identity correlation matrix $\mathbf{I}_{3}$. Also note that $X_{3}$ is independent of $(X_{1},X_{2})$. One can quickly check that (using \eqref{eq: Bures}, see further) 
\begin{equation*}
    \mathcal{D}_{2}(X_{1};X_{2};X_{3}) = \frac{W_{2}^{2}(\mu_{C};\mu_{C_{1}} \times \mu_{C_{2}} \times \mu_{C_{3}})}{\sup_{\pi \in G_{C}}W_{2}^{2}(\pi, \mu_{C_{1}} \times \mu_{C_{2}} \times \mu_{C_{3}})} = \frac{4 - 2\sqrt{1-\rho}-2\sqrt{1+\rho}}{\text{sup}_{\pi \in G_{C}}W_{2}^{2}(\pi,\mu_{C_{1}} \times \mu_{C_{2}} \times \mu_{C_{3}})}.
\end{equation*}
The remaining question is what to pick for $G_{C}$, i.e., which quantity do we put in the denominator and defines the maximal amount of dependence. Well, if $G_{C} = \Gamma(\mu_{C_{1}},\mu_{C_{2}},\mu_{C_{3}})$, we should find the squared $2$-Wasserstein distance between the independence copula and any other trivariate distribution having $\mu_{C_{1}},\mu_{C_{2}}$ and $\mu_{C_{3}}$ as marginals, that is every possible trivariate copula. This is, as far as we are concerned, an open problem. However, in this context, it is reasonable to restrict $G_{C}$ to the Gaussian copula family. Doing so, one has
\begin{equation*}
    \text{sup}_{\pi \in G_{C}}W_{2}^{2}(\pi,\mu_{C_{1}} \times \mu_{C_{2}} \times \mu_{C_{3}}) = W_{2}^{2}(\mu_{C^{\text{co}}}, \mu_{C_{1}} \times \mu_{C_{2}} \times \mu_{C_{3}}) = 6-2\sqrt{3},
\end{equation*}
where $\mu_{C^{\text{co}}}$ stands for the comonotonicity copula measure, i.e., the limit of an equicorrelated ($\mathbf{R} = (1-\rho)\mathbf{I}_{3} + \rho \mathbf{1}_{3} \mathbf{1}_{3}^{\text{T}}$, where $\mathbf{1}_{3}^{\text{T}} = (1,1,1)$ and $\rho \in (-1/2,1)$) Gaussian copula with correlation $\rho$ tending to $1$. With this choice of $G_{C}$, $(X_{1},X_{2},X_{3})$ can never reach the maximum dependence, since $$4-2\sqrt{1-\rho}-2\sqrt{1+\rho} \leq 4 - 2\sqrt{2} < 6 - 2\sqrt{3}.$$
Another way to put it, is that 
\begin{equation*}
    \mathcal{D}_{2}(X_{1};X_{2};X_{3}) = \frac{2-\sqrt{1-\rho}-\sqrt{1+\rho}}{3-\sqrt{3}} < \frac{2-\sqrt{1-\rho}-\sqrt{1+\rho}}{2-\sqrt{2}} = \mathcal{D}_{2}(X_{1};X_{2}),
\end{equation*}
where $\mathcal{D}_{2}(X_{1};X_{2})$ is computed in a similar way, also restricting the couplings $\Gamma(\mu_{C_{1}},\mu_{C_{2}})$ to Gaussian ones. We thus see that, when adding an independent component $X_{3}$ into consideration, the dependence has decreased and hence axiom (A4) is definitely not fulfilled. 
Taking a look at $\mathbf{R}$, it is maybe more tempting to have maximal dependence when $|\rho| = 1$ and restrict $G_{C}$ further to only those $\pi \in \Gamma(\mu_{C_{1}},\mu_{C_{2}},\mu_{C_{3}})$ that are Gaussian and furthermore satisfy $\pi(B_{1} \times B_{2} \times B_{3}) = \pi(B_{1} \times B_{2} \times \mathbb{I}) \cdot \pi(\mathbb{I} \times \mathbb{I} \times B_{3})$ for all $B_{1},B_{2},B_{3} \in \mathcal{B}(\mathbb{I})$. Then, it is quickly seen that 
\begin{equation*}
    \text{sup}_{\pi \in G_{C}}W_{2}^{2}(\pi,\mu_{C_{1}} \times \mu_{C_{2}} \times \mu_{C_{3}}) = 2-\sqrt{2},
\end{equation*}
and hence $\mathcal{D}_{2}(X_{1};X_{2};X_{3}) = \mathcal{D}_{2}(X_{1};X_{2})$, in harmony with axiom (A4). So, for actual computation, it is better to restrict $G_{C}$ to the Gaussian copula family, and if some additional information is given (like zeroes in the correlation matrix), incorporating this in $G_{C}$ can lead to a more interpretative dependence quantification.
\end{example}
\vspace{0.3cm}

Except for some families like normal distributions, computing the Wasserstein distance is very involved and tools and theory for statistical inference are still scarce. The authors of \cite{Mordant2022} give an overview of the literature so far, concluding that additional theory is still needed, and propose a quasi-Gaussian (based on covariance matrices) approach instead. We assume that the copula of $\mathbf{X}$ is Gaussian.

\section{A Gaussian copula approach}\label{sec:3}

In this section, we assume a Gaussian copula model for $\mathbf{X}$, elaborate more on maximal dependence, and discuss statistical inference within this framework.
 
\subsection{The Bures-Wasserstein distance}\label{subsec: 3.1}

The main incentive is the well-known formula for the squared $2$-Wasserstein distance between Gaussian distributions, say with covariance matrices $\mathbf{R}$ and $\mathbf{S}$, becoming the so-called squared Bures-Wasserstein distance (see, e.g., \cite{Takatsu2011}) between $\mathbf{R}$ and $\mathbf{S}$: 
\begin{equation}\label{eq: Bures}
    d_{W}^{2}(\mathbf{R},\mathbf{S}) = \text{tr}(\mathbf{R}) + \text{tr}(\mathbf{S}) - 2 \text{tr}\left \{ \left (\mathbf{R}^{1/2}\mathbf{S}\mathbf{R}^{1/2}\right )^{1/2} \right \},
\end{equation}
were tr stands for the trace of a matrix. We denote by $\mathbb{S}^{q} = \{\mathbf{R} \in \mathbb{R}^{q \times q} : \mathbf{R}^{\text{T}} = \mathbf{R} \}$ the set of symmetric $q \times q$ matrices, $\mathbb{S}_{\geq}^{q} \subset \mathbb{S}^{q}$ the set of positive semi-definite ones and $\mathbb{S}^{q}_{>} \subset \mathbb{S}_{\geq}^{q}$ the set of positive definite ones. Let again $q = d_{1} + \cdots + d_{k}$ and consider $\mathbf{R}_{ii} \in \mathbb{S}^{d_{i}}_{\geq}$ for $i = 1,\dots,k$. We also define the set
\begin{equation*}
\Gamma(\mathbf{R}_{11},\dots,\mathbf{R}_{kk}) = \left\{ \mathbf{A} \in \mathbb{S}^{q}_{\geq} : \mathbf{A} = \begin{pmatrix}
\mathbf{R}_{11} & \boldsymbol{\Psi}_{12} & \cdots & \boldsymbol{\Psi}_{1k} \\
\boldsymbol{\Psi}_{12}^{\text{T}} & \mathbf{R}_{22} & \cdots & \boldsymbol{\Psi}_{2k} \\
\vdots & \vdots & \ddots & \vdots \\
\boldsymbol{\Psi}_{1k}^{\text{T}} & \boldsymbol{\Psi}_{2k}^{\text{T}} & \cdots & \mathbf{R}_{kk}
\end{pmatrix}
\hspace{0.2cm} \text{for some} \hspace{0.2cm} \boldsymbol{\Psi}_{ij} \in \mathbb{R}^{d_{i} \times d_{j}}\right\},
\end{equation*}
as the set of all covariance matrices of random vectors $\mathbf{Z} = (\mathbf{Z}_{1},\dots,\mathbf{Z}_{k})$ such that the covariance matrix of $\mathbf{Z}_{i}$, being $\mathbf{R}_{ii}$, remains fixed for all $i = 1,\dots,k$ and 
\begin{equation}\label{eq: R0}
\mathbf{R}_{0} = \begin{pmatrix}
\mathbf{R}_{11} & \mathbf{0}_{12} & \dots & \mathbf{0}_{1k} \\
\mathbf{0}_{12}^{\text{T}} & \mathbf{R}_{22} & \dots & \mathbf{0}_{2k} \\
\vdots & \vdots & \ddots & \vdots \\
\mathbf{0}_{1k}^{\text{T}} & \mathbf{0}_{2k}^{\text{T}} & \dots & \mathbf{R}_{kk}
\end{pmatrix},
\end{equation}
with $\mathbf{0}_{ij} \in \mathbb{R}^{d_{i} \times d_{j}}$ a matrix of zeroes, as the covariance matrix when the $\mathbf{Z}_{i}$ are all independent. 

Consider now a random vector $\mathbf{X} = (\mathbf{X}_{1},\dots,\mathbf{X}_{k})$ having a Gaussian copula with covariance matrix
\begin{equation}\label{eq: R}
\mathbf{R} = \begin{pmatrix}
\mathbf{R}_{11} & \mathbf{R}_{12} & \dots & \mathbf{R}_{1k} \\
\mathbf{R}_{12}^{\text{T}} & \mathbf{R}_{22} & \dots & \mathbf{R}_{2k} \\
\vdots & \vdots & \ddots & \vdots \\
\mathbf{R}_{1k}^{\text{T}} & \mathbf{R}_{2k}^{\text{T}} & \dots & \mathbf{R}_{kk}
\end{pmatrix}.
\end{equation}
This means that \eqref{eq: R} is the usual covariance matrix of the random vector $\mathbf{Z} = (\mathbf{Z}_{1},\dots,\mathbf{Z}_{k})$, with $\mathbf{Z}_{i} = (Z_{i1},\dots,Z_{id_{i}})$ and $Z_{ij} = (\Phi^{-1} \circ F_{ij})(X_{ij})$ for $i = 1,\dots,k$ and $j = 1,\dots,d_{i}$, where $F_{ij}$ is the marginal distribution of $X_{ij}$ and $\Phi^{-1}$ the univariate standard normal quantile function. 
Measuring the dependence between $\mathbf{X}_{1},\dots,\mathbf{X}_{k}$ can be done by utilizing the $2$-Wasserstein dependence coefficients of Definition \ref{dfn: wdcoef}, now taking $d_{W}$ for $W_{2}$.
\begin{definition}\label{dfn: Gaussian opt meas}
For $q = d_{1} + \cdots + d_{k}$ with $d_{1},\dots,d_{k} \in \mathbb{Z}_{>0}$, $\mathbf{R} \in \Gamma(\mathbf{R}_{11},\dots,\mathbf{R}_{kk})$ with $\mathbf{R}_{ii} \in \mathbb{S}^{d_{i}}_{\geq} \setminus \{\mathbf{0}\}$ for $i = 1,\dots,k$ and a $d_{W}$-compact set $G_{\mathbf{R}} \subset \Gamma(\mathbf{R}_{11},\dots,\mathbf{R}_{kk})$ with $\mathbf{R} \in G_{\mathbf{R}}$, define
\begin{equation*}
\begin{split}
    \mathcal{D}_{1}(\mathbf{R};d_{1},\dots,d_{k}) & = \frac{d_{W}^{2}(\mathbf{R},\mathbf{I}_{q}) - \sum_{i=1}^{k}d_{W}^{2}(\mathbf{R}_{ii},\mathbf{I}_{d_{i}})}{\text{sup}_{\mathbf{A} \in G_{\mathbf{R}}}d_{W}^{2}(\mathbf{A},\mathbf{I}_{q}) - \sum_{i=1}^{k}d_{W}^{2}(\mathbf{R}_{ii},\mathbf{I}_{d_{i}})} \\
    \mathcal{D}_{2}(\mathbf{R};d_{1},\dots,d_{k}) & = \frac{d_{W}^{2}(\mathbf{R},\mathbf{R}_{0})}{\sup_{\mathbf{A} \in G_{\mathbf{R}}}d_{W}^{2}(\mathbf{A},\mathbf{R}_{0})},
\end{split}
\end{equation*}
where $\mathbf{R}_{0}$ is the matrix given in \eqref{eq: R0}. If the context is clear, we also write $\mathcal{D}_{r}(\mathbf{R})$ or $\mathcal{D}_{r}(\mathbf{X}_{1};\cdots;\mathbf{X}_{k})$ for $r \in \{1,2\}$.
\end{definition}
\vspace{0.3cm}

If the true copula is indeed Gaussian, the adequacy of the Bures-Wasserstein dependence remains, and we obtain something way more easy to handle. In order to make them fully practically usable, that is to say suitably attractive for estimation, we ought to find explicit expressions for the suprema in the denominator of the dependence measures. When $G_{\mathbf{R}} = \Gamma(\mathbf{R}_{11},\dots,\mathbf{R}_{kk})$ and $k = 2$, the authors of \cite{Mordant2022} found elegant solutions to this problem, which we generalize to general $k$. 

\subsection{Maximal Bures-Wasserstein dependence}\label{subsec: 3.2}

We need the definition of majorization of two vectors and its behaviour under convex functions, as studied in \cite{Marshall2011}. 
\begin{definition}\label{dfn: maj of eig}
For two vectors $\mathbf{x},\mathbf{y} \in \mathbb{R}^{q}$, we say that $\mathbf{y}$ majorizes $\mathbf{x}$, denoted as $\mathbf{x} \prec \mathbf{y}$, if
\begin{equation*}
\begin{cases} \sum_{i=1}^{\ell} x_{[i]} \leq \sum_{i=1}^{\ell} y_{[i]} & \mbox{for } \ell = 1,\dots, q-1 \\ \sum_{i=1}^{q} x_{[i]} = \sum_{i=1}^{q} y_{[i]} \end{cases},
\end{equation*}
where $x_{[1]} \geq \dots \geq x_{[q]}$ are the components of $\mathbf{x}$ in decreasing order, and similarly for $\mathbf{y}$.
\end{definition}
\vspace{0.3cm}

When $\boldsymbol{\lambda}$ and $\boldsymbol{\mu}$ are the vectors of eigenvalues of two correlation matrices, $\boldsymbol{\lambda}$ being majorized by $\boldsymbol{\mu}$ means that the proportion of the total variance explained by the $\ell$ first principal components is larger for the correlation matrix with eigenvalues $\boldsymbol{\mu}$, for any $\ell \in \{1, \dots, q-1\}$, than for the correlation matrix with eigenvalues $\boldsymbol{\lambda}$. Fixing $k$ covariance matrices $\mathbf{R}_{11},\dots,\mathbf{R}_{kk}$ with $\mathbf{R}_{ii} \in \mathbb{S}^{d_{i}}_{\geq}$, the goal is now to find $\mathbf{R}_{m} \in \Gamma(\mathbf{R}_{11},\dots,\mathbf{R}_{kk})$ whose ordered eigenvalues majorize those of any matrix  $\mathbf{A} \in \Gamma(\mathbf{R}_{11},\dots,\mathbf{R}_{kk})$. Together with Lemma \ref{lem2} (see Proposition 3.C.1 in \cite{Marshall2011}), this will enable us to characterize maximal dependence between $k$ random vectors in terms of covariance matrices.
\begin{lemma}\label{lem2}
If $g : I \rightarrow \mathbb{R}$ is convex, with $I \subset \mathbb{R}$ an interval, then 
\begin{equation*}
    \mathbf{x} \prec \mathbf{y} \hspace{0.2cm} \implies \hspace{0.2cm} \sum_{i=1}^{q} g(x_{i}) \leq \sum_{i=1}^{q} g(y_{i})
\end{equation*}
for all $\mathbf{x},\mathbf{y} \in I^{q}$.
\end{lemma}
\vspace{0.3cm}

Since the $2$-Wasserstein dependence coefficients satisfy axiom (A1), we can assume that $d_{1} \leq d_{2} \leq \dots \leq d_{k}$ without loss of generality. Suppose that
\begin{equation}\label{eq: eig decomp}
\mathbf{R}_{ii} = \mathbf{U}_{ii}\boldsymbol{\Lambda}_{ii} \mathbf{U}_{ii}^{\text{T}}
\end{equation}
is the eigendecomposition of $\mathbf{R}_{ii}$, i.e., with $\boldsymbol{\Lambda}_{ii} = \text{diag}(\lambda_{1,ii},\lambda_{2,ii},\dots,\lambda_{d_{i},ii})$ the $d_{i} \times d_{i}$ diagonal matrix with $d_{i}$ ordered eigenvalues $\lambda_{1,ii} \geq \lambda_{2,ii} \geq \dots \geq \lambda_{d_{i},ii}$ on the diagonal (counting multiplicities), and $\mathbf{U}_{ii}$ an orthogonal matrix containing the corresponding eigenvectors for $i = 1,\dots,k$. The proof of Proposition 2, and all other theoretical results of Section \ref{sec:3}, are provided in \ref{App C}.
\begin{proposition}\label{prop2}
Let $\mathbf{R}_{ii} \in \mathbb{S}^{d_{i}}_{\geq}$ have eigendecomposition \eqref{eq: eig decomp} for $i = 1,\dots,k$. Define the matrix $\mathbf{R}_{m}$ as 
\begin{equation}\label{eq: Rmhelp}
\mathbf{R}_{m} = 
\begin{pmatrix}
\mathbf{R}_{11} & \boldsymbol{\Psi}_{12} & \cdots & \boldsymbol{\Psi}_{1k} \\
\boldsymbol{\Psi}_{12}^{\text{T}} & \mathbf{R}_{22} & \cdots & \boldsymbol{\Psi}_{2k} \\
\vdots & \vdots & \ddots & \vdots \\
\boldsymbol{\Psi}_{1k}^{\text{T}} & \boldsymbol{\Psi}_{2k}^{\text{T}} & \cdots & \mathbf{R}_{kk}
\end{pmatrix} \in \mathbb{R}^{q \times q},
\end{equation}
with $q = d_{1} + \cdots + d_{k}$ and $d_{1} \leq d_{2} \leq \dots \leq d_{k}$, and off-diagonal blocks $$\boldsymbol{\Psi}_{ij} = \mathbf{U}_{ii}\boldsymbol{\Lambda}_{ii}^{1/2}\boldsymbol{\Pi}_{ij} \boldsymbol{\Lambda}_{jj}^{1/2}\mathbf{U}_{jj}^{\text{T}} \in \mathbb{R}^{d_{i} \times d_{j}},$$ where
\begin{equation*}
    \boldsymbol{\Pi}_{ij} = \left (\hspace{0.1cm}\mathbf{I}_{d_{i}} \hspace{0.2cm} \mathbf{0}_{d_{i} \times (d_{j} - d_{i})} \hspace{0.1cm} \right )\in \mathbb{R}^{d_{i} \times d_{j}},
\end{equation*}
the $d_{i} \times d_{j}$ upper left block of $\hspace{0.02cm}$ $\mathbf{I}_{d_{i}+d_{j}}$ for $i = 1,\dots,k, j = i+1,\dots,k$ (denoting $\mathbf{0}_{d_{i} \times (d_{j} - d_{i})}$ for the $d_{i} \times (d_{j} - d_{i})$ matrix of zeroes). If we define $\lambda_{j,ii} = 0$ for $j = d_{i} + 1, \dots, q$, the eigenvalues of $\mathbf{R}_{m}$ are
\begin{equation}\label{eq: eigRm}
    \boldsymbol{\lambda}(\mathbf{R}_{m}) = (\lambda_{j,11} + \lambda_{j,22} + \dots + \lambda_{j,kk})_{j=1}^{q}.
\end{equation} Furthermore, for any $\mathbf{A} \in \Gamma(\mathbf{R}_{11},\dots,\mathbf{R}_{kk})$ with eigenvalues $\boldsymbol{\lambda}(\mathbf{A}) = (\lambda_{j})_{j=1}^{q}$, it holds that
\begin{equation*}
    \boldsymbol{\lambda}(\mathbf{A}) \prec \boldsymbol{\lambda}(\mathbf{R}_{m}).
\end{equation*}
\end{proposition}
\vspace{0.3cm}

Example \ref{ex2} gives the matrix $\mathbf{R}_{m}$ for some specific cases.
\begin{example}\label{ex2}
Some expressions for $\mathbf{R}_{m}$ in case $k = 2$ can be found in \cite{Mordant2022}. If $d_{i} = 1$ with $\mathbf{R}_{ii} = 1$ for all $i = 1,\dots,k$, the matrix $\mathbf{R}_{m}$ is obviously given by $\mathbf{1}_{q \times q}$, a matrix full of ones. Consider next $\mathbf{Z}_{i} = (Z_{i1},Z_{i2})$ for $i = 1,\dots,k$, i.e., $k$ bivariate random vectors, with covariance matrix of $(Z_{i1},Z_{i2})$ given by 
\begin{equation*}
    \mathbf{R}_{ii} = \begin{pmatrix}
    1 & \rho_{i} \\
    \rho_{i} & 1
    \end{pmatrix}.
\end{equation*}
Assuming $\rho_{i},\rho_{j} \geq 0$, one can check that $\boldsymbol{\Psi}_{ij}$ of $\mathbf{R}_{m}$ in \eqref{eq: Rmhelp} for $i = 1,\dots,k-1$ and $j = i+1,\dots,k$ is given by
\begin{equation}\label{eq: psiij} 
    \boldsymbol{\Psi}_{ij} = \begin{pmatrix}
     \frac{\sqrt{1+\rho_{i}}\sqrt{1+\rho_{j}}+\sqrt{1-\rho_{i}}\sqrt{1-\rho_{j}}}{2} & \frac{\sqrt{1+\rho_{i}}\sqrt{1+\rho_{j}}-\sqrt{1-\rho_{i}}\sqrt{1-\rho_{j}}}{2} \vspace{0.1cm} \\ 
    \frac{\sqrt{1+\rho_{i}}\sqrt{1+\rho_{j}}- \sqrt{1-\rho_{i}}\sqrt{1-\rho_{j}}}{2} & \frac{\sqrt{1+\rho_{i}}\sqrt{1+\rho_{j}}+ \sqrt{1-\rho_{i}}\sqrt{1-\rho_{j}}}{2}
    \end{pmatrix}.
\end{equation} \normalsize 
The result is similar in case $\rho_{i} \leq 0$ or $\rho_{j} \leq 0$ up to some signs (orthogonal transformations, to which $d_{W}$ is invariant). The principal components of $(Z_{i1},Z_{i2})$ are
\begin{equation*}
    Y_{i1} = \frac{1}{\sqrt{2}}(Z_{i1} + Z_{i2}) \hspace{0.2cm} \text{and} \hspace{0.2cm} Y_{i2} = \frac{1}{\sqrt{2}}(Z_{i2}-Z_{i1}),
\end{equation*}
corresponding to the eigenvalues $1+\rho_{i} \geq 1-\rho_{i}$ respectively, with
\begin{equation*}
    \text{Corr}(Y_{i1},Y_{j1}) = \frac{\text{Corr}(Z_{i2},Z_{j2})+\text{Corr}(Z_{i2},Z_{j1})+\text{Corr}(Z_{i1},Z_{j2}) + \text{Corr}(Z_{i1},Z_{j1})}{2\sqrt{1+\text{Corr}(Z_{i1},Z_{i2})}\sqrt{1+\text{Corr}(Z_{j1},Z_{j2})}}
\end{equation*}
and 
\begin{equation*}
    \text{Corr}(Y_{i2},Y_{j2}) = \frac{\text{Corr}(Z_{i2},Z_{j2})-\text{Corr}(Z_{i2},Z_{j1})-\text{Corr}(Z_{i1},Z_{j2}) + \text{Corr}(Z_{i1},Z_{j1})}{2\sqrt{1-\text{Corr}(Z_{i1},Z_{i2})}\sqrt{1-\text{Corr}(Z_{j1},Z_{j2})}}.
\end{equation*}
A quick check then verifies that if $ ( (Z_{11},Z_{12}), \dots, (Z_{k1},Z_{k2})  )$ has correlation matrix $\mathbf{R}_{m}$ with blocks \eqref{eq: psiij}, it holds that $\text{Corr}(Y_{i1},Y_{j1}) = \text{Corr}(Y_{i2},Y_{j2}) = 1$ for all $i = 1,\dots,k-1$ and $j = i+1, \dots, k$, i.e., all first principal components are perfectly correlated and all second principal components as well. 
\end{example}
\vspace{0.3cm}

Proposition \ref{prop3} states that the matrix $\mathbf{R}_{m}$ in \eqref{eq: Rmhelp} maximizes the intensity of dependence, i.e., $\mathcal{D}_{1}(\mathbf{R}_{m}) = \mathcal{D}_{2}(\mathbf{R}_{m}) = 1$ when taking $G_{\mathbf{R}} = \Gamma(\mathbf{R}_{11},\dots,\mathbf{R}_{kk})$ for fixed marginal covariance matrices $\mathbf{R}_{11},\dots,\mathbf{R}_{kk}$. 

\begin{proposition}\label{prop3} 
Let $q = d_{1} + \cdots + d_{k}$ and $\mathbf{R}_{ii} \in \mathbb{S}_{\geq}^{d_{i}}$ for $i = 1,\dots,k$. The matrix $\mathbf{R}_{m}$ in \eqref{eq: Rmhelp} maximizes $d_{W}(\mathbf{R},\mathbf{I}_{q})$ and $d_{W}(\mathbf{R},\mathbf{R}_{0})$ with $\mathbf{R}_{0}$ given in \eqref{eq: R0} among all $\mathbf{R} \in \Gamma(\mathbf{R}_{11},\dots,\mathbf{R}_{kk})$.
\end{proposition}
\vspace{0.3cm}

A general interpretation of $\mathbf{R}_{m}$ is given in Remark \ref{r2}.
.
\begin{remark}\label{r2}
Assuming again that $d_{1} \leq d_{2} \leq \cdots \leq d_{k}$, the matrix $\mathbf{R}_{m}$ in \eqref{eq: Rmhelp} is the covariance matrix of 
\begin{equation*}
\begin{pmatrix} \vspace{0.1cm}
\mathbf{U}_{11} \boldsymbol{\Lambda}_{11}^{1/2} \mathbf{Z}_{1} \\
\mathbf{U}_{22} \boldsymbol{\Lambda}_{22}^{1/2} \mathbf{Z}_{2} \\
\vdots \\
\mathbf{U}_{kk} \boldsymbol{\Lambda}_{kk}^{1/2} \mathbf{Z}_{k} 
\end{pmatrix} \in \mathbb{R}^{q \times 1},
\end{equation*}
where $\mathbf{Z}_{1} = (Z_{11},\dots,Z_{1d_{1}})^{\text{T}}, \mathbf{Z}_{2} = (Z_{21},\dots,Z_{2d_{2}})^{\text{T}}, \dots, \mathbf{Z}_{k} = (Z_{k1},\dots,Z_{kd_{k}})^{\text{T}}$ such that for all $i = 1, \dots, k$ we have $\mathbf{Z}_{i} \sim \mathcal{N}_{d_{i}}(\mathbf{0}_{d_{i}},\mathbf{I}_{d_{i}})$, and in addition $Z_{ij} = Z_{(i+1) j}$ for all $i = 1, \dots, k-1$, $j = 1,\dots, d_{i}$, i.e., $\mathbf{Z}_{i}$ and $\mathbf{Z}_{j}$ have the first $\min\{d_{i},d_{j}\}$ components in common for all $i,j \in \{1,\dots,k\}$. The correlation matrix of the principal components of $\mathbf{U}_{11}\boldsymbol{\Lambda}_{11}^{1/2}\mathbf{Z}_{1},\dots,\mathbf{U}_{kk}\boldsymbol{\Lambda}_{kk}^{1/2}\mathbf{Z}_{k}$ is
\begin{equation*}
\begin{pmatrix}
\mathbf{I}_{d_{1}} & \boldsymbol{\Pi}_{12} & \dots & \boldsymbol{\Pi}_{1k} \vspace{0.1cm} \\\boldsymbol{\Pi}_{12}^{\text{T}} & \mathbf{I}_{d_{2}} & \cdots & \boldsymbol{\Pi}_{2k}\\
\vdots & \vdots & \ddots & \vdots \\
\boldsymbol{\Pi}_{1k}^{\text{T}} & \boldsymbol{\Pi}_{2k}^{\text{T}} & \dots & \mathbf{I}_{d_{k}} \end{pmatrix},
\end{equation*}
with $\boldsymbol{\Pi}_{ij}$ as in Proposition \ref{prop2}. Hence, if $(\mathbf{X}_{1},\dots,\mathbf{X}_{k})$ has a Gaussian copula with covariance matrix $\mathbf{R}_{m}$, then for $\ell = 1,\dots,\min\{d_{i},d_{j}\}$ the $\ell$-th principal components of $((\Phi^{-1} \circ F_{i1})(X_{i1}), \dots, (\Phi^{-1} \circ F_{id_{i}})(X_{id_{i}}))$ and $((\Phi^{-1} \circ F_{j1})(X_{j1}), \dots,$ $(\Phi^{-1} \circ F_{jd_{j}})(X_{jd_{j}}))$ are perfectly correlated for all $i,j \in \{1,\dots,k\}$. This is the interpretation of maximal dependence for the Bures-Wasserstein dependence measures. 
\end{remark}
\vspace{0.3cm}

In the upcoming section, we also assume that $G_{\mathbf{R}} = \Gamma(\mathbf{R}_{11},\dots,\mathbf{R}_{kk})$.

\subsection{Statistical inference}\label{subsec: 3.3}

In practice, we have an i.i.d. sample $\mathbf{X}^{(\ell)} = (\mathbf{X}_{1}^{(\ell)},\dots,\mathbf{X}_{k}^{(\ell)})$ for $\ell = 1,\dots,n $ from $\mathbf{X}$, where $\mathbf{X}_{i}^{(\ell)} = (X_{i1}^{(\ell)},\dots,X_{id_{i}}^{(\ell)})$ for $\ell = 1, \dots, n$ is a sample from $\mathbf{X}_{i}$ for $i = 1,\dots,k$. A natural estimator for the Gaussian copula covariance matrix is known as the matrix of sample normal scores rank correlation coefficients (see, e.g., \cite{Hajek1967}), 
\begin{equation}\label{eq: est cor matrix}
    \widehat{\mathbf{R}}_{n} = \begin{pmatrix}
\widehat{\mathbf{R}}_{11} & \widehat{\mathbf{R}}_{12} & \cdots & \widehat{\mathbf{R}}_{1k} \\
\widehat{\mathbf{R}}_{12}^{\text{T}} & \widehat{\mathbf{R}}_{22} & \cdots & \widehat{\mathbf{R}}_{2k} \\
\vdots & \vdots & \ddots & \vdots \\
\widehat{\mathbf{R}}_{1k}^{\text{T}} & \widehat{\mathbf{R}}_{2k}^{\text{T}} & \cdots & \widehat{\mathbf{R}}_{kk}
\end{pmatrix} \hspace{0.2cm} \text{with} \hspace{0.2cm} (\widehat{\mathbf{R}}_{im})_{jt} =  \widehat{\rho}_{ij,mt} =  \frac{\frac{1}{n}\sum_{\ell=1}^{n} \widehat{Z}_{ij}^{(\ell)} \widehat{Z}_{mt}^{(\ell)}}{\frac{1}{n}\sum_{\ell=1}^{n}\left \{ \Phi^{-1} \left (\frac{\ell}{n+1} \right ) \right \}^{2}},
\end{equation}
defined by computing normal scores 
\begin{equation*}
    \widehat{Z}_{ij}^{(\ell)} = \Phi^{-1} \left (\frac{n}{n+1} \widehat{F}_{ij}\left (X_{ij}^{(\ell)} \right )\right )
\end{equation*}
obtained through the empirical cdf $\widehat{F}_{ij}(x_{ij}) = \frac{1}{n} \sum_{\ell=1}^{n} \mathbbm{1} \{X_{ij}^{(\ell)} \leq x_{ij}\}$ for $i = 1,\dots,k$ and $j = 1,\dots,d_{i}$. 
The quantity $\widehat{\rho}_{ij,mt}$ is calculated as the conventional Pearson correlation of the bivariate scores $((\widehat{Z}_{ij}^{(1)},\widehat{Z}_{mt}^{(1)}),\dots,(\widehat{Z}_{ij}^{(n)},\widehat{Z}_{mt}^{(n)}))$ and by observing that 
\begin{equation*}
    \frac{1}{n} \sum_{\ell=1}^{n} \widehat{Z}_{ij}^{(\ell)}  = \frac{1}{n} \sum_{\ell=1}^{n} \Phi^{-1} \left (\frac{\ell}{n+1}\right ) = 0 \hspace{0.2cm} \text{and} \hspace{0.2cm}
    \frac{1}{n} \sum_{\ell=1}^{n} \left ( \widehat{Z}_{ij}^{(\ell)} \right )^{2}  = \frac{1}{n} \sum_{\ell=1}^{n} \left \{  \Phi^{-1} \left (\frac{\ell}{n+1}\right ) \right \} ^{2},
\end{equation*}
which holds because $\Phi^{-1}(\alpha) = -\Phi^{-1}(1-\alpha)$ for $\alpha \in [0,1]$ and $n \widehat{F}_{ij}(X_{ij}^{(\ell)})$ is the rank of $X_{ij}^{(\ell)}$ in the sample $X_{ij}^{(1)},\dots,X_{ij}^{(n)}$. 

A next natural step in estimating $\mathcal{D}_{r}(\mathbf{R})$ is to plug in $\widehat{\mathbf{R}}_{n}$ for the unknown $\mathbf{R}$.  Define the map $\varphi$ by $\varphi(\mathbf{\Sigma}) = \mathbf{D}_{\mathbf{\Sigma}}^{-1/2}\mathbf{\Sigma}\mathbf{D}_{\mathbf{\Sigma}}^{-1/2}$ for $\mathbf{\Sigma} \in \mathbb{S}^{q}$, where $\mathbf{D}_{\boldsymbol{\Sigma}}$ is the diagonal matrix containing the diagonal of $\boldsymbol{\Sigma}$, and let $||\mathbf{\Sigma}||_{\text{F}} = \text{tr}^{1/2}(\mathbf{\Sigma}^{\text{T}}\mathbf{\Sigma})$ be the Frobenius norm. Fr\'echet differentiability of the mapping 
\begin{equation*}
   (\mathbb{S}^{q},||\cdot||_{\text{F}}) \rightarrow (\mathbb{R},|\cdot|) : \mathbf{\Sigma} \mapsto (\mathcal{D}_{r} \circ \varphi)(\mathbf{\Sigma})
\end{equation*}
on $\mathbb{S}^{q}_{>}$ suffices in order for the delta method to transform an asymptotic normality result for $\widehat{\mathbf{R}}_{n}$ into an asymptotic normality result for $\mathcal{D}_{r}(\widehat{\mathbf{R}}_{n})$. We first highlight some notation. We assume $d_{1} \leq d_{2} \leq \cdots \leq d_{k}$, the matrix $\mathbf{R}_{m}$ is again defined by \eqref{eq: Rmhelp} based on the eigendecompositions $\mathbf{R}_{ii} = \mathbf{U}_{ii}\boldsymbol{\Lambda}_{ii}\mathbf{U}_{ii}^{\text{T}}$ as in \eqref{eq: eig decomp}, and $\mathbf{R}_{0}$ is the matrix in \eqref{eq: R0}. Further, let 
\begin{equation}\label{eq: pp1}
    \mathbf{P}_{i} = \left (\mathbf{0}_{d_{i} \times d_{1}} \cdots \hspace{0.1cm} \mathbf{0}_{d_{i} \times d_{i-1}} \hspace{0.1cm}  \mathbf{I}_{d_{i}} \hspace{0.1cm}  \mathbf{0}_{d_{i} \times d_{i+1}} \cdots \hspace{0.1cm}  \mathbf{0}_{d_{i} \times d_{k}} \right ) \in \mathbb{R}^{d_{i} \times q}
\end{equation}
be the projection matrix onto the $d_{i}$ coordinates, satisfying $\mathbf{R}_{ii} = \mathbf{P}_{i} \mathbf{R} \mathbf{P}_{i}^{\text{T}}$. Partition the matrix $\boldsymbol{\Lambda}_{ii}$ as 
\begin{equation}\label{eq: pp2}
    \boldsymbol{\Lambda}_{ii} = \begin{pmatrix}
    \boldsymbol{\Lambda}_{ii,1} & \mathbf{0}_{d_{1} \times (d_{2}-d_{1})} & \cdots & \mathbf{0}_{d_{1} \times (d_{i}-d_{i-1})} \\
    \mathbf{0}_{d_{1} \times (d_{2}-d_{1})}^{\text{T}} & \boldsymbol{\Lambda}_{ii,2} & \cdots & \mathbf{0}_{(d_{2}-d_{1}) \times (d_{i} - d_{i-1})} \\
    \vdots & \vdots & \ddots & \vdots \\
    \mathbf{0}_{d_{1} \times (d_{i} - d_{i-1})}^{\text{T}} & \mathbf{0}_{(d_{2}-d_{1}) \times (d_{i}-d_{i-1})}^{\text{T}} & \cdots & \boldsymbol{\Lambda}_{ii,i}
    \end{pmatrix} \in \mathbb{R}^{d_{i} \times d_{i}},
\end{equation}
with $\boldsymbol{\Lambda}_{ii,j} \in \mathbb{R}^{(d_{j} - d_{j-1}) \times (d_{j} - d_{j-1})}$ for $j = 1,\dots,i$ and defining $d_{0} = 0$. Based on these partitions, further define 
\begin{equation}\label{eq: pp3}
    \begin{split}
        \boldsymbol{\Delta}_{1} & = \left (\boldsymbol{\Lambda}_{11,1} + \boldsymbol{\Lambda}_{22,1} + \cdots + \boldsymbol{\Lambda}_{kk,1} \right )^{-1/2}  \in \mathbb{R}^{d_{1} \times d_{1}} \\
        \widetilde{\boldsymbol{\Delta}}_{1} & = \left ( \boldsymbol{\Lambda}_{11,1}^{2} + \boldsymbol{\Lambda}_{22,1}^{2} + \cdots + \boldsymbol{\Lambda}_{kk,1}^{2} \right )^{-1/2} \boldsymbol{\Lambda}_{11,1}  \in \mathbb{R}^{d_{1} \times d_{1}} \\
        \boldsymbol{\Delta}_{i} & = \begin{pmatrix}
        \boldsymbol{\Delta}_{i-1} & \mathbf{0}_{d_{i-1} \times (d_{i}-d_{i-1})} \\
        \mathbf{0}_{d_{i-1} \times (d_{i}-d_{i-1})}^{\text{T}} & \left (\boldsymbol{\Lambda}_{ii,i} + \boldsymbol{\Lambda}_{(i+1)(i+1),i} + \cdots + \boldsymbol{\Lambda}_{kk,i} \right )^{-1/2} 
        \end{pmatrix}  \in \mathbb{R}^{d_{i} \times d_{i}} \\
         \widetilde{\boldsymbol{\Delta}}_{i} & = \begin{pmatrix}
        \widetilde{\boldsymbol{\Delta}}_{i-1}\mathbf{D}_{i} & \mathbf{0}_{d_{i-1} \times (d_{i}-d_{i-1})} \\
        \mathbf{0}_{d_{i-1} \times (d_{i}-d_{i-1})}^{\text{T}} & \left (\boldsymbol{\Lambda}_{ii,i}^{2} + \boldsymbol{\Lambda}_{(i+1)(i+1),i}^{2} + \cdots + \boldsymbol{\Lambda}_{kk,i}^{2} \right )^{-1/2}\boldsymbol{\Lambda}_{ii,i}
        \end{pmatrix}  \in \mathbb{R}^{d_{i} \times d_{i}}
    \end{split}
\end{equation}
for $i = 2,\dots,k$ with $\mathbf{D}_{i} = \text{diag} \left (\boldsymbol{\Lambda}_{(i-1)(i-1),1}^{-1}\boldsymbol{\Lambda}_{ii,1},\dots,\boldsymbol{\Lambda}_{(i-1)(i-1),i-1}^{-1}\boldsymbol{\Lambda}_{ii,i-1} \right ) \in \mathbb{R}^{d_{i-1}\times d_{i-1}}$. 
Finally, we will need the matrices
\begin{equation}\label{eq: J}
    \begin{split}
        \mathbf{J} = \mathbf{R}_{0}^{-1/2} \left (\mathbf{R}_{0}^{1/2}\mathbf{R}\mathbf{R}_{0}^{1/2} \right )^{1/2} \mathbf{R}_{0}^{-1/2} = \begin{pmatrix}
        \mathbf{J}_{11} & \mathbf{J}_{12} & \cdots & \mathbf{J}_{1k} \\
        \mathbf{J}_{21} & \mathbf{J}_{22} & \cdots & \mathbf{J}_{2k} \\
        \vdots & \vdots & \ddots & \vdots \\
        \mathbf{J}_{k1} & \mathbf{J}_{k2} & \cdots & \mathbf{J}_{kk}
        \end{pmatrix} \in \mathbb{R}^{q \times q}
    \end{split}
\end{equation}
with $\mathbf{J}_{ij} \in \mathbb{R}^{d_{i} \times d_{j}}$ and 
\begin{equation}\label{eq: J0}
    \mathbf{J}_{0} = \text{diag}(\mathbf{J}_{11},\dots,\mathbf{J}_{kk}) \in \mathbb{R}^{q \times q}.
\end{equation}

Theorem \ref{thm1} formally states an asymptotic normality result for the estimator $\mathcal{D}_{r}(\widehat{\mathbf{R}}_{n})$.

\begin{theorem}\label{thm1}
Let $\mathbf{X}$ have a Gaussian copula with correlation matrix $\mathbf{R} \in \mathbb{S}^{q}_{>}$ with $\mathbf{R} \neq \mathbf{R}_{0}$ such that $\mathbf{R}_{ii} \in \mathbb{S}_{>}^{d_{i}}$ has $d_{i}$ distinct eigenvalues, and let $\widehat{\mathbf{R}}_{n}$ be given by \eqref{eq: est cor matrix} based on which the plug-in estimator $\widehat{\mathcal{D}}_{r,n} = \mathcal{D}_{r}(\widehat{\mathbf{R}}_{n})$ is constructed. Then for $r \in \{1,2\}$, it holds that 
\begin{equation*}
    \sqrt{n} \left (\widehat{\mathcal{D}}_{r,n} - \mathcal{D}_{r}(\mathbf{R}) \right ) \xrightarrow{d} \mathcal{N} \left (\mathbf{0}_{q},\zeta_{r}^{2} \right )
\end{equation*}
as $n \to \infty$, with asymptotic variance
\begin{equation*}
    \zeta_{r}^{2} = 2 \text{tr} \left [ \left \{ \mathbf{R}  \left (\mathbf{M}_{r}-\mathbf{D}_{\mathbf{M}_{r}\mathbf{R}} \right ) \right \}^{2} \right ],
\end{equation*}
where $\mathbf{D}_{\mathbf{M}_{r}\mathbf{R}}$ is the diagonal matrix consisting of the diagonal of $\mathbf{M}_{r}\mathbf{R}$, and 
\begin{equation*}
    \mathbf{M}_{1} = \frac{1}{2C_{1}} \left (-\mathbf{R}^{-1/2} + (1-\mathcal{D}_{1}(\mathbf{R}))\mathbf{R}_{0}^{-1/2} + \mathcal{D}_{1}(\mathbf{R}) \boldsymbol{\Upsilon}_{1} \right ), \hspace{0.4cm}   \mathbf{M}_{2} = \frac{1}{C_{2}} \left (-\frac{1}{2} \left (\mathbf{J}_{0} + \mathbf{J}^{-1} \right ) + \left ( 1-\mathcal{D}_{2}(\mathbf{R}) \right ) \mathbf{I}_{q} +   \mathcal{D}_{2}(\mathbf{R}) \boldsymbol{\Upsilon}_{2} \right ), 
\end{equation*}
with
\begin{equation}\label{eq: ups1}
\begin{split}
     C_{1} = \sum_{i=1}^{k} \text{tr} \left (\mathbf{R}_{ii}^{1/2} \right ) - \text{tr} \left (\mathbf{R}_{m}^{1/2} \right ), \hspace{1cm} \boldsymbol{\Upsilon}_{1} = \begin{pmatrix}
     \mathbf{U}_{11} \boldsymbol{\Delta}_{1} \mathbf{U}_{11}^{\text{T}} & \mathbf{0}_{d_{1} \times d_{2}} & \cdots & \mathbf{0}_{d_{1} \times d_{k}} \\
     \mathbf{0}_{d_{1} \times d_{2}}^{\text{T}} & \mathbf{U}_{22} \boldsymbol{\Delta}_{2} \mathbf{U}_{22}^{\text{T}} & \cdots & \mathbf{0}_{d_{2} \times d_{k}} \\
     \vdots & \vdots & \ddots & \vdots \\
     \mathbf{0}_{d_{1} \times d_{k}}^{\text{T}} & \mathbf{0}_{d_{2} \times d_{k}}^{\text{T}} & \cdots & \mathbf{U}_{kk} \boldsymbol{\Delta}_{k} \mathbf{U}_{kk}^{\text{T}}
     \end{pmatrix}
\end{split}
\end{equation}
and 
\begin{equation}\label{eq: ups2}
    \begin{split}
     \hspace{-0.6cm} C_{2}  = \text{tr}(\mathbf{R}) - \text{tr} \left \{ \left (\mathbf{R}_{0}^{1/2}\mathbf{R}_{m}\mathbf{R}_{0}^{1/2} \right )^{1/2} \right \} , \hspace{1cm} \boldsymbol{\Upsilon}_{2} = \begin{pmatrix}
     \mathbf{U}_{11} \widetilde{\boldsymbol{\Delta}}_{1} \mathbf{U}_{11}^{\text{T}} & \mathbf{0}_{d_{1} \times d_{2}} & \cdots & \mathbf{0}_{d_{1} \times d_{k}} \\
     \mathbf{0}_{d_{1} \times d_{2}}^{\text{T}} & \mathbf{U}_{22} \widetilde{\boldsymbol{\Delta}}_{2} \mathbf{U}_{22}^{\text{T}} & \cdots & \mathbf{0}_{d_{2} \times d_{k}} \\
     \vdots & \vdots & \ddots & \vdots \\
     \mathbf{0}_{d_{1} \times d_{k}}^{\text{T}} & \mathbf{0}_{d_{2} \times d_{k}}^{\text{T}} & \cdots & \mathbf{U}_{kk} \widetilde{\boldsymbol{\Delta}}_{k} \mathbf{U}_{kk}^{\text{T}}
     \end{pmatrix}.
    \end{split}
\end{equation}
\end{theorem}
\vspace{0.3cm}

\begin{remark}\label{r3}
When $\mathbf{R} = \mathbf{R}_{0}$, it holds that $\zeta_{r} = 0$ and $\sqrt{n} \widehat{\mathcal{D}}_{r,n} \xrightarrow{p} 0$ for $n \to \infty$. The higher-order delta method can however still provide a weak convergence result for $n \widehat{\mathcal{D}}_{r,n}$. A detailed study of this is research in progress. 
\end{remark}
\vspace{0.3cm}

We look at another example, which is also studied in \cite{Gijbels2023b}, but in the context of $\Phi$-dependence measures.

\begin{example}\label{ex3}
Consider a four dimensional random vector $(X_{1},X_{2},X_{3},X_{4})$ having a Gaussian copula with covariance matrix 
\begin{equation}\label{eq: settingOT}
\begin{pmatrix}
    1 & \rho_{1} & \rho_{2} & \rho_{2} \\
    \rho_{1} & 1 & \rho_{2} & \rho_{2} \\
    \rho_{2} & \rho_{2} & 1 & \rho_{1} \\
    \rho_{2} & \rho_{2} & \rho_{1} & 1
\end{pmatrix}, \hspace{0.3cm} \text{where} \hspace{0.3cm} \rho_{1} \geq 2 |\rho_{2}| - 1.
\end{equation}
Then one can check that (recall also Example \ref{ex2} for finding the matrix $\mathbf{R}_{m}$)
\begin{figure}[h!] 
\includegraphics[scale = 0.8]{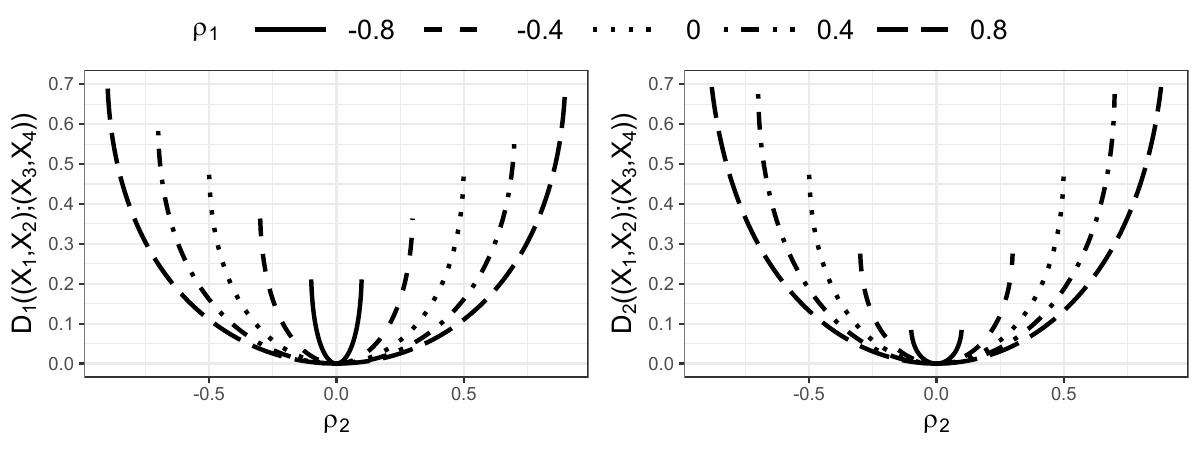}
\caption{Dependence coefficients \eqref{eq: D1Rex} (left) and \eqref{eq: D2Rex} (right) as a function of $\rho_{2}$ for different values of $\rho_{1}$.}
\label{fig: xmplN}
\end{figure} 

\begin{equation}\label{eq: D1Rex}
\begin{split}
  \mathcal{D}_{1} \left ((X_{1},X_{2}) ; (X_{3},X_{4}) \right ) = \frac{2 (1+\rho_{1})^{1/2} - (\rho_{1} - 2 \rho_{2} + 1)^{1/2} - (\rho_{1} + 2 \rho_{2} + 1)^{1/2}}{\left (2 - \sqrt{2} \right ) \left \{(1+\rho_{1})^{1/2} + (1-\rho_{1})^{1/2} \right \}}
\end{split}
\end{equation}
and 
\begin{equation}\label{eq: D2Rex}
\mathcal{D}_{2} \left ((X_{1},X_{2}) ; (X_{3},X_{4}) \right ) = \frac{4-2\left |1-\rho_{1} \right | - \left (\rho_{1}^{2} + 2 \rho_{1} - 2 \rho_{1} \rho_{2} - 2 \rho_{2} + 1 \right )^{1/2} - \left (\rho_{1}^{2} + 2 \rho_{1} + 2 \rho_{1} \rho_{2} + 2 \rho_{2} + 1 \right )^{1/2}}{4 - \sqrt{2} \left \{ \left |1-\rho_{1} \right | + \rho_{1} + 1 \right \}}.
\end{equation}
Fig. \ref{fig: xmplN} shows how \eqref{eq: D1Rex} and \eqref{eq: D2Rex} depend on $\rho_{2}$ for different values of $\rho_{1}$. Clearly, independence holds if and only if $\rho_{2} = 0$, as it should. When 
$\rho_{1} = 2|\rho_{2}|-1$, the second principal component of $(Z_{1},Z_{2}) = ((\Phi^{-1} \circ F_{1})(X_{1}),(\Phi^{-1} \circ F_{2})(X_{2}))$ is perfectly correlated with the second principal component of $(Z_{3},Z_{4}) = ((\Phi^{-1} \circ F_{3})(X_{3}),(\Phi^{-1} \circ F_{4})(X_{4}))$, where $F_{i}$ is the marginal distribution of $X_{i}$ for $i = 1,2,3,4$, see also \cite{Gijbels2023b}. This causes the $\Phi$-dependence measures studied in \cite{Gijbels2023b} to reach their maximum value, because a singularity is attained. In particular, taking for instance $\rho_{1} = -0.4$ and $\rho_{2} = 0.3$, all (normalized) $\Phi$-dependence measures equal $1$, while $\mathcal{D}_{1} = 0.396$ and $\mathcal{D}_{2} = 0.3$. 
\begin{figure}[h!] 
\includegraphics[scale = 0.8]{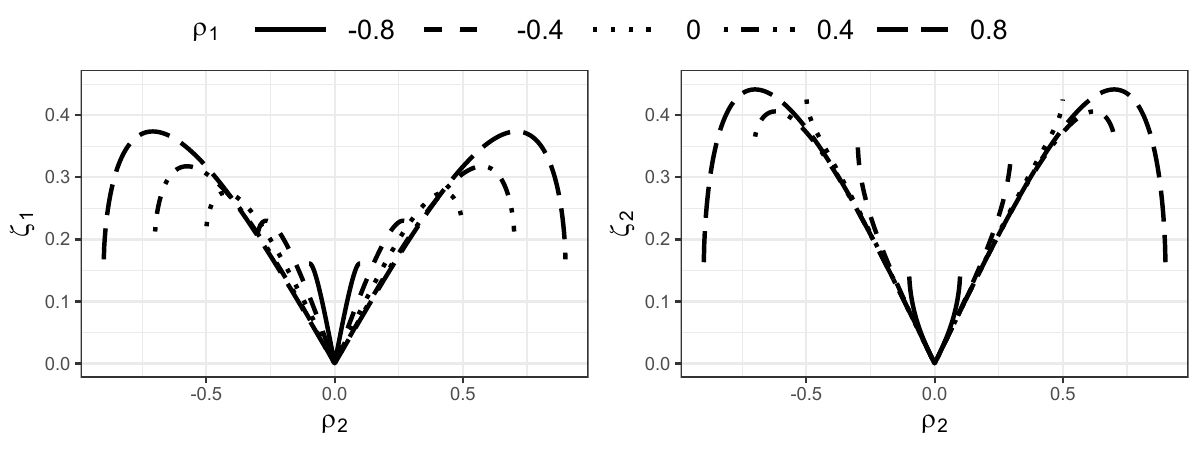}
\caption{Asymptotic standard deviation of dependence coefficients \eqref{eq: D1Rex} (left) and \eqref{eq: D2Rex} (right) as a function of $\rho_{2}$ for different values of $\rho_{1}$.}
\label{fig: xmplN2}
\end{figure} 
The reason for $\mathcal{D}_{1}$ and $\mathcal{D}_{2}$ still being small, is that $\rho_{2}$ is pretty small and, recalling Remark \ref{r2}, not both first and second principal components of $(Z_{1},Z_{2})$ and $(Z_{3},Z_{4})$ are perfectly correlated, only one of them is. Only when $|\rho_{2}| = 1$ and thus also $\rho_{1} = 1$, we have $\mathcal{D}_{1} = \mathcal{D}_{2} = 1$. Maximal Bures-Wasserstein dependence is not attainable for the family \eqref{eq: settingOT} if $\rho_{1} \neq 1$, because it imposes additional restrictions on the correlations (i.e., not every element in $\Gamma(\mathbf{R}_{11},\dots,\mathbf{R}_{kk})$ is a member of \eqref{eq: settingOT}). Picking a set $G_{\mathbf{R}} \subset \Gamma(\mathbf{R}_{11},\dots,\mathbf{R}_{kk})$ for a normalization that adjusts according to these restrictions, would lead to more cases of maximal dependence. 

Fig. \ref{fig: xmplN2} depicts the asymptotic standard deviation $\zeta_{r}$ of $\mathcal{D}_{r}$ for $r \in \{1,2\}$ as given in Theorem \ref{thm1} for this specific example. We mainly observe increasing behaviour when the strength of dependence increases. However, in some cases where $\rho_{1}$ and $\rho_{2}$ get close to satisfying $|\rho_{2}| = (\rho_{1} + 1)/2$, we see the asymptotic standard deviation going down. For example, if $\rho_{1} = 0$, the asymptotic standard deviation $\zeta_{1}$ is maximal ($\approx 0.275$) at $|\rho_{2}| \approx 0.426$, after which it converges to $\approx 0.214$ for $|\rho_{2}| \to 0.5$. When keeping $\rho_{1}$ fixed and letting $|\rho_{2}| \to (\rho_{1} + 1)/2$, the dependence $\mathcal{D}_{r}$ attains a local maximum, which it cannot transcend, resulting in (slightly) lower asymptotic variance. This behaviour was also noticed in the same example in \cite{Gijbels2023b} for their (normalized) $\Phi$-dependence measures, whose asymptotic variance tends to zero when $|\rho_{2}| \to (\rho_{1} + 1)/2$, because the singularity forces them to reach their global maximum. For the optimal transport dependence coefficients, the global maximum is only attained when $|\rho_{2}|, \rho_{1} \to 1$, in which case the asymptotic variance also tends to zero. 

\end{example}

\section{A Gaussian copula approach: regularized estimation}\label{sec: 4}

Empirical covariance matrices are known to approach singularity when the dimension is close to the sample size. The estimator $\widehat{\mathcal{D}}_{r,n}$ does not require an inverted covariance matrix, but it inquires about eigenvalue dispersion, and this tends to be biased when using the empirical covariance matrix, see, e.g., \cite{Ledoit2004}. Typically, estimates of large eigenvalues tend to be biased upwards, and estimates of small eigenvalues tend to be biased downwards. Increasing dimensionality aggravates this, but penalization techniques can be used to restrain. We briefly discuss ridge regularization, as in \cite{Warton2008}, but now in a Gaussian copula context.

Of course, the ridge estimator will not completely shrink elements of the empirical covariance matrix to zero. However, the task might be to find a likely estimate for which multiple variables are marginally independent, i.e., a sparse estimate of the covariance matrix having zero entries. For this purpose, we look at a Gaussian copula formulation of the penalization ideas discussed in \cite{Lam2009}. Finally, penalties can also be applied to groups of elements instead of just individual elements. A group lasso penalty, for instance, enables to shrink entire blocks of the covariance matrix to zero. 

In Section \ref{subsec: 4.1}, we discuss penalization methods for the Gaussian copula covariance matrix in case $q$ remains fixed with the sample size. Afterwards, in Section \ref{subsec: 4.2}, we briefly touch upon the case where $q$ depends on $n$. 

\subsection{Fixed dimension}\label{subsec: 4.1}

Denote $\widehat{\boldsymbol{\Sigma}}_{\text{ML}}$ for the maximum likelihood estimator of a $\mathcal{N}_{q}(\mathbf{0}_{q},\boldsymbol{\Sigma})$ model, and $\mathbf{R} = \varphi(\boldsymbol{\Sigma}) = \mathbf{D}_{\boldsymbol{\Sigma}}^{-1/2} \boldsymbol{\Sigma} \mathbf{D}_{\boldsymbol{\Sigma}}^{-1/2}$ the correlation matrix corresponding to $\boldsymbol{\Sigma}$ (recall that $\mathbf{D}_{\boldsymbol{\Sigma}}$ is the diagonal matrix containing the diagonal of $\boldsymbol{\Sigma}$).  By adding a certain penalty, say $P_{\bigcdot}(\boldsymbol{\Sigma},\omega_{n})$ to the Gaussian log-likelihood, where $\omega_{n}$ is a certain penalty parameter depending on $n$, a general penalized optimization problem for estimating $\boldsymbol{\Sigma}$ (under the constraint of positive definiteness $\boldsymbol{\Sigma} \succ \mathbf{0}$), is given by

\begin{equation}\label{eq: gpenalGauss}
\widehat{\boldsymbol{\Sigma}}_{\bigcdot} \in \text{arg min}_{\boldsymbol{\Sigma} \succ \mathbf{0}} \left \{ \ln \left |\boldsymbol{\Sigma} \right | + \text{tr} \left (\boldsymbol{\Sigma}^{-1}  \widehat{\boldsymbol{\Sigma}}_{\text{ML}} \right ) + P_{\bigcdot}(\boldsymbol{\Sigma},\omega_{n}) \right \},
\end{equation}
and corresponding correlation matrix $\widehat{\mathbf{R}}_{\bigcdot} = \varphi (\widehat{\boldsymbol{\Sigma}}_{\bigcdot})$. However, the core of this paper is that we merely assume a Gaussian copula. Hence, instead of making use of $\widehat{\boldsymbol{\Sigma}}_{\text{ML}}$, we compute $\widehat{\boldsymbol{\Sigma}}_{n}$, being the (block) matrix of sample normal scores rank covariances with entries $n^{-1} \sum_{\ell = 1}^{n} \widehat{Z}_{ij}^{(\ell)} \widehat{Z}_{mt}^{(\ell)}$ for $i,m \in \{1,\dots,k\}, j \in \{1,\dots,d_{i}\}, t \in \{1,\dots,d_{m}\}$ (similar block notation as in \eqref{eq: est cor matrix}). The main difference is that we do not have true Gaussian scores, but only non-parametrically estimated Gaussian scores $\widehat{Z}_{ij}^{(\ell)} = \Phi^{-1}(n/(n+1) \widehat{F}_{ij}(X_{ij}^{(\ell)}))$. The copula formulation of \eqref{eq: gpenalGauss} becomes
\begin{equation}\label{eq: gpenalGaussCop}
\widehat{\boldsymbol{\Sigma}}_{\bigcdot,n} \in \text{arg min}_{\boldsymbol{\Sigma} \succ \mathbf{0}} \left \{ \ln \left |\boldsymbol{\Sigma} \right | + \text{tr} \left (\boldsymbol{\Sigma}^{-1}  \widehat{\boldsymbol{\Sigma}}_{n} \right ) + P_{\bigcdot}(\boldsymbol{\Sigma},\omega_{n}) \right \},
\end{equation}
where we use the additional subscript $n$ in $\widehat{\boldsymbol{\Sigma}}_{\bigcdot,n}$ to indicate that we are in the copula context. A genuine Gaussian copula correlation matrix is then $\widehat{\mathbf{R}}_{\bigcdot,n} = \varphi(\widehat{\boldsymbol{\Sigma}}_{\bigcdot,n})$. 

We go deeper into three choices for the penalty $P_{\bigcdot}$:
\begin{itemize}
 \item[--] the ridge penalty $P_{\text{R}}(\boldsymbol{\Sigma},\omega_{n}) = (1-\omega_{n}) \text{tr}(\mathbf{R}^{-1})$, where $\omega_{n} \in (0,1]$
 \item[--] (adaptive) lasso-type penalties $P_{\text{LT}}(\boldsymbol{\Sigma},\omega_{n}) = \sum_{i,j,m,t} p_{\omega_{n}}(\Delta_{ij,mt} |\sigma_{ij,mt}|)$, where $\sigma_{ij,mt}$ is the $(j,t)$'th element of the $(i,m)$'th block of $\boldsymbol{\Sigma}$ (similar block notation as in \eqref{eq: est cor matrix}), $\Delta_{ij,mt} \geq 0$ is a weight (e.g., zero when $(i,j) = (m,t)$ in order to not shrink diagonal elements, or larger for smaller preliminary estimated entries in order to shrink these more) and $p_{\omega_{n}}$ is a certain penalty function depending on $\omega_{n} \geq 0$
 \item[--] group lasso-type penalties $P_{\text{GLT}}(\boldsymbol{\Sigma},\omega_{n}) =  2 \sum_{i,m = 1 , m > i}^{k} p_{\omega_{n}} (\sqrt{d_{i}d_{m}} \left | \left |\boldsymbol{\Sigma}_{im} \right | \right |_{\text{F}}) + \sum_{i=1}^{k} p_{\omega_{n}} (\sqrt{d_{i}(d_{i}-1)} \left | \left | \boldsymbol{\Delta}_{i} * \boldsymbol{\Sigma}_{ii} \right | \right |_{\text{F}})$, where $|| \cdot ||_{\text{F}}$ is the Frobenius matrix norm, $\boldsymbol{\Delta}_{i} \in \mathbb{R}^{d_{i} \times d_{i}}$ a matrix with ones as off-diagonal elements and zeroes on the diagonal (in order to avoid shrinking the variances), and $p_{\omega_{n}}$ a certain penalty function depending on $\omega_{n} \geq 0$.
\end{itemize}

The ridge penalty is different from the other ones (and will also be considered separately in the simulations in Section \ref{sec:5}) in the sense that it only focuses on improving the estimation of the Gaussian copula covariance matrix (and corresponding dependence coefficients) when $q$ is large compared to $n$, and not on a sparsity assumption. Asymptotic properties are centred around consistency. Having $q > n$ is definitely manageable for ridge penalization. For the latter two penalties, the set $\mathcal{A}$ defined as $\mathcal{A} = \{\alpha : \boldsymbol{\sigma}_{\alpha} \neq 0, \alpha = 1,\dots,\widetilde{q} \}$, where $\boldsymbol{\sigma} = \text{vech}(\boldsymbol{\Sigma}) \in \mathbb{R}^{\widetilde{q}}$ is the vector of upper diagonal elements of $\boldsymbol{\Sigma}$ and $\widetilde{q} = q(q-1)/2$, is of crucial importance. Indeed, next to consistency, we hope that $\mathbb{P}(\mathcal{A}_{n} = \mathcal{A}) \to 1$ for $n \to \infty$, where $\mathcal{A}_{n} = \{\alpha : \widehat{\boldsymbol{\sigma}}_{n,\alpha} \neq 0, \alpha = 1, \dots, \widetilde{q} \}$, with $\widehat{\boldsymbol{\sigma}}_{n} = \text{vech}(\widehat{\boldsymbol{\Sigma}}_{\bigcdot,n})$, a property called sparsistency. Having $q > n$ leads to degeneracy of the lasso-type estimators, which is why we restrict ourselves to $q \leq n$ in the simulations.
\newline 

Example \ref{ex4} shows that the sparsity of $\boldsymbol{\Sigma}$ (which is obviously preserved by $\mathbf{R}$), i.e., the entries belonging to $\mathcal{A}^{c}$, can manifest itself in different forms, calling for different shrinkage penalties. 

\begin{example}\label{ex4}
A covariance graph is a graphical model that represents variables as nodes and marginal dependencies as edges (similar to a Markov random field representing conditional dependencies), see, e.g., Section 1 of \cite{Bien2011} for several references. Marginal independencies correspond to individual zeroes in the correlation matrix and many different sparsity patterns can occur. The first plot of Fig. \ref{fig: SparseSettings} shows the sparsity structure of a correlation matrix of a random covariance graph for a total of $20$ variables, where $32.5\%$ of the elements equal zero, corresponding to the TRUE entries. A penalty of the form $P_{\text{LT}}$ allows reflecting these marginal independencies in the estimated correlation matrix, and could result in a better plug-in estimator for our dependence measures.

\begin{figure}[h!] 
\includegraphics[scale = 0.65]{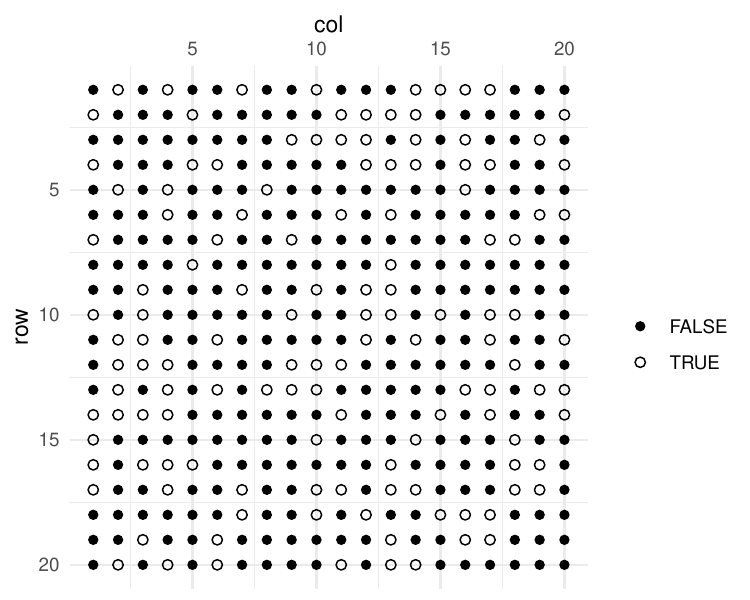}
\includegraphics[scale = 0.65]{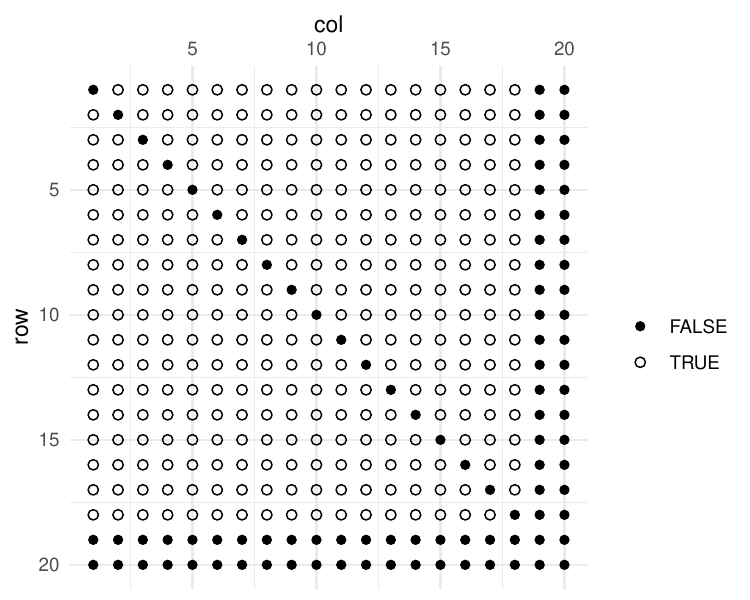}
\caption{Different sparsity patterns of a $20$ dimensional correlation matrix.}
\label{fig: SparseSettings}
\end{figure} 

Next, imagine a person answering a total of twenty questions in the form of seven short questionnaires $\mathbf{X}_{1},\dots,\mathbf{X}_{7}$, all consisting of three questions, except for $\mathbf{X}_{7}$, which only contains two. The interest is in the relationship between the answers of the different questionnaires. Furthermore, assume that $\mathbf{X}_{1},\dots,\mathbf{X}_{6}$ are on completely different topics, and all three questions are each time self-contained. Only the very last two questions contained in $\mathbf{X}_{7}$ are related to each other, and to the other questions in the first six questionnaires as well. Then, it can be expected that only the blocks $\mathbf{R}_{i7}$ for $i = 1,\dots,6$ and $\mathbf{R}_{77}$ are different from zero (and the diagonal of $\mathbf{R}$, of course), i.e., $\mathbf{R}$ has a sparsity pattern as visualized in the second plot of Fig. \ref{fig: SparseSettings}, where $76.5\%$ of the elements equal zero. Such a pattern calls for a $P_{\text{GLT}}$ penalty, enabling shrinkage of entire blocks, and possibly more accurate estimation of the Gaussian copula dependence coefficients compared to using the matrix of normal scores rank correlations $\widehat{\mathbf{R}}_{n}$, or a $P_{\text{LT}}$ penalty.
\end{example}
\vspace{0.3cm}

We shall now elaborate more on the theoretical properties and computational aspects of the estimators $\widehat{\boldsymbol{\Sigma}}_{\text{R},n},\widehat{\boldsymbol{\Sigma}}_{\text{LT},n}$ and $\widehat{\boldsymbol{\Sigma}}_{\text{GLT},n}$ given in \eqref{eq: gpenalGaussCop} corresponding to the three types of penalties $P_{\text{R}},P_{\text{LT}}$ and $P_{\text{GLT}}$.
\\

\noindent
\textbf{Ridge regularization} 
\\

Warton \cite{Warton2008} extensively studies the estimator $\widehat{\boldsymbol{\Sigma}}_{\text{R}}$, i.e., the estimator \eqref{eq: gpenalGauss} under the assumption of a normal distribution, with penalty term $P_{\text{R}}(\boldsymbol{\Sigma},\omega_{n}) = (1-\omega_{n})\text{tr}(\mathbf{R}^{-1})$ for $\mathbf{R} = \varphi (\boldsymbol{\Sigma})  = \mathbf{D}_{\mathbf{\Sigma}}^{-1/2}\mathbf{\Sigma}\mathbf{D}_{\mathbf{\Sigma}}^{-1/2}$ and $\omega_{n} \in (0,1]$.
His Theorem 1 tells us that 
\begin{equation*}
    \widehat{\mathbf{R}}_{\text{R}} = \omega_{n} \widehat{\mathbf{R}}_{\text{ML}} + (1-\omega_{n})\mathbf{I}_{q},
\end{equation*}
where $\widehat{\mathbf{R}}_{\text{ML}} = \varphi(\widehat{\boldsymbol{\Sigma}}_{\text{ML}})$ is the corresponding maximum penalized likelihood estimator. The value of $\omega_{n}$ is picked through $K$-fold cross-validation with the normal likelihood as objective function. If $\widehat{\lambda}_{\text{ML}}$ is an eigenvalue of $\widehat{\mathbf{R}}_{\text{ML}}$, then $\omega_{n} \widehat{\lambda}_{\text{ML}} + 1 - \omega_{n}$ is an eigenvalue of $\widehat{\mathbf{R}}_{\text{R}}$, i.e., all eigenvalues smaller than one are expanded, and all eigenvalues larger than one are shrunk towards one (recall the discrepancy of biased eigenvalue dispersion of the non-regularized empirical covariance matrix), at a pace that increases as $\omega_{n}$ decreases. Furthermore, the value of $\omega_{n}$ tends to one in probability if $n \to \infty$, and there is zero probability that $\omega_{n} = 1$ (no regularization) if $\widehat{\mathbf{R}}_{\text{ML}}$ does not have full rank, i.e., $\widehat{\mathbf{R}}_{\text{R}}$ is guaranteed to be non-singular, even if $q > n$ (see Theorem 2 and Theorem 3 in \cite{Warton2008}). In our Gaussian copula context, we suggest the estimator
\begin{equation}\label{eq: Ridge_est}
    \widehat{\mathbf{R}}_{\text{R},n} = \omega_{n} \widehat{\mathbf{R}}_{n} + (1-\omega_{n})\mathbf{I}_{q},
\end{equation}
and the same $K$-fold cross-validation procedure for selecting $\omega_{n}$, but now based on an estimated Gaussian sample. Heuristically (we do not go into detail here), the asymptotic properties of $\widehat{\mathbf{R}}_{\text{R}}$ carry over to $\widehat{\mathbf{R}}_{\text{R},n}$, since we know that (see, e.g., the proof of Theorem \ref{thm1}) $||\widehat{\mathbf{R}}_{\text{ML}} - \widehat{\mathbf{R}}_{n}||_{\infty} = \smallO_{p}(n^{-1/2})$ for $n \to \infty$.
\newline \\ \noindent 
\textbf{Regularized estimation and sparsity}
\\

Concentrating on the fully Gaussian setting, the general penalization criterion $P_{\text{LT}}(\boldsymbol{\Sigma},\omega_{n}) = \sum_{i,j,m,t} p_{\omega_{n}}(\Delta_{ij,mt}|\sigma_{ij,mt}|)$ is considered in \cite{Lam2009}. Typically, one takes $\Delta_{ij,mt} = 0$ if $(i,j) = (m,t)$ and $1$ otherwise in order to avoid penalization of the diagonal elements, which do not vanish. Another possibility is $\Delta_{ij,mt} = |\widehat{\sigma}_{ij,mt}|^{-1} \mathds{1}(|\widehat{\rho}_{ij,mt}| < \rho_{n})$, where $\widehat{\sigma}_{ij,mt}$ is a preliminary estimate for $\sigma_{ij,mt}$ with corresponding correlation  $\widehat{\rho}_{ij,mt}$ and $0 < \rho_{n} < 1$ a certain threshold value, i.e., we only shrink those elements that have a sufficiently small preliminary estimated correlation, and the amount of shrinkage is proportional to the size of the preliminary estimated covariance. This is an idea similar to the adaptive lasso of \cite{Fan2009} and used in, e.g., \cite{Fop2021}.
A similar problem often arises when estimation of the precision matrix (inverse covariance matrix) is of interest, think of graphical models for example (see, e.g., \cite{Lam2009} for some specific references). It is usually solved by first performing a local linear approximation to the penalty function (see, e.g., \cite{Zou2008} in case of a precision matrix):

\begin{equation*}
 p_{\omega_{n}} \left (\Delta_{ij,mt} |\sigma_{ij,mt}| \right ) \approx p_{\omega_{n}} \left (\Delta_{ij,mt} \left | \widehat{\sigma}_{ij,mt}^{(w)}\right | \right ) + \Delta_{ij,mt} \hspace{0.06cm} p_{\omega_{n}}^{\prime} \left (\Delta_{ij,mt} \left |\widehat{\sigma}_{ij,mt}^{(w)} \right | \right ) \left (|\sigma_{ij,mt}| - \left |\widehat{\sigma}_{ij,mt}^{(w)}\right | \right ),
\end{equation*}
with $\widehat{\sigma}_{ij,mt}^{(w)}$ a current estimated entry of $\boldsymbol{\Sigma}$ in step $w$. Hence, $\widehat{\boldsymbol{\Sigma}}_{\text{LT}}^{(w+1)}$ should be taken as (typically, one iteration already suffices for satisfactory results)
\begin{equation}\label{eq: stepll}
    \widehat{\boldsymbol{\Sigma}}_{\text{LT}}^{(w+1)} \in \text{arg min}_{\boldsymbol{\Sigma}} \left \{ \ln \left |\boldsymbol{\Sigma} \right | + \text{tr} \left (\boldsymbol{\Sigma}^{-1} \widehat{\boldsymbol{\Sigma}}_{\text{ML}} \right ) + \sum_{i,m,j,t} \Delta_{ij,mt} \hspace{0.06cm} p_{\omega_{n}}^{\prime} \left (\Delta_{ij,mt} \left |\widehat{\sigma}_{ij,mt}^{(w)}\right | \right ) |\sigma_{ij,mt}|  \right \},
\end{equation}
which is a weighted covariance graphical lasso problem with weights $$\widetilde{\Delta}_{ij,mt} = \Delta_{ij,mt} \hspace{0.06cm} p_{\omega_{n}}^{\prime} \left (\Delta_{ij,mt} \left |\widehat{\sigma}_{ij,mt}^{(w)}\right | \right ). $$ We can summarize the $\widetilde{\Delta}_{ij,mt}$ values in a block matrix $\widetilde{\boldsymbol{\Delta}}$ (again similar as in \eqref{eq: est cor matrix}), allowing us to rewrite \eqref{eq: stepll} as
\begin{equation}\label{eq: step2ll}
    \widehat{\boldsymbol{\Sigma}}_{\text{LT}}^{(w+1)} \in \text{arg min}_{\boldsymbol{\Sigma}} \left \{ \ln \left |\boldsymbol{\Sigma} \right | + \text{tr} \left (\boldsymbol{\Sigma}^{-1} \widehat{\boldsymbol{\Sigma}}_{\text{ML}} \right ) + \left |\left|\widetilde{\boldsymbol{\Delta}} * \boldsymbol{\Sigma} \right | \right |_{1}\right \},
\end{equation}
where $||\mathbf{A}||_{1} = ||\text{vec} (\mathbf{A})||_{1} = \sum_{ij}|\mathbf{A}_{ij}|$ is the $L_{1}$-norm of the vector of all elements contained in $\mathbf{A}$, and $*$ stands for elementwise multiplication.
As illustrated in \cite{Bien2011}, the optimization \eqref{eq: step2ll} consists of a convex part $\text{tr}(\boldsymbol{\Sigma}^{-1} \widehat{\boldsymbol{\Sigma}}_{\text{ML}}) +  ||\widetilde{\boldsymbol{\Delta}} * \boldsymbol{\Sigma} ||_{1}$, and a concave part $\ln  |\boldsymbol{\Sigma}|$, making the entire problem non-convex (big difference with the precision matrix case, where graphical lasso is convex), i.e., convergence to a global minimum is not guaranteed. Also, when $q > n$, the solution to \eqref{eq: step2ll} will be degenerate because $\widehat{\boldsymbol{\Sigma}}_{\text{ML}}$ is not full rank. The authors of \cite{Bien2011} suggest to use $\widehat{\boldsymbol{\Sigma}}_{\text{ML}} + \delta \mathbf{I}_{q}$ for some $\delta > 0$ in such cases, where $\delta$ is chosen such that, e.g., the resulting matrix has condition number equal to $q$. Still, they encounter difficulties of estimation when $q > n$. For this reason, we restrict ourselves to $q \leq n$ in the simulations. 

Using a majorize approach, they propose in \cite{Bien2011} to solve convex approximations to the original problem in an iterative way. Next to sparsity, their algorithm achieves positive definiteness. Another optimization technique for solving \eqref{eq: step2ll} is developed in \cite{Wang2014}, who uses coordinate descent, resulting in a faster and more stable algorithm. We use this algorithm as it is implemented in \cite{Fop2021}.

Under a set of typical assumptions, consistency and sparsistency results for the estimator $\widehat{\boldsymbol{\Sigma}}_{\text{LT}}$ are given in \cite{Lam2009}. One of their main conclusions is the preference for non-convex penalty functions such as the scad penalty of \cite{Fan2001}:

\begin{equation*}
p_{\omega_{n}}^{\text{scad}}(t) = \begin{cases} \omega_{n} t & \mbox{for } t \leq \omega_{n} \\ \frac{1}{2(a-1)} \left (2a\omega_{n} t-t^{2}-\omega_{n}^{2} \right ) & \mbox{for } \omega_{n} < t \leq a \omega_{n} \\ 
(a+1) \omega_{n}^{2} / 2 & \mbox{for } t > a \omega_{n} \end{cases}, 
\end{equation*}
where $a > 2$. We take $a = 3.7$ because of the arguments given in \cite{Fan2001}. Such penalties shrink less entries that are large in magnitude, and as such reduce the bias. Moreover, a strong theoretical upper bound on the tuning parameter $\omega_{n}$, as is needed for, e.g., the lasso penalty $p^{\text{lasso}}_{\omega_{n}}(t) = \omega_{n} t$ (being the limit of $p_{\omega_{n}}^{\text{scad}}(t)$ for $a \to \infty$) in order to guarantee consistency, is not needed, yielding better sparsity properties. So, also in case of sparse covariance matrix estimation, the scad has the oracle property in the sense of \cite{Fan2001} (zeroes are asymptotically estimated as zero, and estimated non-zeroes are asymptotically normal). 

So far, we have been assuming a multivariate normal model, while the core of this paper is that we merely assume a semi-parametric normal copula model, i.e., within reach is the estimator $\widehat{\boldsymbol{\Sigma}}_{\text{LT},n}$, and not $\widehat{\boldsymbol{\Sigma}}_{\text{LT}}$. As mentioned in \cite{Fermanian2023}, who use the term sparse M-estimator, ``such estimation problems have benefited from a very limited attention so far", and ``the large sample analysis amply differs from the fully parametric viewpoint". Recall the sets $\mathcal{A}_{n}$ and $\mathcal{A}$ introduced for the property of sparsistency. Define also 
\begin{equation}\label{eq: anbn}
a_{n} = \max_{1 \leq \alpha \leq \widetilde{q}} \left \{\frac{dp_{\omega_{n}}}{dt} \left (\left |\widehat{\boldsymbol{\sigma}}_{n,\alpha} \right | \right ) : \alpha \in \mathcal{A}_{n} \right \} \hspace{0.3cm} \text{and} \hspace{0.3cm} b_{n} = \max_{1 \leq \alpha \leq \widetilde{q}} \left \{\frac{d^{2}p_{\omega_{n}}}{dt^{2}} \left (\left |\widehat{\boldsymbol{\sigma}}_{n,\alpha} \right | \right ) : \alpha \in \mathcal{A}_{n} \right \},
\end{equation}
where $dp_{\omega_{n}}(|\widehat{\boldsymbol{\sigma}}_{n,\alpha}|)/dt$ denotes the derivative of $p_{\omega_{n}}(t)$ with respect to $t$, evaluated in $|\widehat{\boldsymbol{\sigma}}_{n,\alpha}|$, and similarly for the second derivative.
The sequence $a_{n}$ is related to the asymptotic bias of the penalized estimator, and equals $a_{n} = \omega_{n}$ for the lasso penalty.  Theorem \ref{thm2} states the consistency and sparsistency of the estimator $\widehat{\boldsymbol{\Sigma}}_{\text{LT},n}$.

\begin{theorem}\label{thm2}
Assume that there exist $M,\widetilde{M} > 0$ such that 
\begin{equation}\label{eq: pC1}
\left |\frac{d^{2}p_{\omega_{n}}}{dt^{2}} \left (t_{1} \right ) - \frac{d^{2}p_{\omega_{n}}}{dt^{2}} \left (t_{2} \right ) \right | \leq M \left |t_{1} - t_{2} \right |,
\end{equation} for any $t_{1},t_{2} \in \mathbb{R}$ such that $t_{1},t_{2} > \widetilde{M} \omega_{n}$. Suppose that for $a_{n}$ and $b_{n}$ defined in \eqref{eq: anbn}, it holds that $a_{n},b_{n} \to 0$ when $n \to \infty$. Then, there exists a solution $\widehat{\boldsymbol{\Sigma}}_{\text{LT},n}$ to \eqref{eq: gpenalGaussCop} with penalty $P_{\bigcdot} = P_{\text{LT}}$, for each $n$, satisfying 
\begin{equation*}
\left | \left |\widehat{\boldsymbol{\sigma}}_{n} - \boldsymbol{\sigma} \right | \right |_{2} = \mathcal{O}_{p} \left (\ln(\ln(n)) n^{-1/2} + a_{n} \right )
\end{equation*}
for $n \to \infty$, where $\widehat{\boldsymbol{\sigma}}_{n} = \text{vech}(\widehat{\boldsymbol{\Sigma}}_{\text{LT},n})$ and $\boldsymbol{\sigma} = \text{vech}(\boldsymbol{\Sigma})$. If, in addition, $\omega_{n} \to 0$, $a_{n}\omega_{n}^{-1} \to 0$, $\sqrt{n} \omega_{n} (\ln(\ln(n)))^{-1} \to \infty$ when $n \to \infty$, and 
\begin{equation}\label{eq: pC2}
\lim \inf_{ \hspace{-0.5cm} n \to \infty} \lim \inf_{\hspace{-0.5cm} t \to 0} \omega_{n}^{-1} \frac{dp_{\omega_{n}}}{dt} \left ( \left |t \right | \right ) > 0,
\end{equation} it also holds that
\begin{equation*}
\lim_{n \to \infty} \mathbb{P} \left (\mathcal{A}_{n} = \mathcal{A} \right ) = 1.
\end{equation*} 
\end{theorem}
\begin{proof}
This follows immediately from Theorem 3.1. and Theorem 3.2. in \cite{Fermanian2023}, noting that the necessary assumptions are satisfied, since they are verified for the Gaussian copula case and Gaussian likelihood loss function in Section E of the Appendix in \cite{Fermanian2023}.
\end{proof}

The conditions \eqref{eq: pC1} and \eqref{eq: pC2} are satisfied by, e.g., the lasso and scad penalty. The condition in \eqref{eq: pC1} is a smoothing condition on the penalty function, and \eqref{eq: pC2} guarantees sparsity in the estimates. However, note that the condition $a_{n}\omega_{n}^{-1} \to 0$ for $n \to \infty$ cannot be fulfilled by the lasso since then $a_{n} = \omega_{n}$. Hence, Theorem \ref{thm2} does not guarantee sparsistency of the lasso estimator.
\newline \\ \noindent 
\textbf{Regularized estimation and group sparsity}
\newline 

Recall the second plot in Fig. \ref{fig: SparseSettings} of Example \ref{ex4}, where entire groups of variables are independent of each other, and/or have no dependence within, resulting in a correlation sparsity pattern with entire zero blocks. A penalty of the form $P_{\text{GLT}}(\boldsymbol{\Sigma},\omega_{n}) =  2 \sum_{i,m = 1 , m > i}^{k} p_{\omega_{n}} (\sqrt{d_{i}d_{m}} \left | \left |\boldsymbol{\Sigma}_{im} \right | \right |_{\text{F}}) + \sum_{i=1}^{k} p_{\omega_{n}} (\sqrt{d_{i}(d_{i}-1)} \left | \left | \boldsymbol{\Delta}_{i} * \boldsymbol{\Sigma}_{ii} \right | \right |_{\text{F}})$, where $||\cdot||_{\text{F}}$ is the Frobenius norm and $\boldsymbol{\Delta}_{i} \in \mathbb{R}^{d_{i} \times d_{i}}$ a matrix with ones as off-diagonal elements and zeroes on the diagonal in order to avoid shrinkage of the variances, might lead to an estimator $\widehat{\boldsymbol{\Sigma}}_{\text{GLT},n}$ that incorporates this sparsity structure. 
The non-differentiability of $|| \cdot ||_{\text{F}}$ at $\mathbf{0}$ allows to simultaneously shrink all entries of a certain block, and hence performs a group selection. For actually computing $\widehat{\boldsymbol{\Sigma}}_{\text{GLT},n}$, a local linear approximation of $p_{\omega_{n}}$ can again be performed, which will boil down to a similar problem as finding $\widehat{\boldsymbol{\Sigma}}_{\text{GLT},n}$ in case of the lasso penalty $p^{\text{lasso}}_{\omega_{n}}(t) = \omega_{n} t$. So, let us assume from now on that $p_{\omega_{n}} = p_{\omega_{n}}^{\text{lasso}}$.
The ideas explained in \cite{Wang2014}, which we use for numerically finding $\widehat{\boldsymbol{\Sigma}}_{\text{LT},n}$, can in general not be used for finding $\widehat{\boldsymbol{\Sigma}}_{\text{GLT},n}$, because some of the arguments do not apply to the Frobenius norm. Nevertheless, the problem 
\begin{equation}\label{eq: grouplFr}
\widehat{\boldsymbol{\Sigma}}_{\text{GLT},n}^{*} = \text{arg min}_{\boldsymbol{\Sigma}} \left \{\frac{1}{2} \left | \left |\boldsymbol{\Sigma} - \widehat{\boldsymbol{\Sigma}}_{n} \right | \right |_{\text{F}}^{2} + 2 \omega_{n}  \sum_{\substack{i,m = 1 \\ m > i}}^{k} \sqrt{d_{i}d_{m}} \left | \left |\boldsymbol{\Sigma}_{im} \right | \right |_{\text{F}} + \omega_{n} \sum_{i=1}^{k} \sqrt{d_{i} (d_{i}-1)} \left | \left | \boldsymbol{\Delta}_{i} * \boldsymbol{\Sigma}_{ii} \right | \right |_{\text{F}} \right \}
\end{equation}
has a solution in the form of an elementwise soft thresholding operation. Denoting $(\widehat{\boldsymbol{\Sigma}}^{*}_{\text{GLT},n})_{ij,mt}$ for the $(j,t)$'th element of the $(i,m)$'th block of $\widehat{\boldsymbol{\Sigma}}^{*}_{\text{GLT},n}$, similarly for $\widehat{\boldsymbol{\Sigma}}_{n}$, and further 
\begin{equation*}
\mathcal{S}_{i,m} = \sqrt{\mathop{\sum_{j = 1}^{d_{i}}\sum_{t=1}^{d_{m}}}_{(i,j) \neq (m,t)} \left \{ \left (\widehat{\boldsymbol{\Sigma}}_{n} \right )_{ij,mt} \right \}^{2}}, \hspace{0.3cm} \text{and} \hspace{0.3cm} \gamma_{i,m} = \begin{cases} \sqrt{d_{i}d_{m}} & \mbox{if } i \neq m  \\ \sqrt{d_{i}(d_{i}-1)} & \mbox{if } i = m, \end{cases}
\end{equation*}
it is known that (similar to, e.g., Proposition 1 in \cite{Bigot2011})
\begin{equation*}
\left (\widehat{\boldsymbol{\Sigma}}^{*}_{\text{GLT},n} \right )_{ij,mt} = \begin{cases} 0 & \mbox{if } \mathcal{S}_{i,m} \leq \omega_{n} \gamma_{i,m} \hspace{0.1cm} \text{and} \hspace{0.1cm} (i,j) \neq (m,t) \\ \left (\widehat{\boldsymbol{\Sigma}}_{n}  \right )_{ij,mt} & \mbox{if } (i,j) = (m,t) \\ \left (\widehat{\boldsymbol{\Sigma}}_{n}  \right )_{ij,mt} \left (1 - \frac{\omega_{n} \gamma_{i,m}}{\mathcal{S}_{i,m}} \right ) & \mbox{if } \mathcal{S}_{i,m} > \omega_{n} \gamma_{i,m} \hspace{0.1cm} \text{and} \hspace{0.1cm} (i,j) \neq (m,t). \end{cases}
\end{equation*}
Hence, we can numerically compute $\widehat{\boldsymbol{\Sigma}}_{\text{GLT},n}$ by using ideas similar to the optimization approach of \cite{Bien2011}, i.e., by majorizing $\ln | \boldsymbol{\Sigma}|$ by its tangent plane and using generalized gradient descent steps, which comes down to iteratively solving (convex) problems of the form \eqref{eq: grouplFr}. For the actual implementation, we used to code available in \cite{Bien2022}, and fine-tuned it such that it can also cope with a group lasso penalty.
Regarding the asymptotic properties, it is intuitively clear that a result similar to Theorem \ref{thm2} will hold, with sparsistency formulated at the level of blocks instead of individual elements. For asymptotic properties in case of truly independent copies, we refer to Section 3 of \cite{Bigot2011}.

\subsection{Dimension depending on the sample size}\label{subsec: 4.2}

So far, we have been assuming that the total dimension $q$ of the random vector $\mathbf{X}$ remains fixed with $n$. When $\mathbf{X} \in \mathbb{R}^{q_{n}}$ for $q_{n}$ depending on $n$, of primary interest might be the behaviour of the empirical Gaussian copula covariance matrix $\widehat{\mathbf{R}}_{n}$ given in \eqref{eq: est cor matrix}. Proposition \ref{prop4} (see \ref{App D} for a proof) tells us how the dimension $q_{n}$ influences the consistency of the estimator $\widehat{\mathbf{R}}_{n}$ in max norm $||\mathbf{A}||_{\infty} = \max_{i,j}|\mathbf{A}_{ij}|$. 

\begin{proposition}\label{prop4}
Let $q_{n}$ be a sequence of dimensions depending on $n$, and $\mathbf{R}_{n} \in \mathbb{R}^{q_{n} \times q_{n}}$ corresponding Gaussian copula covariance matrices.  Assume that $\sup_{n} \lambda_{\max}(\mathbf{R}_{n}) < \infty$, where $\lambda_{\max}$ denotes the maximum eigenvalue, and let $\widehat{\mathbf{R}}_{n}$ be the estimator given in \eqref{eq: est cor matrix} with $q = q_{n}$. Then, it holds that 
\begin{equation*}
\left | \left |\widehat{\mathbf{R}}_{n} - \mathbf{R}_{n} \right | \right |_{\infty} = \mathcal{O}_{\text{p}} \left \{ \left ( \ln (q_{n}) / n \right )^{1/2} \right \},
\end{equation*}
for $n \to \infty$.
\end{proposition}
\vspace{0.3cm}

So, Proposition \ref{prop4} ensures that $\widehat{\mathbf{R}}_{n}$ is a consistent estimator as long as $\ln(q_{n})/n \to 0$ as $n \to \infty$, i.e., as long as we are not in an ultra-high dimensional setting. In the context of the penalization techniques, one can also assume that $q = q_{n}$ depends on $n$. 
For the ridge regularization, this will lead to inconsistencies in high-dimensional settings, see Section 3.1 in \cite{Warton2008} (basically because sample eigenvalues are known to be inconsistent when $q_{n} \to \infty$). Regarding the $P_{\text{LT}}$ penalties, Theorem H.1. in \cite{Fermanian2023} states that
\begin{equation*}
\left | \left |\widehat{\boldsymbol{\sigma}}_{n} - \boldsymbol{\sigma} \right | \right |_{2} = \mathcal{O}_{p} \left \{\sqrt{\hspace{0.05cm}\widetilde{q}_{n}} \left (\ln(\ln(n))n^{-1/2} + a_{n} \right ) \right \},
\end{equation*}
where $\widetilde{q}_{n} = q_{n}(q_{n}-1)/2$, and 
under the additional conditions (next to those of Theorem \ref{thm2}) that $\widetilde{q}_{n}^{2} \ln(\ln(n)) n^{-1/2} \to 0$ and $\widetilde{q}_{n}^{2} a_{n} \to 0$ for $n \to \infty$. Hence, $q_{n}$ is allowed to increase with $n$, but consistency requires $\widetilde{q}_{n} = \smallO(n^{1/4})$. If, in addition (see Theorem H.2. in \cite{Fermanian2023}), $\widetilde{q}_{n} a_{n} \omega_{n}^{-1}$ and $\sqrt{n} \omega_{n} (\widetilde{q}_{n}\ln(\ln(n)))^{-1} \to \infty$ for $n \to \infty$, true zeroes are asymptotically identified with probability one (but there might still be some false positives).

\section{Simulation study}\label{sec:5}

The aim of this section is to empirically study the finite sample performance of the (non-)regularized plug-in estimators for the Gaussian copula based dependence coefficients $\mathcal{D}_{1}$ and $\mathcal{D}_{2}$ discussed in Section \ref{sec:3}. The finite-sample distribution of the estimator of Section \ref{sec:3}, is investigated in Section \ref{subsec: 5.1}, and the regularized estimators are considered in Section \ref{subsec: 5.2}.

\subsection{Finite-sample performance of the estimator of Section \ref{sec:3}}\label{subsec: 5.1}

Theorem \ref{thm1} gives an asymptotic normality result with explicit asymptotic variance for the estimator $\widehat{\mathcal{D}}_{r,n} = \mathcal{D}_{r}(\widehat{\mathbf{R}}_{n})$ for $r \in \{1,2\}$, with $\widehat{\mathbf{R}}_{n}$ the matrix of sample normal scores rank correlations given in \eqref{eq: est cor matrix}. Let now $\widehat{\zeta}_{r,n}$ be the plug-in estimator of the asymptotic standard deviation $\zeta_{r}$ obtained by replacing $\mathbf{R}$ with $\widehat{\mathbf{R}}_{n}$. By simulating a sample from a certain multivariate distribution having a Gaussian copula, we can compute a realization of the actual sampling distribution of the studentized estimator $\sqrt{n}(\widehat{\mathcal{D}}_{r,n}-\mathcal{D}_{r})/\widehat{\zeta}_{r,n}$, and several replications can be used to represent characteristics of the entire distribution, which should approximately (for large $n$) be a standard normal one according to Theorem \ref{thm1}.

We consider four settings which we can generate samples from:

\begin{figure}[h!] 
\includegraphics[scale = 0.9]{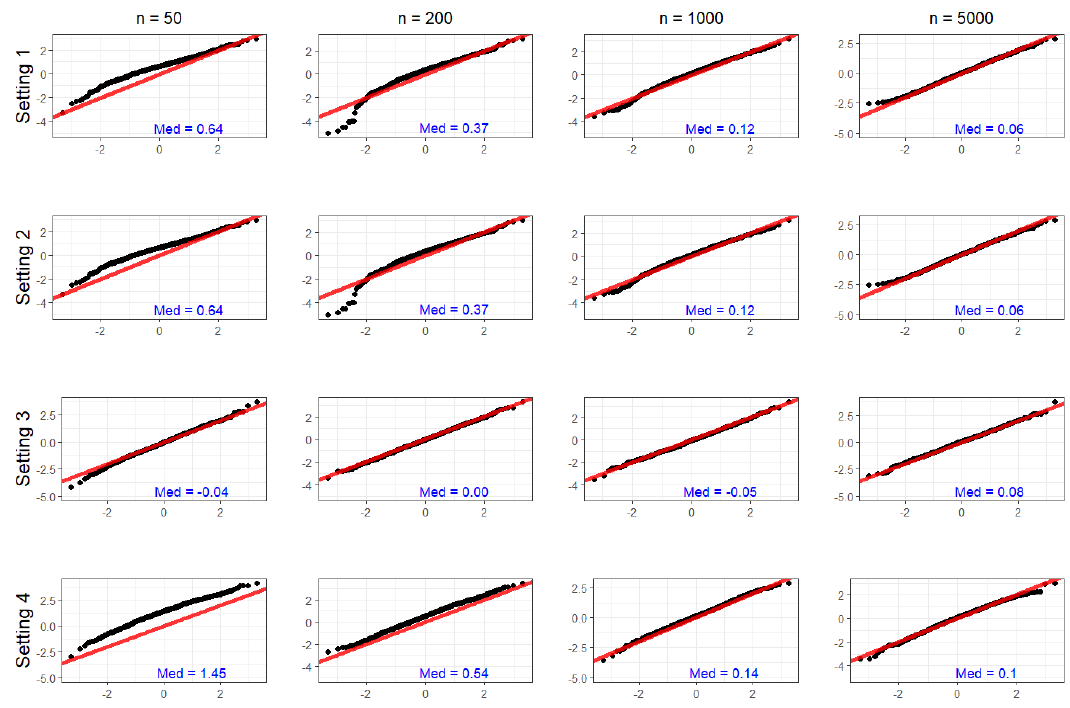}
\caption{Standard normal Q--Q plots for $1000$ Monte Carlo runs of the studentized plug-in estimator for $\mathcal{D}_{1}$ in four different settings with sample sizes $n = 50, 200, 1000, 5000$. The median (``Med") of the studentized estimates is given in blue.}
\label{fig: QQplots}
\end{figure} 

\begin{itemize}
    \item Setting 1: $k = 2, d_{1} = d_{2} = 2$, with standard normal marginals and a Gaussian copula having an autoregressive AR(1) correlation matrix with $\rho = 0.25$ (i.e., the $(i,j)$'th element of $\mathbf{R}$ equals $0.25^{|i-j|}$).
    \item Setting 2: as Setting 1, but now with marginals
    \begin{itemize}
        \item a $t$ distribution with $3$ degrees of freedom for $X_{11}$
        \item an exponential distribution with mean $1$ for $X_{12}$
        \item a beta distribution with parameters $2$ and $2$ for $X_{21}$
        \item an $F$-distribution with degrees of freedom $2$ and $6$ for $X_{22}$.
    \end{itemize}
    \item Setting 3: similar to Setting 1,  but with $\rho = 0.8$.
    \item Setting 4: $k = 5, d_{1} = 4, d_{2} = 5, d_{3} = 3, d_{4} = 1, d_{5} = 2$, with standard normal marginals and a Gaussian copula having an equicorrelated correlation matrix with $\rho = 0.5$.
\end{itemize}

Each time, we draw $1000$ samples of sizes $n = 50,200,1000,5000$ and make standard normal Q--Q plots to assess the goodness-of-fit with a standard normal distribution. See Fig. \ref{fig: QQplots} for the results in case of $\mathcal{D}_{1}$. Similar plots are obtained for $\mathcal{D}_{2}$. In each setting, there is a qualitative fit with the standard normal distribution for larger sample sizes. Changing the marginals has no influence (Setting 1 versus Setting 2). Small correlations (Settings 1 and 2), yield a more pronounced lack-of-fit than higher correlations (Setting 3). Increasing the total dimension to $15$ (Setting 4) results in a large positive bias for small sample sizes, calling for regularization techniques.

Next, in Fig. \ref{fig: xmplN2}, asymptotic standard deviations $\zeta_{1}$ and $\zeta_{2}$ were shown in the context of Example \ref{ex3}. Generating $N$ samples from, e.g., a multivariate normal distribution with mean zero and covariance matrix given in \eqref{eq: settingOT}, we obtain $N$ estimates $\widehat{\mathcal{D}}_{r,n}^{(1)},\dots, \widehat{\mathcal{D}}_{r,n}^{(N)}$, whose sample standard deviation multiplied with $\sqrt{n}$, say $\widehat{\zeta}_{r}$, can be seen as an approximation for $\zeta_{r}$ when $n$ is large.  Fig. \ref{fig: xmplN2Emp} depicts $\widehat{\zeta}_{r}$ for different values of $\rho_{1}$ and $\rho_{2}$ in case $n = 10 \hspace{0.05cm} 000$ and $N = 1000$. We see the same patterns popping up as in Fig. \ref{fig: xmplN2}, illustrating the asymptotic variance formula given in Theorem \ref{thm1} empirically in this particular setting. 

\begin{figure}[h!] 
\includegraphics[scale = 0.64]{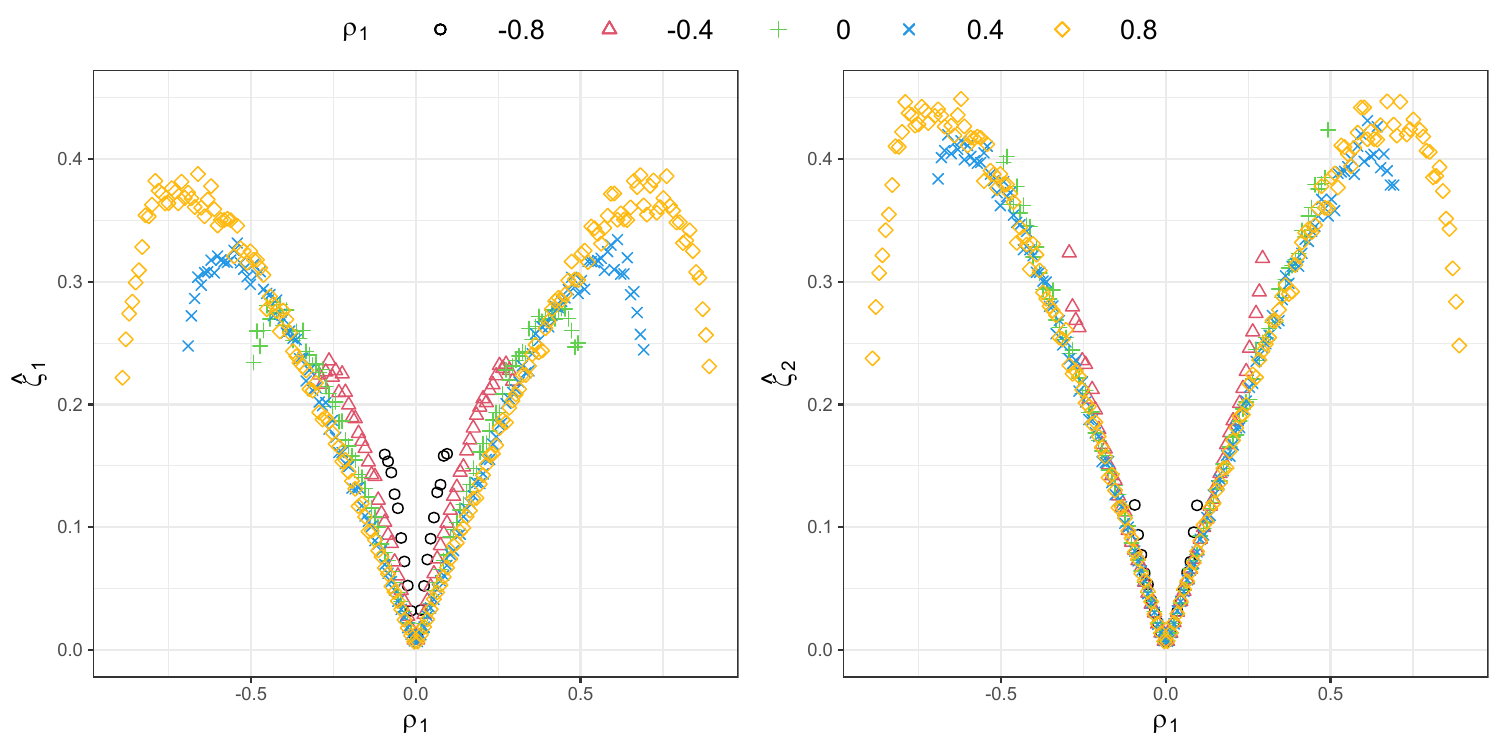}
\caption{Empirical standard deviation (sample size $N = 1000$) $\widehat{\zeta}_{r}$ for $r \in \{1,2\}$ in the setting of Example \ref{ex3} when $n = 10 \hspace{0.05cm} 000$.}
\label{fig: xmplN2Emp}
\end{figure} 

\subsection{Penalization techniques}\label{subsec: 5.2}

We now turn our attention to the different covariance matrix penalization techniques discussed in Section \ref{sec: 4}. We start with illustrating ridge regularization, in particular the improvements it gives in plug-in estimation of Gaussian copula based dependence coefficients between multiple random vectors. We do this for increasing values of $q$, possibly larger than $n$. Recall that the ridge estimator is rather easy to compute, and there are no additional difficulties when $q > n$. This study thus also allows for an impression on how the dependence coefficients behave with increasing $q$. Afterwards, we go to sparsity inducing (group)-lasso type methods, and investigate for two fixed values of $q$ and different sample sizes (with $q < n$) their ability of recovering marginal independencies (interpretability) on the one hand, and whether this improves the estimation of the dependence coefficients (accuracy) on the other hand. 
\newline \\ \noindent
\textbf{Ridge regularization} \newline 

As explained in Section \ref{sec: 4}, the ridge estimator $\widehat{\mathbf{R}}_{\text{R},n}$ of the Gaussian copula correlation matrix tries to cope with biased eigenvalue dispersion of the empirical correlation matrix which aggravates when $q$ is large compared to $n$. Moreover, it is easy to implement (with a straightforward cross-validation procedure for selecting the penalty parameter $\omega_{n}$), and guarantees a positive definite outcome. 
One can expect that the performance of the estimator $\mathcal{D}_{r}(\widehat{\mathbf{R}}_{\text{R},n})$ is better than the performance of $\mathcal{D}_{r}(\widehat{\mathbf{R}}_{n})$ when $q$ is large compared to $n$. To demonstrate this, we consider the following two designs: 

\begin{itemize}
\item Design 1: We let $q \in \{4,6,8,\dots,98\}$ and $k = q/2$ with $d_{1} = d_{2} = \cdots = d_{k} = 2$, e.g., if $q = 30$, we are measuring the dependence between $15$ random vectors of dimension $2$
\item Design 2: We let $q \in \{4,6,8,\dots,98\}$ and $k = 2$ with $d_{1} = d_{2} = q/2$, e.g., if $q = 30$, we are measuring the dependence between $2$ random vectors of dimension $15$.
\end{itemize}
In each design, we generate $1000$ samples of sizes $n = 50,100,500$ from a $\mathcal{N}_{q} (\mathbf{0}_{q},\mathbf{R} )$ distribution (but, assume unknown marginals), where $\mathbf{R}$ is the correlation matrix of an $\text{AR}(1)$ process with $\rho = 0.5$, and compute the empirical mean squared error of $\mathcal{D}_{r}(\widehat{\mathbf{R}}_{\text{R},n})$ and $\mathcal{D}_{r}(\widehat{\mathbf{R}}_{n})$.

\begin{figure}[h!] 
\includegraphics[scale = 0.43]{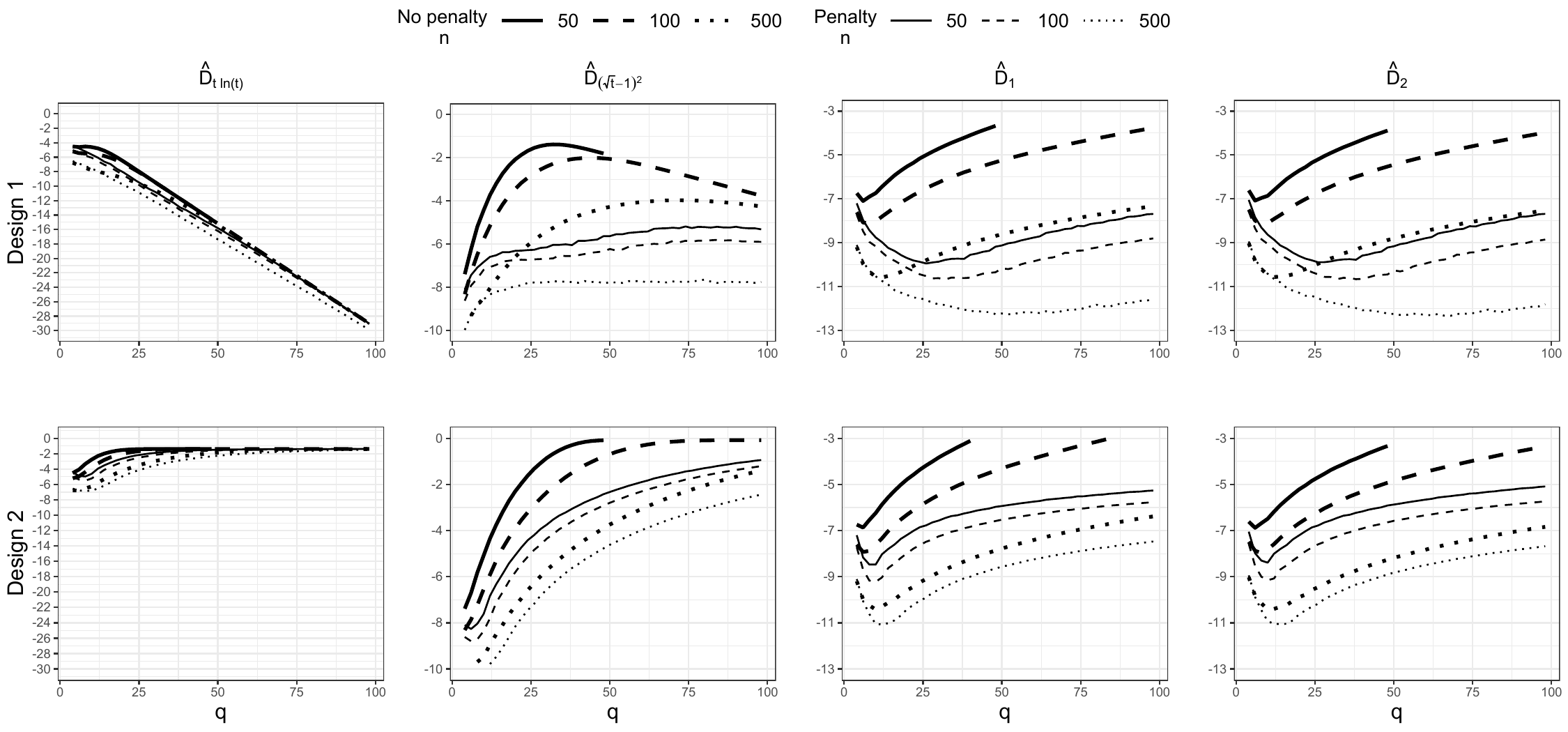}
\caption{Logarithm of the Monte Carlo mean squared error of the estimators $\mathcal{D}_{\bigcdot}(\widehat{\mathbf{R}}_{\text{R},n})$ (penalty) and $\mathcal{D}_{\bigcdot}(\widehat{\mathbf{R}}_{n})$ (no penalty) for $\mathcal{D}_{t \ln(t)}, \mathcal{D}_{(\sqrt{t}-1)^{2}}, \mathcal{D}_{1}$ and $\mathcal{D}_{2}$, based on $1000$ replications with sample sizes $n = 50,100,500$ as a function of $q$, in two different designs.}
\label{fig: Ridge}
\end{figure} 

The penalty parameter $\omega_{n}$ is determined by $5$-fold cross-validation (as described by equation (4) in \cite{Warton2008}) on a grid of $50$ equidistant elements in $[0.01,0.999]$. In addition, we do the same for two Gaussian copula-based $\Phi$-dependence measures discussed in \cite{Gijbels2023b}, known as the (normalized) mutual information and Hellinger distance, respectively given by
\begin{equation}\label{eq: mutNhelN}
        \mathcal{D}_{t\ln(t)}(\mathbf{R}) = \left ( 1 - \frac{\left |\mathbf{R} \right |}{\prod_{i = 1}^{k} \left |\mathbf{R}_{ii} \right |} \right )^{1/2}, \hspace{0.2cm} \text{and} \hspace{0.2cm}
  \mathcal{D}_{(\sqrt{t}-1)^{2}}(\mathbf{R}) =  1 - \frac{2^{q/2}\left |\mathbf{R} \right |^{1/4}}{\left |\mathbf{I}_{q}+\mathbf{R}_{0}^{-1} \mathbf{R} \right |^{1/2}\prod_{i=1}^{k}\left |\mathbf{R}_{ii} \right |^{1/4}},
\end{equation}
where $\mathbf{R}_{0}$ is given in \eqref{eq: R0}, and $| \cdot |$ denotes the determinant. 
Note that the dependence coefficients in \eqref{eq: mutNhelN} depend on products of eigenvalues (they are based on divergences of copula densities), while $\mathcal{D}_{1}$ and $\mathcal{D}_{2}$ depend on sums of eigenvalues (arising from the Bures-Wasserstein distance). 

Fig. \ref{fig: Ridge} shows plots of the logarithm of the Monte Carlo mean squared error of $\mathcal{D}_{\bigcdot}(\widehat{\mathbf{R}}_{\text{R},n})$ (penalty) and $\mathcal{D}_{\bigcdot}(\widehat{\mathbf{R}}_{n})$ (no penalty) for the four dependence measures for each sample size in both designs as a function of $q$. In all cases, we see vast improvements when using ridge regularization, especially when $n$ is small compared to $q$ (note that when $n = 50$ and $q > 50$, the estimators $\mathcal{D}_{\bigcdot}(\widehat{\mathbf{R}}_{n})$ are not defined because of singularity of $\widehat{\mathbf{R}}_{n})$. Fig. \ref{fig: RidgeOmega} shows boxplots of the selected values for the penalization parameter $\omega_{n}$. We clearly see compatibility with the theoretical property that $\omega_{n}$ tends to one in probability when $n \to \infty$, and higher dimensions require a larger correction of the eigenvalue dispersion, i.e., a smaller $\omega_{n}$. 

\begin{figure}[h!] \centering 
\includegraphics[scale = 0.8]{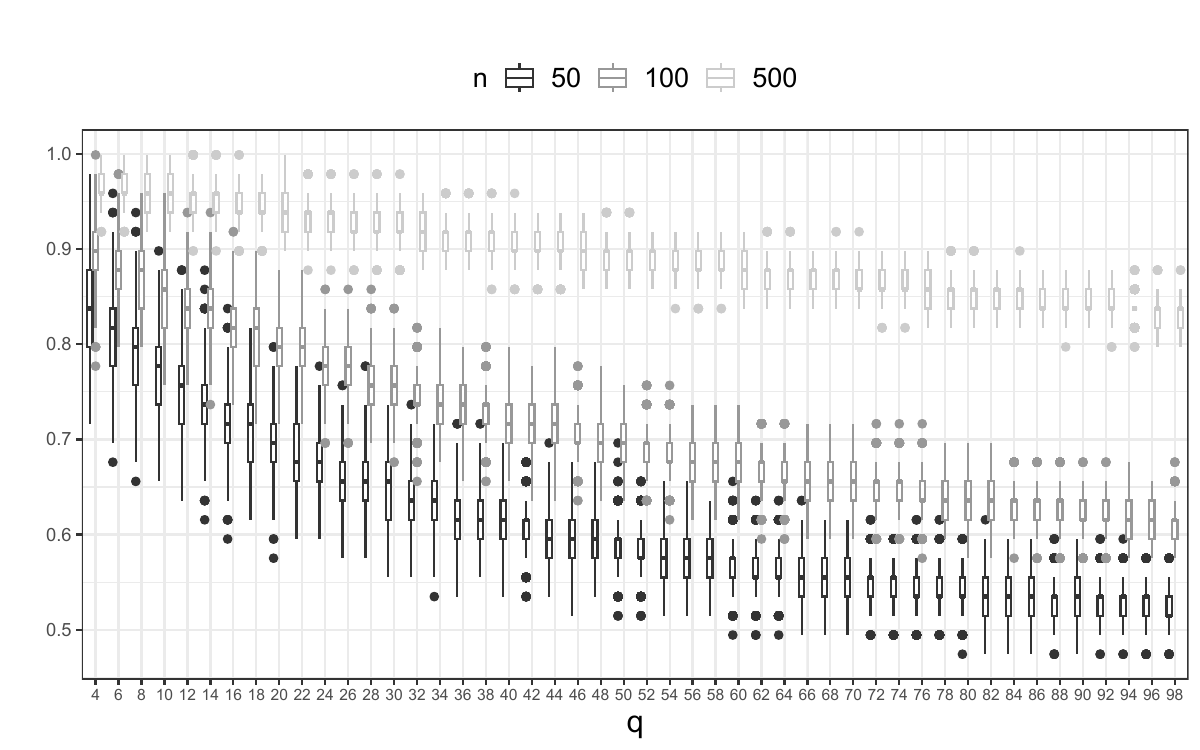}
\caption{Boxplots of selected penalty parameter $\omega_{n}$ via a $5$-fold cross-validation search on a grid of $50$ equidistant elements in $[0.01,0.999]$, for different values of $n$ and $q$.}
\label{fig: RidgeOmega}
\end{figure} 

Finally, Fig. \ref{fig: Ridge} also indicates that it would be interesting to study the behaviour of the dependence coefficients when $q \to \infty$, which can happen in multiple ways. In the first design for example, $d_{i}$ remains fixed for all $i = 1,\dots,k$, but $k \to \infty$. Because $\Phi$-dependence measures satisfy Axiom (A4) of \ref{App A} (see, e.g., \cite{Gijbels2023b}), we know that $\mathcal{D}_{t \ln(t)}, \mathcal{D}_{(\sqrt{t}-1)^{2}} \to 1$ when $k \to \infty$, which is probably why the mean squared error of $\widehat{\mathcal{D}}_{t \ln(t)},\widehat{\mathcal{D}}_{(\sqrt{t}-1)^{2}}$ first increases and afterwards decreases/becomes constant (in particular, the variance tends to $0$ when $k \to \infty$, and the bias decreases/becomes constant, see Fig. S1 and S2 of the Supplementary Material for plots of the bias and variance). Regarding the optimal transport dependence measures $\mathcal{D}_{1}$ and $\mathcal{D}_{2}$, we expect them (based on our simulations) to converge to $0$ when $k \to \infty$, probably because the strong normalization (denominator) grows faster than the numerator does. We see the variance of $\widehat{\mathcal{D}}_{r}$ decreasing in $q$, and the biases increasing in $q$. 

In the second design, $k = 2$ remains fixed, but $d_{1} = d_{2} \to \infty$. The optimal transport measures seem (in our simulations) to converge to zero again, while the $\Phi$-dependence measures remain constant. We leave a formal study of the behaviour of these dependence coefficients when $q \to \infty$ for further research.
\newline \\
\noindent
\textbf{(Adaptive/Group) lasso-type estimation}
\\

Recall the two different sparsity patterns discussed in Example \ref{ex4} for a $20 \times 20$ correlation matrix. By performing lasso-type estimation, we hope to recover zero entries (interpretability) on the one hand, and improve accuracy on the other hand. In particular, for the latter, we desire better performance of the plug-in estimator $\widehat{\mathcal{D}} = \mathcal{D}(\widehat{\mathbf{R}}_{\bigcdot,n})$ than of the non-penalized estimator $\mathcal{D}(\widehat{\mathbf{R}}_{n})$, where $\mathcal{D}$ is a certain correlation matrix based dependence coefficient. 

For the simulation, we consider the following correlation structures
\begin{itemize}
\item Scenario 1: We take $q = 20$ with $k = 7$, $d_{1} = \cdots = d_{6} = 3$ and $d_{7} = 2$ having the non-block sparse correlation structure given in the left plot of Fig. \ref{fig: SparseSettings}, with $32.5\%$ zeroes.
\item Scenario 2: We take $q = 20$ with $k = 7$, $d_{1} = \cdots = d_{6} = 3$ and $d_{7} = 2$ such that $\mathbf{X}_{1},\dots,\mathbf{X}_{6}$ have independence within and between each other, but all components of $\mathbf{X}_{7}$ are related and all dependent on all components of $\mathbf{X}_{1},\dots,\mathbf{X}_{6}$, i.e., the block sparse correlation structure of Example \ref{ex4}, right plot in Fig. \ref{fig: SparseSettings}, with $76.5\%$ zeroes. 
\item Scenario 3: We take $q = 70$ with $k = 4$, $d_{1} = d_{2} = d_{3} = 21$ and $d_{4} = 7$ such that there is only dependence with and within $\mathbf{X}_{4}$, yielding a block sparse correlation structure, with $79.7\%$ zeroes.
\end{itemize}
Each time, we generate $1000$ samples from a $\mathcal{N}_{q}(\mathbf{0}_{q},\mathbf{R})$ distribution (marginals are again assumed to be unknown), and compute the no penalty, lasso, adaptive lasso (of \cite{Fop2021} with $\omega_{n} = 1$ and tuning parameter $\rho_{n}$), scad, and group lasso estimator for $\mathbf{R}$. The considered sample sizes are $n = 50,100,500$ for Scenarios 1 and 2, and $n = 100,500$ for Scenario 3. Based on the $1000$ replications, we report on the average true positive rate (TPR), which we want to be close to one, and false positive rate (FPR), which we want to be close to zero. We also compute the empirical root mean squared error of $||\widehat{\mathbf{R}} - \mathbf{R}||_{\text{F}}/q$ (where $\widehat{\mathbf{R}}$ is the estimated correlation matrix in question), and empirical mean squared error of $\widehat{\mathcal{D}}$, with $\mathcal{D}$ the mutual information $\mathcal{D}_{t \ln(t)}$, Hellinger distance $\mathcal{D}_{(\sqrt{t}-1)^{2}}$, or one of the optimal transport dependence measures $\mathcal{D}_{1}$ or $\mathcal{D}_{2}$. 

For tuning $\omega_{n}$ (or $\rho_{n}$ in case of the adaptive lasso), we use the BIC criterion (see, e.g., \cite{Foygel2010} for the case of a precision matrix in Gaussian graphical models):
\begin{equation}\label{eq: BIC}
	\text{BIC} \left (\widehat{\boldsymbol{\Sigma}}_{\omega_{n}} \right ) = -n \left [ \ln \left |\widehat{\boldsymbol{\Sigma}}_{\omega_{n}} \right | + \text{tr} \left (\widehat{\boldsymbol{\Sigma}}_{\omega_{n}}^{-1} \widehat{\boldsymbol{\Sigma}}_{n} \right ) \right ] - \ln(n) \text{df}\left ( \widehat{\boldsymbol{\Sigma}}_{\omega_{n}} \right ),
\end{equation}
where $\widehat{\boldsymbol{\Sigma}}_{\omega_{n}}$ is the estimated candidate covariance matrix using the penalty parameter $\omega_{n}$. We want to maximize \eqref{eq: BIC} in $\omega_{n}$. We do this over an equidistant grid of $50$ values in $[0.01;0.6]$. The number $\text{df}(\widehat{\boldsymbol{\Sigma}}_{\omega_{n}})$ stands for the degrees of freedom, and is estimated for the (adaptive) lasso and scad by the number of non-zero entries in $\widehat{\boldsymbol{\Sigma}}_{\omega_{n}}$, not taking the elements under the diagonal into account. For the group lasso, the degrees of freedom can be estimated in a similar spirit as equation (23) in \cite{Chiquet2012}:

\begin{equation*}
\begin{split}
\text{df} \left ( \widehat{\boldsymbol{\Sigma}}_{\omega_{n}} \right ) = \sum_{\substack{i,m = 1 \\ m > i}}^{k} \mathds{1} \left (\left | \left |\widehat{\boldsymbol{\Sigma}}_{\omega_{n},im} \right | \right |_{\text{F}} > 0 \right ) & \left (1 + \frac{\left | \left |\widehat{\boldsymbol{\Sigma}}_{\omega_{n},im} \right | \right |_{\text{F}}}{\left | \left |\widehat{\boldsymbol{\Sigma}}_{n,im} \right | \right |_{\text{F}}} \left (d_{i}d_{m} - 1 \right ) \right ) \\ & + \sum_{i = 1}^{k} \mathds{1} \left (\left | \left |\boldsymbol{\Delta}_{i} * \widehat{\boldsymbol{\Sigma}}_{\omega_{n},ii} \right | \right |_{\text{F}} > 0 \right ) \left (1 + \frac{\left | \left |\boldsymbol{\Delta}_{i} * \widehat{\boldsymbol{\Sigma}}_{\omega_{n},ii}  \right | \right |_{\text{F}}}{\left | \left |\boldsymbol{\Delta}_{i} * \widehat{\boldsymbol{\Sigma}}_{n,ii}  \right | \right |_{\text{F}}} \left (\frac{d_{i}(d_{i}-1)}{2} - 1 \right ) \right ) + q,
\end{split}
\end{equation*}
where $\widehat{\boldsymbol{\Sigma}}_{\omega_{n},im}$ is the $(i,m)$'th block of $\widehat{\boldsymbol{\Sigma}}_{\omega_{n}}$, similarly for $\widehat{\boldsymbol{\Sigma}}_{n,im}$, and $\boldsymbol{\Delta}_{ii} \in \mathbb{R}^{d_{i} \times d_{i}}$ is a matrix with ones as off-diagonal elements and zeroes on the diagonal.
The results are summarized in Table \ref{tab: table1}.

In Scenario 1 ($q = 20$, no block sparsity, $32.5\%$ zeroes), there is no group sparsity, which results in low TPR (and also low FPR, since only few entries get shrunk to zero) for the group lasso estimator. The other penalization techniques are, as expected, preferred for recovering zeroes. Clearly, any type of considered penalization yields more accurate estimation of $\mathbf{R}$ in Frobenius norm than in case no penalty is used. However, this does not necessarily imply better estimation of the dependence coefficients, especially for lasso and scad. For good estimation of these, it is important not to lose sight of the $67.5\%$ non-zero entries, which is why the adaptive lasso performs really well. Note that the mutual information is estimated very well in the non-penalized case, but this is mainly due to the fact that the true value equals $0.994$, which is very close to one, being an effect of the dimension (recall also Fig. \ref{fig: Ridge}). All non-penalized mutual information estimates are close to one because of a relatively large $q$, yielding low estimation error.

In Scenario 2 ($q = 20$, block sparsity, $76.5\%$ zeroes), all penalization techniques perform well in identifying zeroes. For obtaining both high TPR and low FPR, the group lasso performs slightly better than lasso and scad. Also, when the focus is on estimating $\mathbf{R}$, penalization is clearly beneficial, and the group lasso outperforms the other techniques. Regarding the dependence coefficients, we see that, for the optimal transport measures, lasso and scad give improvement compared to using no penalty, especially for smaller sample sizes. The error of the $\Phi$-dependence
\begin{table}[h] \centering 
\caption{True positive rate (TPR), false positive rate (FPR), and empirical mean squared error for the estimated correlation matrix and corresponding plug-in estimators of dependence coefficients based on $1000$ replications in different scenarios, using different penalties.}
\label{tab: table1}
\begin{tabular*}{\linewidth}{@{\extracolsep{\fill}}cccccccccc@{\extracolsep{\fill}}}
\hline 
\multicolumn{10}{c}{\rule{0pt}{3ex}\textbf{Scenario 1 ($q = 20$, no block sparsity, $32.5\%$ zeroes)}} \\ \rule{0pt}{3ex}  & n & TPR & FPR & $||\widehat{\mathbf{R}}-\mathbf{R}||_{\text{F}}/q $ & $\widehat{\mathcal{D}}_{t \ln(t)}$ & $\widehat{\mathcal{D}}_{(\sqrt{t}-1)^{2}}$ & $\widehat{\mathcal{D}}_{1} $ & $\widehat{\mathcal{D}}_{2}$ & \vspace*{0.1cm} \\  \hdashline  \rule{0pt}{4ex}
& $50$ & $\backslash$ & $\backslash$ & $0.137$ & $3.993 \cdot 10^{-4}$ & $0.109$ & $0.058$ & $0.057$ \\ 
no penalty & $100$ & $\backslash$ & $\backslash$ & $0.096$ & $2.898 \cdot 10^{-4}$ & $0.044$ & $0.026$ & $0.026$ \\
& $500$ & $\backslash$ & $\backslash$ & $0.043$ & $8.331 \cdot 10^{-5}$ & $0.009$ & $0.005$ & \hspace{0.02cm} $0.005$ \vspace*{0.1cm} \\ 
\hdashline  \rule{0pt}{4ex}
& $50$ & $0.831$ & $0.592$ & $0.105$ & $0.209$ & $0.327$ & $0.054$ & $0.052$ \\ 
lasso & $100$ & $0.724$ & $0.355$ & $0.077$ & $0.005$ & $0.153$ & $0.032$ & $0.031$ \\
& $500$ & $0.554$ & $0.121$ & $0.039$ & $0.001$ & $0.054$ & $0.014$ & \hspace{0.02cm} $0.014$ \vspace*{0.1cm} \\ 
\hdashline  \rule{0pt}{4ex}
& $50$ & $0.818$ & $0.565$ & $0.122$ & $0.002$ & $0.075$ & $0.025$ & $0.025$ \\ 
adaptive lasso & $100$ & $0.848$ & $0.480$ & $0.092$ & $2.640 \cdot 10^{-4}$ & $0.030$ & $0.012$ & $0.013$ \\
& $500$ & $0.932$ & $0.303$ & $0.044$ & $8.041 \cdot 10^{-5}$ & $0.010$ & $0.003$ & \hspace{0.02cm} $0.004$ \vspace*{0.1cm} \\ 
\hdashline  \rule{0pt}{4ex}
& $50$ & $0.831$ & $0.592$ & $0.105$ & $0.209$ & $0.327$ & $0.054$ & $0.053$ \\ 
scad & $100$ & $0.744$ & $0.373$ & $0.078$ & $0.006$ & $0.165$ & $0.035$ & $0.034$ \\
& $500$ & $0.432$ & $0.091$ & $0.039$ & $0.001$ & $0.054$ & $0.014$ & \hspace{0.02cm} $0.014$ \vspace*{0.1cm} \\ 
\hdashline  \rule{0pt}{4ex}
& $50$ & $0.074$ & $0.038$ & $0.091$ & $0.045$ & $0.116$ & $0.024$ & $0.023$ \\ 
group lasso & $100$ & $0.062$ & $0.026$ & $0.073$ & $0.004$ & $0.130$ & $0.029$ & $0.028$ \\
& $500$ & $0.020$ & $0.006$ & $0.042$ & $0.001$ & $0.067$ & $0.018$ & $0.018$  \\ \hline 
\multicolumn{10}{c}{\rule{0pt}{3ex} \textbf{Scenario 2 ($q = 20$, block sparsity, $76.5\%$ zeroes)}} \\ \rule{0pt}{3ex} & n & TPR & FPR & $||\widehat{\mathbf{R}}-\mathbf{R}||_{\text{F}}/q $ & $\widehat{\mathcal{D}}_{t \ln(t)}$ & $\widehat{\mathcal{D}}_{(\sqrt{t}-1)^{2}}$ & $\widehat{\mathcal{D}}_{1} $ & $\widehat{\mathcal{D}}_{2}$ & \vspace*{0.1cm} \\  \hdashline \rule{0pt}{4ex}
& $50$ & $\backslash$ & $\backslash$ & $0.138$ & $0.006$ & $0.178$ & $0.065$ & $0.064$ \\ 
no penalty & $100$ & $\backslash$ & $\backslash$ & $0.097$ & $0.004$ & $0.077$ & $0.030$ & $0.030$ \\
& $500$ & $\backslash$ & $\backslash$ & $0.043$ & $0.001$ & $0.016$ & $0.006$ & \hspace{0.02cm} $0.006$  \vspace*{0.1cm} \\ 
\hdashline  \rule{0pt}{4ex}
& $50$ & $0.940$ & $0.340$ & $0.058$ & $0.108$ & $0.246$ & $0.026$ & $0.026$ \\ 
lasso & $100$ & $0.903$ & $0.135$ & $0.033$ & $0.020$ & $0.139$ & $0.013$ & $0.013$ \\
& $500$ & $0.888$ & $0.048$ & $0.013$ & $0.004$ & $0.048$ & $0.005$ & \hspace{0.02cm} $0.005$ \vspace*{0.1cm} \\ 
\hdashline  \rule{0pt}{4ex}
& $50$ & $0.858$ & $0.431$ & $0.101$ & $0.010$ & $0.103$ & $0.024$ & $0.024$ \\ 
adaptive lasso & $100$ & $0.860$ & $0.329$ & $0.076$ & $0.004$ & $0.053$ & $0.013$ & $0.013$ \\
& $500$ & $0.900$ & $0.164$ & $0.032$ & $0.001$ & $0.013$ & $0.003$ & \hspace{0.02cm} $0.003$ \vspace*{0.1cm} \\ 
\hdashline  \rule{0pt}{4ex}
& $50$ & $0.940$ & $0.340$ & $0.058$ & $0.108$ & $0.246$ & $0.026$ & $0.026$ \\ 
scad & $100$ & $0.906$ & $0.137$ & $0.033$ & $0.021$ & $0.142$ & $0.
013$ & $0.013$ \\
& $500$ & $0.846$ & $0.051$ & $0.016$ & $0.005$ & $0.062$ & $0.006$ & \hspace{0.02cm} $0
.006 $  \vspace*{0.1cm} \\ 
\hdashline  \rule{0pt}{4ex}
& $50$ & $0.889$ & $0.011$ & $0.035$ & $0.021$ & $0.138$ & $0.013$ & $0.013$ \\ 
group lasso & $100$ & $0.934$ & $0.012$ & $0.021$ & $0.011$ & $0.108$ & $0.010$ & $0.010$ \\
& $500$ & $0.965$ & $0.013$ & $0.008$ & $0.003$ & $0.044$ & $0.004$ & \hspace{0.02cm} $0.004$ \vspace*{0.1cm} \\ \hline 

\end{tabular*}
\end{table}
\clearpage 
\begin{table}[ht!] \centering 
\begin{tabular*}{\linewidth}{@{\extracolsep{\fill}}cccccccccc@{\extracolsep{\fill}}} \hline 
\multicolumn{10}{c}{\rule{0pt}{3ex} \textbf{Scenario 3 ($q = 70$, block sparsity, $79.7\%$ zeroes)}} \\ \rule{0pt}{3ex} & n & TPR & FPR & $||\widehat{\mathbf{R}}-\mathbf{R}||_{\text{F}}/q $ & $\widehat{\mathcal{D}}_{t \ln(t)}$ & $\widehat{\mathcal{D}}_{(\sqrt{t}-1)^{2}}$ & $\widehat{\mathcal{D}}_{1} $ & $\widehat{\mathcal{D}}_{2}$ & \vspace*{0.1cm} \\  \hdashline  \rule{0pt}{4ex}
\multirow{2}{*}{no penalty} & $100$ & $\backslash$ & $\backslash$ & $0.099$ & $2.344 \cdot 10^{-9}$ & $0.021$ & $0.138$ & $0.125$ \\ & $500$ & $\backslash$ & $\backslash$ & $0.044$ & $2.239 \cdot 10^{-9}$ & $0.007$ & $0.023$ & \hspace{0.02cm} $0.022$  \vspace*{0.1cm} \\ 
\hdashline  \rule{0pt}{4ex} 
\multirow{2}{*}{lasso} & $100$ & $0.993$ & $0.757$ & $0.048$ & $0.532$ & $0.629$ & $0.073$ & $0.072$ \\ & $500$ & $0.874$ & $0.087$ & $0.015$ & $2.520 \cdot 10^{-8}$ & $0.015$ & $0.007$ & \hspace{0.02cm} $0.008$  \vspace*{0.1cm} \\ 
\hdashline  \rule{0pt}{4ex}
\multirow{2}{*}{adaptive lasso} & $100$ & $0.955$ & $0.740$ & $0.058$ & $0.001$ & $0.240$ & $0.033$ & $0.033$ \\ & $500$ & $0.838$ & $0.323$ & $0.038$ & $1.887 \cdot 10^{-9}$ & $0.005$ & $0.012$ & \hspace{0.02cm} $0.011$  \vspace*{0.1cm} \\ 
\hdashline  \rule{0pt}{4ex}
\multirow{2}{*}{scad} & $100$ & $0.993$ & $0.757$ & $0.048$ & $0.532$ & $0.629$ & $0.073$ & $0.072$ \\ & $500$ & $0.883$ & $0.092$ & $0.015$ & $3.969 \cdot 10^{-8}$ & $0.019$ & $0.009$ & \hspace{0.02cm} $0.009$  \vspace*{0.1cm} \\ 
\hdashline  \rule{0pt}{4ex}
\multirow{2}{*}{group lasso} & $100$ & $0.914$ & $0$ & $0.023$ & $4.331 \cdot 10^{-8}$ & $0.013$ & $0.009$ & $0.009$ \\ & $500$ & $0.970$ & $0$ & $0.007$ & $8.943 \cdot 10^{-9}$ & $0.009$ & $0.004$ & \hspace{0.02cm} $0.003$  \vspace*{0.1cm} \\ 
\hline
\end{tabular*}
\end{table} \noindent 
estimates is rather low in case no penalty is used, which is again an effect of the dimension. Aside from this, the group lasso or the adaptive lasso (especially for larger sample sizes) gives the lowest error.

In Scenario 3 ($q = 70$, block sparsity, $79.7\%$ zeroes), the group lasso is again desirable for exploiting the sparsity structure, and even achieves zero FPR. For accurate estimation of $\mathbf{R}$ in Frobenius norm, the group lasso also performs best. Most accurate estimation of dependence (except for the mutual information, where no penalty performs best because the true value is again very close to one) is obtained by the group lasso when $n = 100$, and by the adaptive lasso when $n = 500$ is large compared to $q$.

In conclusion, when the true correlation matrix is sparse, interpretability can be enhanced by using a penalty that is able to completely shrink entries to zero. The group lasso is preferred for obtaining both good TPR and FPR when this sparsity is at the block level (Scenarios 2 and 3), and also performs well in estimating dependence in such cases, particularly when $n$ is rather small. The true $\mathbf{R}$ is estimated more accurately (compared to using no penalty) in Frobenius norm when using any penalty and in any scenario, but this does not necessarily result in better estimation of dependence. Especially when there are still quite some non-zeroes (Scenario 1), but also when $n$ is rather large, adaptive lasso is recommended for good accuracy of the estimated dependence coefficients. The $\Phi$-dependence measures are estimated with rather low error when using no penalty, but this is because they attain their upper bound of $1$ rather quickly when the dimension increases.

\section{Real data applications}\label{sec:6}
In Section \ref{subsec: 6.1}, we look into an application of optimal transport dependence measures between possibly more than two random vectors to sensory analysis. In Section \ref{subsec: 6.2}, we illustrate how these measures, together with the considered penalization techniques, can be useful in finding clusters among speech signal attributes used for detecting Parkinson's disease. 

\subsection{Application to sensory analysis}\label{subsec: 6.1}

Consider a caterer who wants to sell eight different smoothies, say $S_{1},\dots,S_{8}$, on an event, and is looking for three employees willing to take up this job. A total of $24$ people, say $P_{1},\dots,P_{24}$, show up for this job opportunity, all equally qualified, and the caterer is looking for a fair way to pick three candidates. Every candidate is asked to taste each of the eight smoothies, and is given a sheet of paper in order to position the different smoothies, knowing that the closer certain smoothies are to each other, the more similar they are considered by the individual. For example, according to candidate $P_{j}$ in Fig. \ref{fig: smoothie}, the smoothies $S_{1},S_{2},S_{3}$ and $S_{4}$ are similar, but quite different from the similarly tasting smoothies $S_{5},S_{6}$ and $S_{7}$, and none of them resembles $S_{8}$.
\begin{figure}[h!] 
\includegraphics[scale = 0.45]{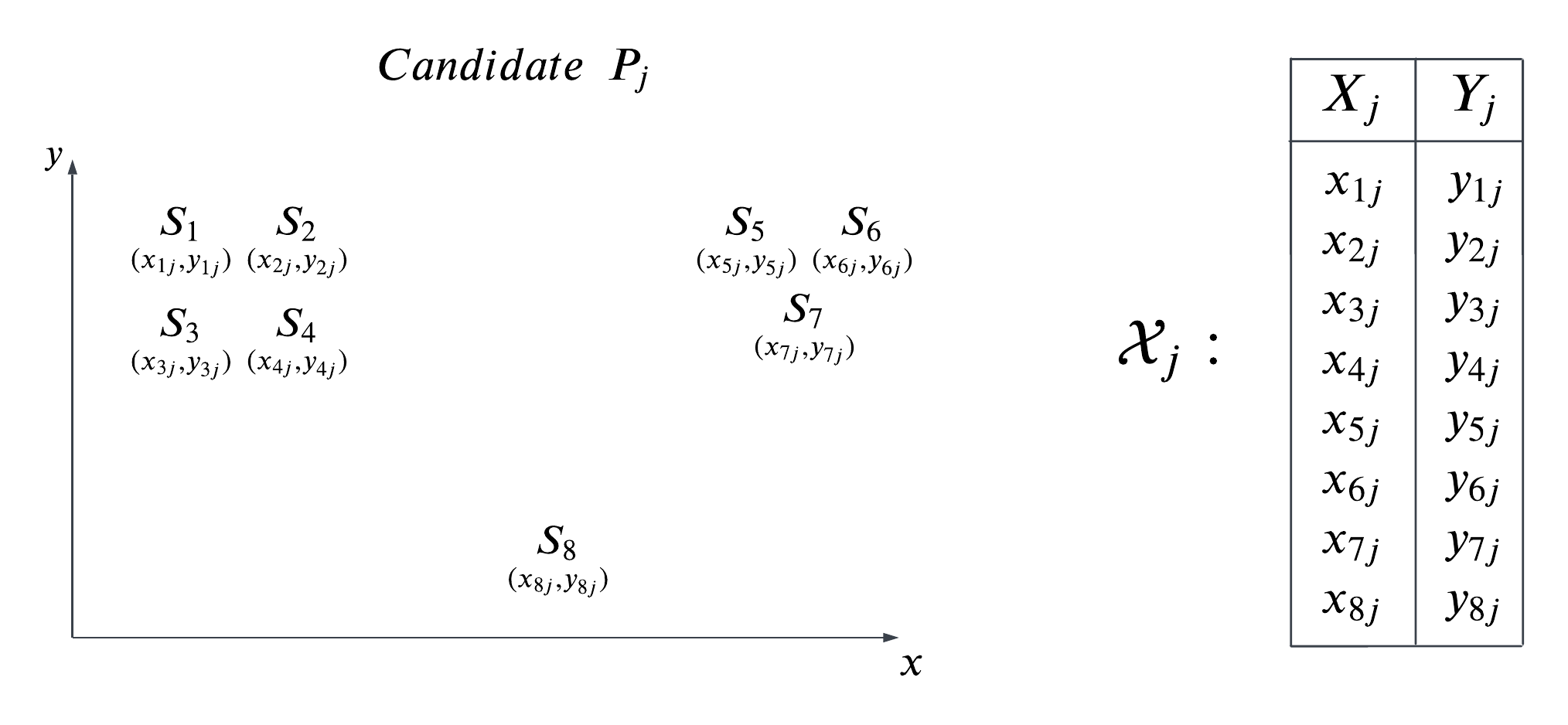}
\caption{Example of smoothie similarity rating by candidate $P_{j}$ on a sheet of paper and corresponding dataset $\mathcal{X}_{j}$.}
\label{fig: smoothie}
\end{figure} 
As such, the caterer acquires $24$ datasets, say $\boldsymbol{\mathcal{X}}_{j} \in \mathbb{R}^{8 \times 2}$ for $j = 1,\dots,24$, containing the smoothies as rows and the $X$-$Y$ coordinates on the sheet of paper for person $j$, denoted as a random vector $(X_{j},Y_{j})$, as columns, representing the smoothie similarities of each candidate. So, the dataset $\boldsymbol{\mathcal{X}}_{j}$ contains a sample from $(X_{j},Y_{j})$ of size eight, denoted as $(x_{ij},y_{ij})$ for $i = 1,\dots,8$. The data is available in the \textsf{R} package SensoMineR (\cite{Le2008}).
The criterion based on which three employees are picked consists of finding the three individuals that have the least similar spatial configurations, meaning three very diversified tastes (in the hope of not selling only a few smoothies because of prepossessed preferences by the sellers).

Typically, see, e.g., \cite{Llobell2020}, the similarity between two configurations is measured by the RV coefficient 
\begin{equation}\label{eq: RVc}
\text{RV}(\boldsymbol{\mathcal{X}}_{j_{1}},\boldsymbol{\mathcal{X}}_{j_{2}}) = \frac{\text{tr} \left (\boldsymbol{\mathcal{X}}_{j_{1}} \boldsymbol{\mathcal{X}}_{j_{1}}^{\text{T}} \boldsymbol{\mathcal{X}}_{j_{2}} \boldsymbol{\mathcal{X}}_{j_{2}}^{\text{T}} \right )}{\sqrt{\text{tr}\left \{\left (\boldsymbol{\mathcal{X}}_{j_{1}} \boldsymbol{\mathcal{X}}_{j_{1}}^{\text{T}} \right )^{2} \right \} \text{tr}\left \{\left (\boldsymbol{\mathcal{X}}_{j_{2}} \boldsymbol{\mathcal{X}}_{j_{2}}^{\text{T}} \right )^{2} \right \} }},
\end{equation}
and for three configurations one can take, e.g., the average of all pairwise RV coefficients. Yet, pairwise coefficients feel unnatural and it would be better to compute a trivariate vector similarity. Note that \eqref{eq: RVc} is actually the RV coefficient between two random vectors $(X_{j_{1}},Y_{j_{1}})$ and $(X_{j_{2}},Y_{j_{2}})$ of size two having joint, empirical covariance matrix
\begin{equation*}
\begin{pmatrix}
\boldsymbol{\mathcal{X}}_{j_{1}}^{\text{T}} \boldsymbol{\mathcal{X}}_{j_{1}} & \boldsymbol{\mathcal{X}}_{j_{1}}^{\text{T}} \boldsymbol{\mathcal{X}}_{j_{2}} \vspace{0.2cm} \\ \boldsymbol{\mathcal{X}}_{j_{2}}^{\text{T}} \boldsymbol{\mathcal{X}}_{j_{1}} & \boldsymbol{\mathcal{X}}_{j_{2}}^{\text{T}} \boldsymbol{\mathcal{X}}_{j_{2}}
\end{pmatrix} \in \mathbb{R}^{4 \times 4},
\end{equation*}
and we get a similar (larger) block covariance matrix in $\mathbb{R}^{2m \times 2m}$ when taking $m$ individuals into account. 

Since not restricted to two random vectors anymore, one can also opt for an (estimated) optimal transport dependence coefficient $\mathcal{D}_{1}$ or $\mathcal{D}_{2}$ between three vectors of size two for measuring the similarity between three individual spatial configurations. Note that the dispersion of the coordinates might differ among the individuals (some might use the entire sheet, while others only use the right corner), but since we use the normal scores rank correlations, we do not need any centering or scaling of the data. 
Pairwise scatterplots of the normal scores of the $X$-$Y$ coordinates of the $24$ people are shown in Fig. S3 of the Supplementary Material, indicating that dependencies are mainly correlation based, i.e., a Gaussian copula model is suitable for modelling the dependencies. When zooming in on candidates $P_{12},P_{13},P_{18}$ and $P_{20}$, we get the pairwise scatterplots given in Fig. \ref{fig: Gcopulaver1}. From this, we expect for example that candidates $P_{12}$ and $P_{13}$ have quite independent smoothies preferences, while candidates $P_{18}$ and $P_{20}$ have rather strong correlations between their coordinates.

\begin{table}[h] \centering 
        \caption{Arrangement of two and three candidates according to largest estimated similarity $\widehat{\mathcal{D}}_{r}$ for $r = 1,2$. The first two and last two are shown.}
        \label{tab: table2}
        \begin{tabular*}{\linewidth}{@{\extracolsep{\fill}}cccccccc@{\extracolsep{\fill}}}
        \hline \\ \multicolumn{2}{c}{two based on $\widehat{\mathcal{D}}_{1}$} & \multicolumn{2}{c}{two based on $\widehat{\mathcal{D}}_{2}$} & \multicolumn{2}{c}{three based on $\widehat{\mathcal{D}}_{1}$} & \multicolumn{2}{c}{three based on $\widehat{\mathcal{D}}_{2}$} \\ \cdashline{1-2} \cdashline{3-4} \cdashline{5-6} \cdashline{7-8} \rule{0pt}{4ex}
        $\widehat{\mathcal{D}}_{1}(P_{18},P_{20})$ & $0.561$ & $\widehat{\mathcal{D}}_{2}(P_{18},P_{20})$ & $0.561$ & \hspace{0.05cm} $\widehat{\mathcal{D}}_{1}(P_{15},P_{18},P_{20})$ & $0.585$ & $\widehat{\mathcal{D}}_{2}(P_{15},P_{18},P_{20})$ & $0.595$ \\
        $\widehat{\mathcal{D}}_{1}(P_{9},P_{23})$ & $0.542$ & $\widehat{\mathcal{D}}_{2}(P_{15},P_{20})$ & $0.521$ & $\widehat{\mathcal{D}}_{1}(P_{9},P_{10},P_{23})$ & $0.550$ & $\widehat{\mathcal{D}}_{2}(P_{10},P_{18},P_{23})$ & $0.561$ \\ \vdots & \vdots & \vdots & \vdots & \vdots & \vdots & \vdots & \vdots \\
        \hspace{0.1cm} $\widehat{\mathcal{D}}_{1}(P_{12},P_{19})$ & $0.020$ & \hspace{0.05cm} $\widehat{\mathcal{D}}_{2}(P_{12},P_{19})$ & $0.016$ & $\widehat{\mathcal{D}}_{1}(P_{2},P_{3},P_{14})$ & $0.080$ & $\widehat{\mathcal{D}}_{2}(P_{2},P_{12},P_{19})$ & $0.075$ \\ \hspace{0.1cm}
        $\widehat{\mathcal{D}}_{1}(P_{12},P_{13})$ & $0.015$ & $\hspace{0.08cm} \widehat{\mathcal{D}}_{2}(P_{12},P_{13})$ & $0.013$& \hspace{0.09cm} $\widehat{\mathcal{D}}_{1}(P_{12},P_{13},P_{21})$ & $0.069$ & \hspace{0.08cm} $\widehat{\mathcal{D}}_{2}(P_{12},P_{13},P_{21})$ & $0.074$ \\ \hline 
        \end{tabular*}
\end{table}

We denote $\widehat{\mathcal{D}}_{r}(P_{j_{1}},P_{j_{2}})$ and $\widehat{\mathcal{D}}_{r}(P_{j_{1}},P_{j_{2}},P_{j_{3}})$ for the estimated similarity between two candidates $P_{j_{1}},P_{j_{2}}$ or three candidates $P_{j_{1}},P_{j_{2}},P_{j_{3}}$ for $r = 1,2$. Recall that these are actually dependencies between $2$ and $3$ random vectors of size $2$ respectively, i.e., $q = 4$ with $k = 2, d_{1} = d_{2} = 2$, or $q = 6$ with $k = 3, d_{1} = d_{2} = d_{3} = 2$ respectively, estimated based on a sample of size $n = 8$.

Table \ref{tab: table2} shows the two strongest and two weakest couples or triplets according to $\widehat{\mathcal{D}}_{1}$ or $\widehat{\mathcal{D}}_{2}$. The caterer will definitely hire candidate $P_{12},P_{13}$ and $P_{21}$. In \cite{Llobell2020}, they cluster the candidates using the clustatis method, bringing forward three classes of individuals, see their Fig. 5. We see that $P_{12}$ belongs to class $2$, while $P_{13}$ belongs to class $3$ and $P_{21}$ to class $1$, also indicating their diversified smoothie similarity pattern. The optimal transport dependence measures give an unequivocal ordering of patterns based on similarities that go beyond two random vectors. 

\noindent
Note that we can also switch the role of the smoothies and the individuals, i.e., construct $8$ datasets in $\mathbb{R}^{24 \times 2}$, and similarly look at the dependence between smoothies. 
Doing so, both $\widehat{\mathcal{D}}_{1}$ $( = 0.0336)$ and $\widehat{\mathcal{D}}_{2}$ $(= 0.0340)$ agree that $S_{6},S_{7}$ and $S_{8}$ are the least similar among all possible triplets. They are respectively called \textit{Casino\_PBC}, \textit{Innocent\_SB} and \textit{Carrefour\_SB}, and the biplots (based on various sensometrics methods) in Fig. 2. of \cite{Llobell2020} indeed also reveal angels between these smoothies that are close to $90^{\circ}$. So, $S_{6},S_{7}$ and $S_{8}$ would be a good choice if the caterer wanted to limit his smoothie supply to three flavours that still have a satisfactory amount of diversity.

\begin{figure}
\centering
\begin{minipage}{.5\textwidth}
  \centering \hspace{-0.9cm} 
  \includegraphics[width= 7.9cm, height = 7.8cm]{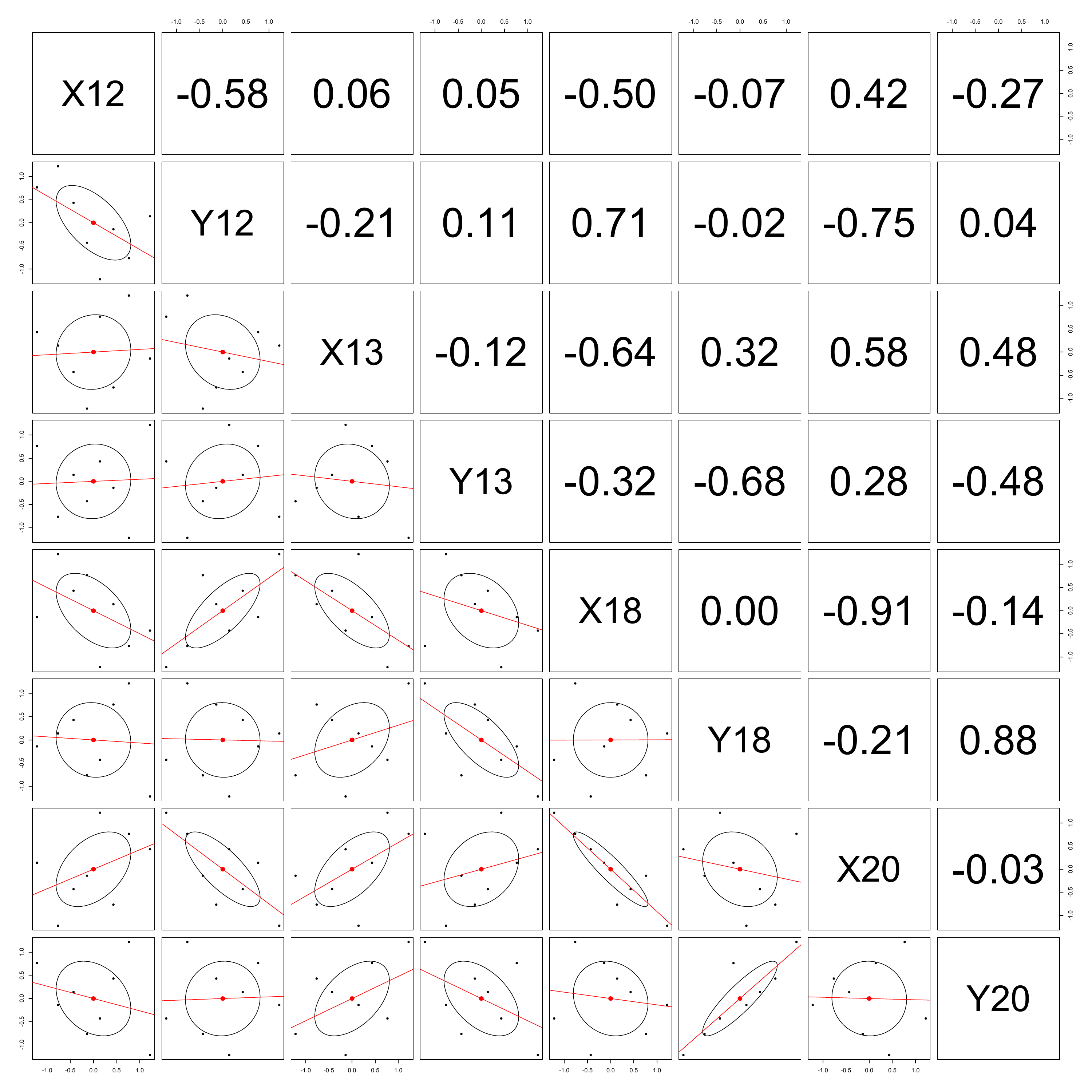}
  \captionof{figure}{Pairwise scatterplots of normal scores of $X$-$Y$ coordinates of \newline candidates $P_{12},P_{13},P_{18}$ and $P_{20}$ of the smoothies dataset.}
  \label{fig: Gcopulaver1}
\end{minipage}%
\begin{minipage}{.5\textwidth}
  \centering \hspace*{-0.7cm}
  \includegraphics[width= 8.9cm, height = 7.7cm]{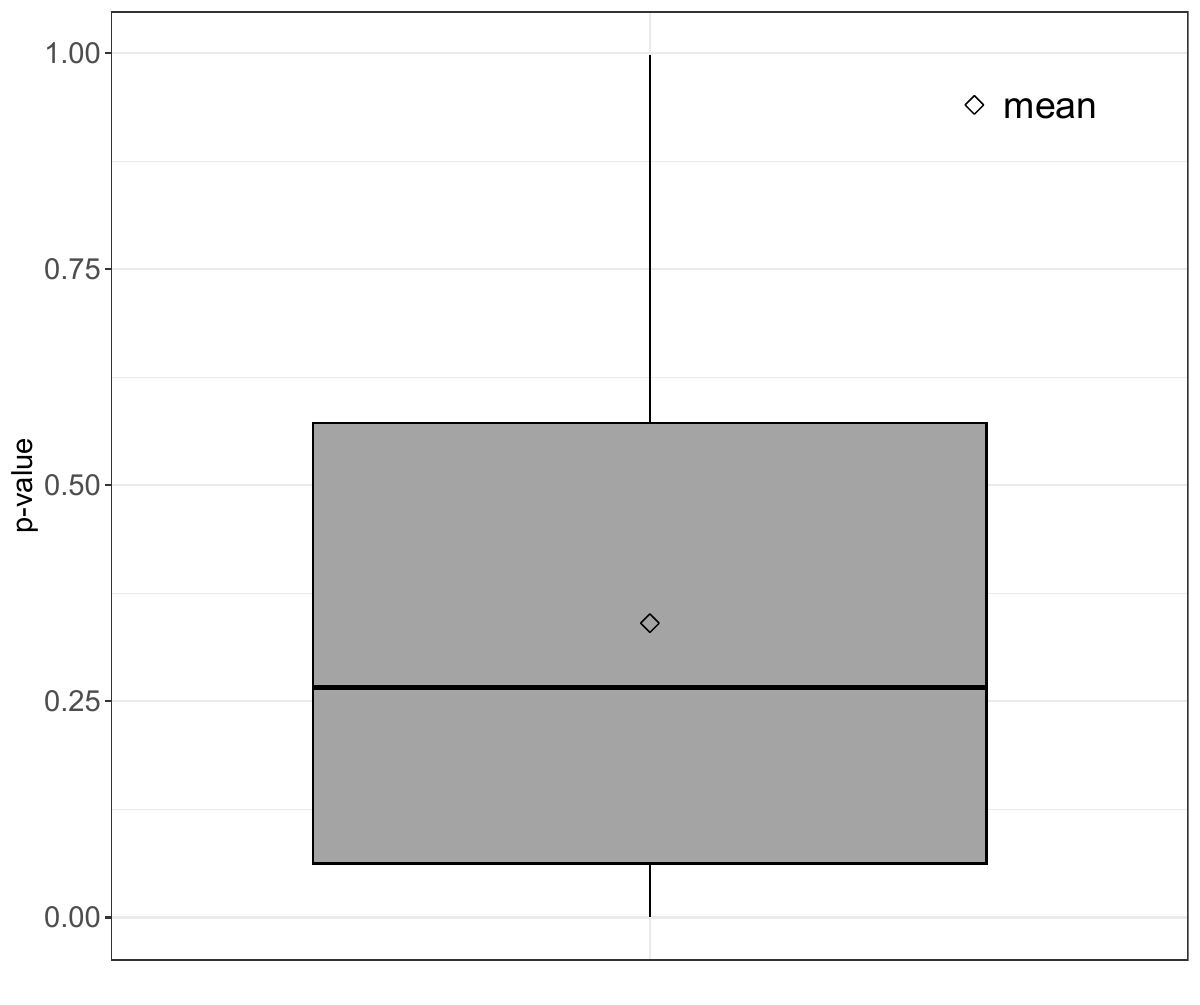}
  \captionof{figure}{Boxplot of $p$-values of pairwise goodness-of-fit tests for Gaussian copula on $91$ dysphonia measures of the LSVT dataset.}
  \label{fig: Gcopulaver2}
\end{minipage}
\end{figure}

\subsection{Clustering dysphonia measures}\label{subsec: 6.2}

In a second real data application, we study the LSVT voice rehabilitation dataset of \cite{Tsanas2014}, freely accessible at the UCI Machine Learning Repository (\url{https://archive.ics.uci.edu/dataset/282/lsvt+voice+rehabilitation}). Parkinson's disease frequently leads to vocal impairment, the extent of which can be assessed using sustained vowel phonations. In particular, the sustained vowel ``ahh..." (denoted as /a/) is typically studied. Next, dysphonia measures are used to extract clinically useful information from speech signals. The dataset consists of $q = 310$ such measurements on $n = 126$ phonations. As mentioned in Section C of \cite{Tsanas2014}, several of these dysphonia measures are similar (algorithmic variations of the same basic ideas), e.g., there are many jitter (a measure of cycle-to-cycle variation in frequency) and shimmer (a measure of cycle-to-cycle variation in amplitude) variants. Hence, there is a large amount of redundancy among the attributes, which could worsen the performance of supervised learning (like, e.g., a classifier as studied in \cite{Tsanas2014}), and some kind of preliminary feature selection is recommended.

There are indeed many dysphonia measures that exhibit a (very) strong sample normal scores rank correlation, see Fig. S4 of the Supplementary Material. In order to eliminate strong detrimental redundancies (that closely approach singularity), we iteratively search (pairwise) for attributes that have a normal scores rank correlation larger than $0.8$ in absolute value, and each time discard one of them. Remaining are $91$ dysphonia measures, whose (non-singular) sample normal scores rank correlation matrix is given in the left panel of Fig. S5 of the Supplementary Material.

For testing the adequacy of a Gaussian copula for modelling (at least pairwise) dependencies, we consider the test based on the statistic $S_{n}$ given in equation (2) of \cite{Genest2009}, where we test the null hypothesis that the copula between a pair of dysphonia measures is Gaussian, and compute the $p$-value based on $1000$ bootstrap samples. In Fig. \ref{fig: Gcopulaver2}, a boxplot of $p$-values of pairwise (that is between all $91 \cdot 90 / 2 = 4095$ pairs of variables) goodness-of-fit tests is shown.  Based on the boxplot, we see that is it reasonable to assume (at least pairwise) Gaussian dependencies.

We still have quite a large amount of attributes compared to the sample size of $n = 126$, and we still see some quite large correlations between several first and last dysphonia measures (right upper corner of left panel in Fig. S5), so further feature selection is desired. With this objective in mind, we decide to cluster the remaining $91$ attributes in order to get an idea about how they can be divided into properly separated groups that have rather strong similarity within.

\begin{figure}[h!] 
\includegraphics[scale = 0.65]{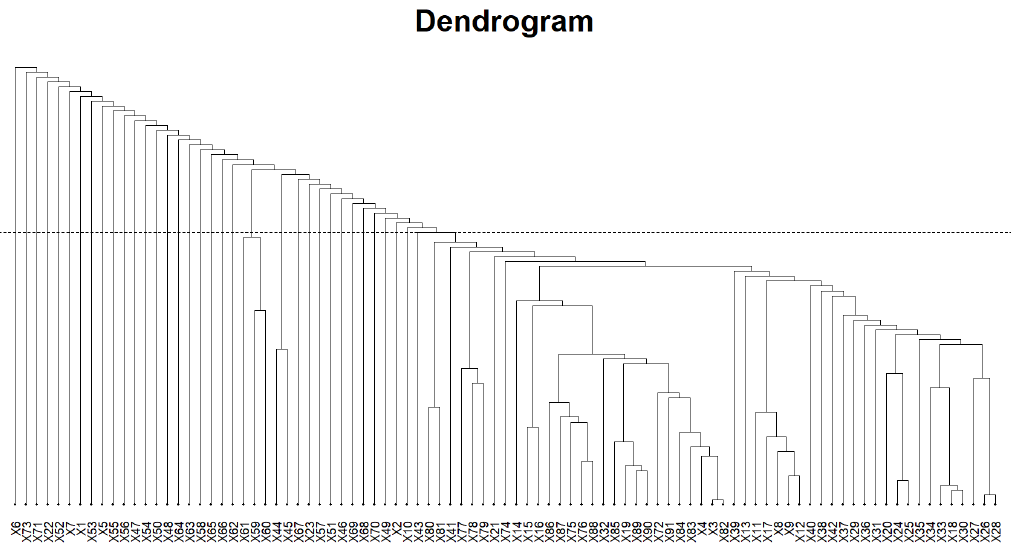}
\includegraphics[scale = 0.45]{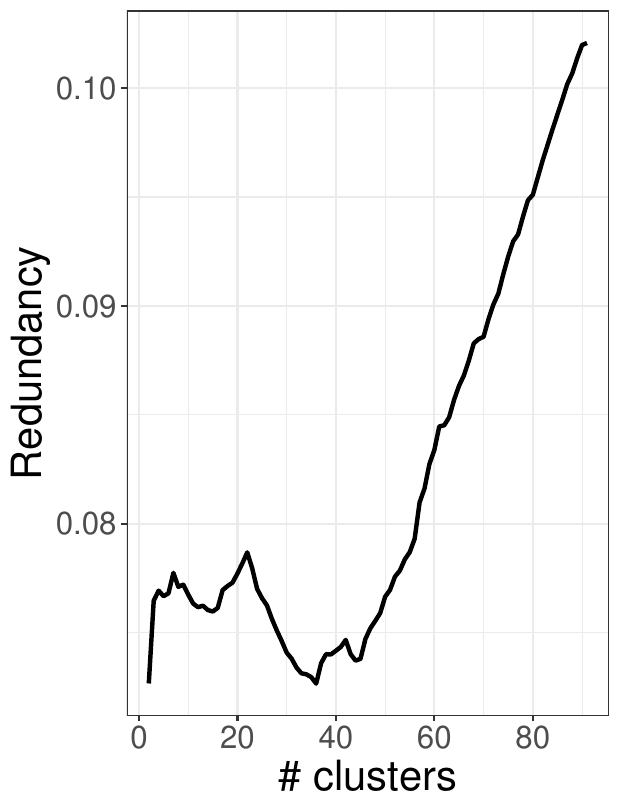}
\caption{Dendrogram (left) and redundancy (right) for the clustering of $91$ dysphonia measures.}
\label{fig: redun}
\end{figure} 

In particular, we opt for agglomerative hierarchical variable clustering based on similarity measures (see, e.g., \cite{Fuchs2021} and references within for a survey), where the main task is to measure the similarity, say $\mathcal{D}(\mathbb{X}_{1},\mathbb{X}_{2})$ (or alternatively dissimilarity), between two groups of variables $\mathbb{X}_{1}$ and $\mathbb{X}_{2}$. Denote $X_{1},\dots,X_{91}$ for the $91$ dysphonia measures, which are initially considered to be clusters on their own. In the first step, the two variables $X_{i}$ and $X_{j}$ (for certain $i,j \in \{1,\dots,91\}$ with $i \neq j$) that exhibit the largest similarity according to $\mathcal{D}(\{X_{i}\},\{X_{j}\})$, i.e., taking $\mathbb{X}_{1} = \{X_{i}\}$ and $\mathbb{X}_{2} = \{X_{j}\}$, are merged together forming one single cluster. Next, all similarities between the current clusters are again computed, and the two clusters showing the largest similarity are merged. One keeps repeating this until only one big cluster (composed of all the $91$ attributes) remains, yielding a total of $91$ partitions of the attribute space. In general, the main task is thus to compute (for $\mathbb{X}_{1} = \{X_{i_{1}},\dots,X_{i_{m}} \}$ and $\mathbb{X}_{2} = \{X_{j_{1}},\dots,X_{j_{r}} \}$)
\begin{equation}\label{eq: msim}
\mathcal{D} \left (\mathbb{X}_{1},\mathbb{X}_{2} \right ) = \mathcal{D} \left (\{X_{i_{1}},\dots,X_{i_{m}} \}, \{X_{j_{1}},\dots,X_{j_{r}} \}  \right )
\end{equation} 
for certain mutually exclusive $i_{1},\dots,i_{m},j_{1},\dots,j_{r} \in \{1,\dots,91\}$. Typically, one specifies a certain (estimated) bivariate dependence coefficient and a certain link function (overlooking multivariate dependence structures) for computing \eqref{eq: msim}, or (estimated) multivariate concordance measures (between univariate random variables) have also been studied (see, e.g., \cite{Fuchs2021}). We decide to take $\widehat{\mathcal{D}}_{1}((X_{i_{1}},\dots,X_{i_{m}}),(X_{j_{1}},\dots,X_{j_{r}}))$ for \eqref{eq: msim}, and as such measure the similarity between clusters (in spite of possible similarity within), which feels more natural and does not require a link function. This $\widehat{\mathcal{D}}_{1}$ is taken to be $\mathcal{D}_{1}(\widehat{\mathbf{R}}_{\text{R},n})$, where $\widehat{\mathbf{R}}_{\text{R},n}$ is the ridge penalized estimated Gaussian copula correlation matrix \eqref{eq: Ridge_est}, keeping in mind that the total dimension $m + r$ in \eqref{eq: msim} might be large compared to $n$. The tuning parameter $\omega_{n}$ is again determined by a $5$-fold cross-validation search on a grid of $50$ equidistant elements in $[0.01,0.999]$. 

We obtain $91$ partitions, and one of these should be picked as preferable clustering of the attributes. Since not restricted to two groups, we can also use $\widehat{\mathcal{D}}_{1}$ for measuring the similarity $\mathcal{D}(\mathbb{X}_{1},\dots,\mathbb{X}_{k})$ between $k$ clusters $\mathbb{X}_{1},\dots,\mathbb{X}_{k}$ comprising a specific partition. This in fact measures how separated the obtained clusters of that particular partition are, i.e., it is a measure of redundancy among the clusters of attributes, which we want to be rather small. The dendrogram of the clustering procedure and redundancy $\mathcal{D}(\mathbb{X}_{1},\dots,\mathbb{X}_{k})$ as a function of the number of clusters $k$ (for each of the $91$ obtained partitions) are shown in Fig. \ref{fig: redun}. Based on this, we decide to look deeper into the $36$ cluster partition (where the dotted line cuts the dendrogram), since the redundancy is minimal here (in particular, $\mathcal{D}(\mathbb{X}_{1},\dots,\mathbb{X}_{36}) = 0.07265$).
Next, we rearrange the attributes such that dysphonia measures that belong to the same cluster follow each other in the dataset, and the cluster dimensions are
\begin{equation*}
(d_{1},\dots,d_{36}) = (1,1,53,1,1,1,1,1,1,1,2,1,1,1,1,1,1,1,1,1,1,1,1,1,3,1,1,1,1,1,1,1,1,1,1,1),
\end{equation*}
i.e., there is a big cluster of size $53$, two smaller clusters of size $2$ and $3$, and $33$ clusters of size $1$.
The sample normal scores rank correlation matrix after rearrangement, estimated without a penalty and with a ridge penalty are respectively given in the middle and right panel of Fig. S5 in the Supplementary Material. 

Denote now the $36$ clusters as random vectors $\mathbf{X}_{1},\dots,\mathbf{X}_{36}$. We already know that $\widehat{\mathcal{D}}_{1}(\mathbf{X}_{1},\dots,\mathbf{X}_{36}) = 0.07265$ when using ridge penalization ($\omega_{n} = 0.696$), which is already way closer to $0$ than when using no penalization (then $\widehat{\mathcal{D}}_{1}(\mathbf{X}_{1},\dots,\mathbf{X}_{36}) = 0.19013$). Moreover, the rationale behind variable clustering is that clusters are truly separated, i.e., there is (block) sparsity in the Gaussian copula correlation matrix at levels corresponding to attributes belonging to different clusters.
\begin{figure}[h!] \centering
\includegraphics[scale = 0.8]{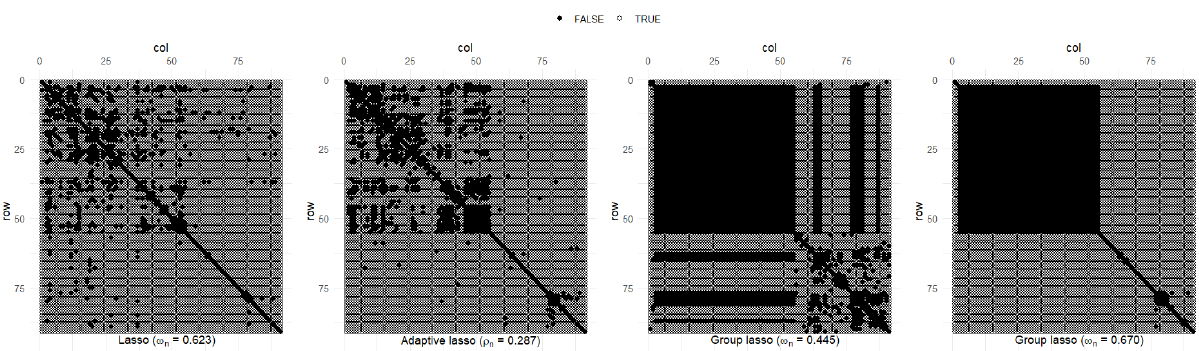}
\caption{Sparsity patterns of estimated Gaussian copula correlation matrix of $91$ dysphonia measures using different penalties.}
\label{fig: SP}
\end{figure} 
Therefore, we also consider lasso, adaptive lasso, and group lasso estimation of the Gaussian copula correlation matrix. 
For the lasso and group lasso, we search for an optimal $\omega_{n}$ on an equidistant grid of $50$ values in $[0.01;0.8]$, and for the adaptive lasso, we search for an optimal $\rho_{n}$ on an equidistant grid of $50$ values in $[0.01;0.6]$, each time aiming to maximize the BIC \eqref{eq: BIC}. 

The first panel in Fig. \ref{fig: SP} shows the sparse structure of the estimated normal scores rank correlation matrix of the $91$ clustered attributes when using the lasso. To a certain extent, we recognize the diagonal blocks corresponding to the within cluster correlations, and observe many zeroes between attributes belonging to different clusters. The estimated inter-cluster dependence equals $\widehat{\mathcal{D}}_{1}(\mathbf{X}_{1},\dots,\mathbf{X}_{36}) = 0.0014$. Using the adaptive lasso (second panel), we get $\widehat{\mathcal{D}}_{1}(\mathbf{X}_{1},\dots,\mathbf{X}_{36}) = 0.0141$, and a group lasso estimator (third panel) results in $\widehat{\mathcal{D}}_{1}(\mathbf{X}_{1},\dots,\mathbf{X}_{36}) = 0.0130$. Note that the group lasso estimator did not shrink any of the diagonal blocks, reflecting the stronger intra-cluster similarities. In the fourth panel of Fig. \ref{fig: SP}, we visualize a group lasso estimate with a larger penalty parameter (not determined through BIC). When forcing more sparsity by increasing $\omega_{n}$, we see more off-diagonal blocks (and not diagonal blocks) getting shrunk, and this leads to $\widehat{\mathcal{D}}_{1}(\mathbf{X}_{1},\dots,\mathbf{X}_{36}) = 0.0001$. More detailed images are given in Fig. S6 of the Supplementary Material.

Further feature selection can now be done via replacing each cluster by a single (or maybe multiple) attribute that represents the cluster well (e.g., the attribute within the cluster that has largest similarity with the cluster). 

\section{Discussion}\label{sec: 7}

In this paper, we proposed to use the $2$-Wasserstein distance between the joint copula measure and the product of the marginal copula measures for  quantifying dependence between a finite, arbitrary amount
of random vectors. The obtained dependence coefficients satisfy several desirable properties and especially have the powerful theoretical quality of detecting any departure from independence. Examples illustrate that the choice of normalization strongly influences the overall dependence quantification. Whether or not one can explicitly calculate the infimum for the optimal transport map and/or the supremum required for the normalization, depends on the specific form of the copula, and a great deal of interesting work remains to be done (also when no explicit copula can be assumed, i.e., when nonparametric estimators need to be considered). 

A Gaussian copula approach yields explicit formulas with a clear interpretation. Using the sample matrix of normal scores rank correlation coefficients results in an easily computable plug-in estimator for which we obtained an asymptotic normality result with explicit asymptotic variance in arbitrary dimensions. Expectedly, higher dimensions aggravate the finite sample estimation performance. To cope with this, we studied rank-invariant penalization techniques for estimating the Gaussian copula correlation matrix, leading to estimators that are able to improve accuracy on the one hand, and enhance interpretability by detecting marginal independencies on the other hand. Such estimation challenges have enjoyed rather little attention, and further research would definitely be worthwhile.

Another interesting theoretical challenge is to study the optimal transport (and others as well) dependence coefficients when the dimension grows unboundedly. In the Gaussian copula context, random matrix theory could definitely be useful, and, as also touched upon in our simulations, various behaviours can be expected depending on the nature of the dependence measure, the normalization, whether letting $d_{i} \to \infty$ for some $i$, or $k \to \infty$. Keeping the dimension fixed, our simulations illustrated the asymptotic normality result and the benefits of using penalization techniques when the sample size is rather small and/or when (group) sparsity is pursued.

Finally, in a first real data application, we illustrated the use of the dependence coefficients in evaluating and comparing (possibly more than two) consumer products or similarities in sensory analysis. Alongside, on a second real dataset containing expressions on sustained vowel phonations,  we demonstrated how attributes can be hierarchically clustered via multivariate similarities between random vectors (despite similarities within), disposing of traditional link functions. Ridge penalization is preferred when the number of attributes is large compared to the sample size, and (group)-lasso type penalties can be used to reflect the homogeneity and separation of a partition of, in general $k$, groups of variables. 
\medskip 

\noindent
\textbf{Acknowledgments.}
The authors thank Dr Gilles Mordant, Georg-August-Universit\"at G\"ottingen, for scientific discussions during the startup phase of this research. The authors gratefully acknowledge support from the Research
Fund KU Leuven [C16/20/002 project].

\clearpage 

\appendix
\section{Axioms for dependence measures between random vectors}\label{App A}

In \cite{Gijbels2023}, a list of axioms is stated for a dependence measure $\mathcal{D}^{d_{1}\dots,d_{k}}(\mathbf{X}) = \mathcal{D}(\mathbf{X}_{1},\dots,\mathbf{X}_{k})$. Up to some minor differences (small corrections and simplifications), the axioms are given as follows.

\vspace{0.3cm}
\begin{enumerate} 
\addtolength{\itemindent}{0.4cm}
\item[(A1)]{For every permutation $\{\pi(1),\dots,\pi(k)\}$ of $\{1,\dots,k\}$: $\mathcal{D}^{d_{1},\dots,d_{k}}(\mathbf{X}) = \mathcal{D}(\mathbf{X}_{\pi(1)},\dots,\mathbf{X}_{\pi(k)})$: and for every permu- \hspace*{0.4cm} tation $\{\pi_{i}(1),\dots,\pi_{i}(d_{i})\}$ of $\{1,\dots,d_{i}\}$, $i \in \{1,\dots,k\}$: $\mathcal{D}^{d_{1},\dots,d_{k}}(\mathbf{X}) = \mathcal{D}(\mathbf{X}_{1},\dots,(X_{i\pi_{i}(1)},\dots,X_{i\pi_{i}(d_{i})}),\dots,\mathbf{X}_{k})$.}
\item[(A2)]{$0 \leq \mathcal{D}^{d_{1},\dots,d_{k}}(\mathbf{X}) \leq 1$.}
\item[(A3)]{$\mathcal{D}^{d_{1},\dots,d_{k}}(\mathbf{X}) = 0$ if and only if $\mathbf{X}_{1}, \dots, \mathbf{X}_{k}$ are mutually independent.}
\item[(A4)]{$\mathcal{D} (\mathbf{X}_{1},\dots,\mathbf{X}_{k}, \mathbf{X}_{k+1}) \geq \mathcal{D}(\mathbf{X}_{1},\dots,\mathbf{X}_{k})$ with equality if and only if $\mathbf{X}_{k+1}$ is independent of $(\mathbf{X}_{1},\dots,\mathbf{X}_{k})$.}
\item[(A5)]{$\mathcal{D}^{d_{1},\dots,d_{k}}(\mathbf{X})$ is well-defined for any $q$-dimensional random vector $\mathbf{X}$ (even if there is a singular part in the \hspace*{0.4cm} distribution of $\mathbf{X}$).}
\item[(A6)]{$\mathcal{D}^{d_{1},\dots,d_{k}}(\mathbf{X})$ is a function of solely the copula $C$ of $\mathbf{X}$ (which is equivalent to $\mathcal{D}^{d_{1},\dots,d_{k}}(\mathbf{X})$ being invariant under \hspace*{0.4cm} strictly increasing transformations of any of the components of $\mathbf{X}$).}
\item[(A7)]{Let $T_{ij}$ be a strictly decreasing, continuous transformation for a fixed $i \in \{1,\dots,k\}$ and a fixed $j \in \{1,\dots,d_{i} \}$. \hspace*{0.4cm} Then $$\mathcal{D} \big ( \mathbf{X}_{1},\dots, T_{i}(\mathbf{X}_{i}), \dots,\mathbf{X}_{k} \big ) = \mathcal{D}(\mathbf{X}_{1},\dots,\mathbf{X}_{k}),$$ \hspace*{0.4cm} where $T_{i}(\mathbf{X}_{i}) = (X_{i1},\dots,T_{ij}(X_{ij}),\dots,X_{id_{i}})$}.
\item[(A8)] Let $(\mathbf{X}_n)_{n \in \mathbb{N}}$ be a sequence of $q$-dimensional random vectors with corresponding copulas $(C_{n})_{n \in \mathbb{N}}$, then 
$$\lim_{n \to \infty} \mathcal{D}^{d_{1},\dots,d_{k}}(\mathbf{X}_{n}) = \mathcal{D}^{d_{1},\dots,d_{k}}(\mathbf{X})$$ \hspace*{0.4cm} if $C_{n} \to C$ uniformly, where $C$ denotes the copula of $\mathbf{X}$.
\end{enumerate}

\section{Proofs of theoretical results of Section \ref{sec:2}}\label{App B}
\subsection{Proof of Lemma \ref{lem1}}\label{App B1}
For (a), let $\mathbf{V} = (\mathbf{V}_{1},\dots,\mathbf{V}_{k})$ with $\mathbf{V}_{i} = (V_{i1},\dots,V_{id_{i}})$ be a random vector with distribution $\nu_{1} \times \cdots \times \nu_{k}$ and let $\mathbf{U} = (\mathbf{U}_{1},\dots,\mathbf{U}_{k})$ with $\mathbf{U}_{i} = (U_{i1},\dots,U_{id_{i}})$ be a random vector with distribution $\mu_{C}$, as marginal distribution of an arbitrary coupling $\gamma \in \Gamma(\mu_{C},\nu_{1} \times \cdots \times \nu_{k})$ of $(\mathbf{U},\mathbf{V})$, and having marginals $\mu_{C_{1}},\dots,\mu_{C_{k}}$ itself. Then, the distribution of $(\mathbf{U}_{i},\mathbf{V}_{i})$ is a coupling of $\mu_{C_{i}}$ and $\nu_{i}$ for all $i = 1,\dots,k$ such that
\begin{equation}\label{eq: helpproof}
\mathbb{E} \left (||\mathbf{U} - \mathbf{V}||^{2} \right ) = \mathbb{E} \left (\sum_{i=1}^{k} \sum_{j=1}^{d_{i}} (U_{ij}-V_{ij})^{2}  \right ) = \sum_{i=1}^{k} \mathbb{E} \left ( ||\mathbf{U}_{i}-\mathbf{V}_{i}||^{2} \right ) \geq \sum_{i=1}^{k} W_{2}^{2}(\mu_{C_{i}},\nu_{i}).
\end{equation}
Taking the infimum over all couplings $\gamma  \in \Gamma(\mu_{C},\nu_{1} \times \cdots \times \nu_{k})$ yields the result. 
Part (b) follows immediately from the definition. 

Regarding (c), if $\mu_{C_{i}} = \nu_{i}$ for all $i = 1,\dots,k$, we have $T_{d_{1},\dots,d_{k}}(\mu_{C};\nu_{1},\dots,\nu_{k}) = W_{2}^{2}(\mu_{C}, \mu_{C_{1}} \times \cdots \times \mu_{C_{k}})$, making the statement trivial since $W_{2}$ defines a metric. Suppose now that $\nu_{i}$ is absolutely continuous for all $i = 1,\dots,k$ and $T_{d_{1},\dots,d_{k}}(\mu_{C};\nu_{1},\dots,\nu_{k}) = 0$. The latter means that, working further on the proof of (a), there exists a  $\gamma  \in \Gamma(\mu_{C},\nu_{1} \times \cdots \times \nu_{k})$ minimizing the left-hand side of \eqref{eq: helpproof}, and with equality instead of inequality.  Hence, the $2$-Wasserstein distance between $\mu_{C_{i}}$ and $\nu_{i}$ is obtained at the coupling distribution of $(\mathbf{U}_{i},\mathbf{V}_{i})$ coming from $\gamma$, for all $i = 1,\dots,k$. However, Brenier's theorem (see, e.g., Theorem 2.12 in \cite{Villani2008}) tells us that, since $\nu_{i}$ is absolutely continuous, this optimum is uniquely and deterministically attained, i.e., (denoting $\nabla$ for the gradient) there must exist convex functions $\psi_{i} : \mathbb{I}^{d_{i}} \rightarrow \mathbb{R} \cup \{\infty\}$ such that $\mathbf{U}_{i} = \nabla \psi_{i}(\mathbf{V}_{i})$ almost surely for all $i = 1,\dots,k$. Since $\mathbf{V}_{1},\dots,\mathbf{V}_{k}$ are independent, this implies that $\mathbf{U}_{1},\dots,\mathbf{U}_{k}$ are independent and thus
$\mu_{C} = \mu_{C_{1}} \times \cdots \times \mu_{C_{k}}$.

We next prove statement (d). The $W_{2}$-compactness of $\Gamma(\mu_{C_{1}},\dots,\mu_{C_{k}})$  follows from the well-known result that if $(\gamma_{n})_{n} \in \Gamma(\mu_{C_{1}},\dots,\mu_{C_{k}})$ is $W_{2}$-convergent, then it is also weakly convergent (see, e.g., \cite{Clement2008}) and the limit will again be in $\Gamma(\mu_{C_{1}},\dots,\mu_{C_{k}})$ as the marginals remain fixed. Finally, let us consider a fixed $\pi \in \Gamma(\mu_{C_{1}},\dots,\mu_{C_{k}})$  and arbitrary $\epsilon > 0$. Take an arbitrary $\widetilde{\pi} \in \Gamma(\mu_{C_{1}},\dots,\mu_{C_{k}})$ such that $W_{2}(\pi,\widetilde{\pi}) < \delta$, where we pick $\delta > 0$ such that $\delta < \min\{1,\epsilon/(1+2D_{\pi})\}$ with $D_{\pi} = W_{2}(\pi;\nu_{1} \times \cdots \times \nu_{k})$. We also denote $D_{\widetilde{\pi}} = W_{2}(\widetilde{\pi};\nu_{1} \times \cdots \times \nu_{k})$. It holds that
\begin{equation*}
\begin{split}
    |T_{d_{1},\dots,d_{k}}(\pi;\nu_{1}, \dots, \nu_{k}) - T_{d_{1},\dots,d_{k}}(\widetilde{\pi};\nu_{1}, \dots, \nu_{k})| = |D_{\pi}^{2} - D_{\widetilde{\pi}}^{2}|  & = |D_{\widetilde{\pi}} - D_{\pi}| \cdot |D_{\widetilde{\pi}} - D_{\pi} + 2 D_{\pi}| \\ & \leq |D_{\widetilde{\pi}} - D_{\pi}| \left (|D_{\widetilde{\pi}} - D_{\pi}| + 2 D_{\pi} \right ) \\
    & < \epsilon,
\end{split}
\end{equation*}
where we used that $|D_{\widetilde{\pi}} - D_{\pi}| \leq W_{2}(\pi,\widetilde{\pi})$ because of the triangle inequality (recall that $W_{2}$ defines a metric). This finishes the proof of statement (d). \hfill \qedsymbol

\subsection{Proof of Proposition \ref{prop1}}\label{App B2}
We start with proving (A4) for $T_{d_{1},\dots,d_{k}}$. Suppose we consider an additional random vector $\mathbf{X}_{k+1}$ having copula measure $\mu_{C_{k+1}}$, an additional absolutely continuous reference measure $\nu_{k+1}$, and let $\mu_{\widetilde{C}}$ be the copula measure of $(\mathbf{X}_{1},\dots,\mathbf{X}_{k+1})$. If we first assume that $\mathbf{X}_{k+1}$ is independent of $(\mathbf{X}_{1},\dots,\mathbf{X}_{k})$, we have $\mu_{\widetilde{C}} = \mu_{C} \times \mu_{C_{k+1}}$ and hence 
\begin{equation*}
\begin{split}
    T_{d_{1},\dots,d_{k+1}}(\mu_{\widetilde{C}};\nu_{1},\dots,\nu_{k+1}) & = W_{2}^{2}(\mu_{C} \times \mu_{C_{k+1}},\nu_{1} \times \cdots \times \nu_{k+1}) - \sum_{i=1}^{k+1} W_{2}^{2}(\mu_{C_{i}},\nu_{i}) \\
    & = W_{2}^{2}(\mu_{C},\nu_{1} \times \cdots \times \nu_{k}) + W_{2}^{2}(\mu_{C_{k+1}},\nu_{k+1}) - \sum_{i=1}^{k+1} W_{2}^{2}(\mu_{C_{i}},\nu_{i}) \\
    & =  W_{2}^{2}(\mu_{C},\nu_{1} \times \cdots \times \nu_{k}) - \sum_{i=1}^{k} W_{2}^{2}(\mu_{C_{i}},\nu_{i}) \\
    & =  T_{d_{1},\dots,d_{k}}(\mu_{C};\nu_{1},\dots,\nu_{k}),
\end{split}
\end{equation*}
such that the Wasserstein dependence measure remains unchanged when adding $\mathbf{X}_{k+1}$ into consideration. In general, suppose that $\gamma \in \Gamma(\mu_{\widetilde{C}},\nu_{1} \times \cdots \times \nu_{k+1})$ is an optimal transport map from $\mu_{\widetilde{C}}$ to $\nu_{1} \times \cdots \times \nu_{k+1}$, as joint distribution of $(\widetilde{\mathbf{U}},\widetilde{\mathbf{V}})$ with $\widetilde{\mathbf{U}} = (\mathbf{U},\mathbf{U}_{k+1}) = (\mathbf{U}_{1},\dots,\mathbf{U}_{k+1})$ and $\widetilde{\mathbf{V}} = (\mathbf{V},\mathbf{V}_{k+1}) = (\mathbf{V}_{1},\dots,\mathbf{V}_{k+1})$, i.e.,
\begin{equation*}
\begin{split}
    T_{d_{1},\dots,d_{k+1}}(\mu_{\widetilde{C}};\nu_{1},\dots,\nu_{k+1}) & = \mathbb{E} \left (||\widetilde{\mathbf{U}} - \widetilde{\mathbf{V}}||^{2} \right ) - \sum_{i=1}^{k+1} W_{2}^{2}(\mu_{C_{i}},\nu_{i}) \\
    & = \mathbb{E} \left (||\mathbf{U} - \mathbf{V}||^{2} \right ) + \mathbb{E} \left (||\mathbf{U}_{k+1} - \mathbf{V}_{k+1}||^{2} \right )  - \sum_{i=1}^{k+1} W_{2}^{2}(\mu_{C_{i}},\nu_{i}) \\
    & \geq W_{2}^{2}(\mu_{C},\nu_{1} \times \cdots \times \nu_{k}) -  \sum_{i=1}^{k} W_{2}^{2}(\mu_{C_{i}},\nu_{i}) \\
    & = T_{d_{1},\dots,d_{k}}(\mu_{C};\nu_{1},\dots,\nu_{k}),
\end{split}
\end{equation*}
where the inequality follows from the definition of the Wasserstein distance. If this inequality is an equality, it must hold that
\begin{equation*}
\begin{split}
   W_{2}^{2}(\mu_{C},\nu_{1} \times \cdots \times \nu_{k}) & = \mathbb{E} \left (||\mathbf{U} - \mathbf{V}||^{2} \right )  + \mathbb{E} \left (||\mathbf{U}_{k+1} - \mathbf{V}_{k+1}||^{2} \right ) - W_{2}^{2}(\mu_{C_{k+1}},\nu_{k+1}) \geq \mathbb{E} \left (||\mathbf{U} - \mathbf{V}||^{2} \right ),
\end{split}
\end{equation*}
from which
\begin{equation*}
    \mathbb{E} \left (||\mathbf{U} - \mathbf{V}||^{2} \right ) = W_{2}^{2}(\mu_{C},\nu_{1} \times \cdots \times \nu_{k}) \hspace{0.2cm} \text{and} \hspace{0.2cm} \mathbb{E} \left (||\mathbf{U}_{k+1} - \mathbf{V}_{k+1}||^{2} \right ) = W_{2}^{2}(\mu_{C_{k+1}},\nu_{k+1}).
\end{equation*}
Because of the absolute continuity of the $\nu_{i}$, the latter implies (using Brenier's theorem) that there exists $T_{k+1} : \mathbb{I}^{d_{k+1}} \rightarrow \mathbb{R}^{d_{k+1}}$ such that $\mathbf{U}_{k+1} = T_{k+1}(\mathbf{V}_{k+1})$ almost surely, while the former means that there exist $T_{i}: \mathbb{I}^{q} \rightarrow \mathbb{R}^{d_{i}}$ such that $\mathbf{U}_{i} = T_{i}(\mathbf{V}_{1},\dots,\mathbf{V}_{k})$ almost surely for all $i = 1,\dots,k$. Since $\mathbf{V}_{k+1}$ is independent of $(\mathbf{V}_{1},\dots,\mathbf{V}_{k})$, this shows that $\mathbf{U}_{k+1}$ is independent of $(\mathbf{U}_{1},\dots,\mathbf{U}_{k})$, i.e., $\mu_{\widetilde{C}} = \mu_{C} \times \mu_{C_{k+1}}$ and thus $\mathbf{X}_{k+1}$ independent from $(\mathbf{X}_{1},\dots,\mathbf{X}_{k})$.

The fact that $\mathcal{D}$ satisfies (A1)-(A3),(A5),(A6) follows from basic properties of copulas and Lemma \ref{lem1}.

In the context of property (A7), assume without loss of generality that $X_{11}$ gets transformed to $T_{11}(X_{11})$ for a strictly decreasing transformation $T_{11}$ and let $\mu_{\widetilde{C}}$ be the copula distribution of $(T_{1}(\mathbf{X}_{1}),\mathbf{X}_{2},\dots,\mathbf{X}_{k})$ with $T_{1}(\mathbf{X}_{1}) = (T_{11}(X_{11}),X_{12},\dots,X_{1d_{1}})$. Then, $\widetilde{C}(\mathbf{u}_{1},\dots,\mathbf{u}_{k}) = C(\overline{\mathbf{u}}_{1},\mathbf{u}_{2},\dots,\mathbf{u}_{k}) - C(\widetilde{\mathbf{u}}_{1},\mathbf{u}_{2},\dots,\mathbf{u}_{k})$ with $\overline{\mathbf{u}}_{1} = (1,u_{12},\dots,u_{1d_{1}})$ and $\widetilde{\mathbf{u}}_{1} = (1-u_{11},u_{12},\dots,u_{1d_{1}})$. Suppose now that $\gamma \in \Gamma(\mu_{C},\nu_{1} \times \cdots \times \nu_{k})$ is an optimal transport map from $\mu_{C}$ to $\nu_{1} \times \cdots \times \nu_{k}$, as joint distribution of $(\mathbf{U},\mathbf{V})$, i.e.,
\begin{equation*}
    W_{2}^{2}(\mu_{C},\nu_{1} \times \cdots \times \nu_{k}) = \mathbb{E} \left (||\mathbf{U}-\mathbf{V}||^{2} \right ) = \int_{\mathbb{I}^{2q}} ||\mathbf{u}-\mathbf{v}||^{2} d\gamma(\mathbf{u},\mathbf{v}). 
\end{equation*}
Consider then $\widetilde{\gamma}(\mathbf{u},\mathbf{v}) = \gamma(\overline{\mathbf{u}}_{1},\mathbf{u}_{2},\dots,\mathbf{u}_{k},\mathbf{v}) - \gamma(\widetilde{\mathbf{u}}_{1},\mathbf{u}_{2},\dots,\mathbf{u}_{k},\mathbf{v})$, which clearly is a coupling of $\mu_{\widetilde{C}}$ and $\nu_{1} \times \cdots \times \nu_{k}$, as joint distribution of $(\widetilde{\mathbf{U}},\mathbf{V})$ with $\widetilde{\mathbf{U}} = (\widetilde{\mathbf{U}}_{1},\mathbf{U}_{2},\dots,\mathbf{U}_{k})$ and $\widetilde{\mathbf{U}}_{1} = (1-U_{11},U_{12},\dots,U_{1d_{1}})$. Moreover, putting $\widetilde{\mathbf{u}} = (\widetilde{\mathbf{u}}_{1},\mathbf{u}_{2},\dots,\mathbf{u}_{2})$, we have 
\begin{equation*}
\begin{split}
    W_{2}^{2}(\mu_{\widetilde{C}}, \nu_{1} \times \cdots \times \nu_{k})  \leq \mathbb{E} \left (||\widetilde{\mathbf{U}} - \mathbf{V}||^{2} \right )  = \int_{\mathbb{I}^{2q}} ||\widetilde{\mathbf{u}} - \mathbf{v} ||^{2} d\widetilde{\gamma}(\mathbf{u},\mathbf{v}) 
    & = \int_{\mathbb{I}^{2q}} ||\mathbf{u}-\mathbf{v}||^{2} d\gamma(\mathbf{u},\mathbf{v}) = W_{2}^{2}(\mu_{C},\nu_{1} \times \cdots \times \nu_{k}),
\end{split}
\end{equation*}
by simply doing a substitution $t_{11} = 1 - u_{11}$. Following a reversed reasoning, we also obtain $W_{2}^{2}(\mu_{C},\nu_{1} \times \cdots \times \nu_{k}) \leq W_{2}^{2}(\mu_{\widetilde{C}}, \nu_{1} \times \cdots \times \nu_{k})$ and hence $W_{2}^{2}(\mu_{C},\nu_{1} \times \cdots \times \nu_{k}) = W_{2}^{2}(\mu_{\widetilde{C}}, \nu_{1} \times \cdots \times \nu_{k})$. Similarly, we can show that $W_{2}^{2}(\mu_{C_{1}},\nu_{1}) = W_{2}^{2}(\mu_{\widetilde{C}_{1}},\nu_{1})$ where $\mu_{\widetilde{C}_{1}}$ is the copula distribution of $T_{1}(\mathbf{X}_{1})$ and hence conclude (A7).

Finally, for (A8), we want that if $C_{n} \to C$ uniformly for $n \to \infty$, it is true that 
\begin{equation}\label{eq: convl}
    |T_{d_{1},\dots,d_{k}}(\mu_{C_{n}}; \nu_{1}, \dots, \nu_{k}) - T_{d_{1},\dots,d_{k}}(\mu_{C}; \nu_{1}, \dots, \nu_{k})| \to 0
\end{equation}
as $n \to \infty$. Note that if $C_{n} \to C$ uniformly, we also have $\int_{\mathbb{I}^{q}} ||\mathbf{u}||^{p} dC_{n}(\mathbf{u}) \to \int_{\mathbb{I}^{q}} ||\mathbf{u}||^{p} dC(\mathbf{u})$ as $n \to \infty$ for every $p > 0$ (Helly-Bray theorem). Since convergence in distribution together with convergence of the first two moments implies $W_{2}$-convergence, equation \eqref{eq: convl} is proven by using similar arguments as for (d) of Lemma \ref{lem1}.  \hfill \qedsymbol 

\section{Proofs of theoretical results of Section \ref{sec:3}}\label{App C}
\subsection{Proof of Proposition \ref{prop2}}\label{App C1}

The first step of the proof consists of showing that the eigenvalues of $\mathbf{R}_{m}$ are indeed given by \eqref{eq: eigRm}. Therefore, already note that because of the orthogonality of the matrix 
\begin{equation*}
    \mathbf{U} = \begin{pmatrix}
    \mathbf{U}_{11} & \mathbf{0}_{12} & \cdots & \mathbf{0}_{1k} \\
    \mathbf{0}_{12}^{\text{T}} & \mathbf{U}_{22} & \cdots & \mathbf{0}_{2k} \\
    \vdots & \vdots & \ddots & \vdots \\
    \mathbf{0}_{1k}^{\text{T}} & \mathbf{0}_{2k}^{\text{T}} & \cdots & \mathbf{U}_{kk}
    \end{pmatrix},
\end{equation*}
denoting $\mathbf{0}_{ij} \in \mathbb{R}^{d_{i} \times d_{j}}$ for a matrix of zeroes, 
the eigenvalues of $\mathbf{R}_{m}$ are the same as those of
\begin{equation}\label{eq: biglambda}
\boldsymbol{\Lambda}_{m} = \mathbf{U}^{\text{T}}\mathbf{R}_{m}\mathbf{U} = 
\begin{pmatrix} \vspace{0.1cm}
\boldsymbol{\Lambda}_{11} & \boldsymbol{\Lambda}_{11}^{1/2} \boldsymbol{\Pi}_{12} \boldsymbol{\Lambda}_{22}^{1/2} & \cdots & \boldsymbol{\Lambda}_{11}^{1/2} \boldsymbol{\Pi}_{1k} \boldsymbol{\Lambda}_{kk}^{1/2} \\ 
\boldsymbol{\Lambda}_{22}^{1/2} \boldsymbol{\Pi}_{12}^{\text{T}}  \boldsymbol{\Lambda}_{11}^{1/2} & \boldsymbol{\Lambda}_{22} & \cdots & \boldsymbol{\Lambda}_{22}^{1/2} \boldsymbol{\Pi}_{2k} \boldsymbol{\Lambda}_{kk}^{1/2} \\
\vdots & \vdots & \ddots & \vdots \\
\boldsymbol{\Lambda}_{kk}^{1/2} \boldsymbol{\Pi}_{1k}^{\text{T}}  \boldsymbol{\Lambda}_{11}^{1/2} & \boldsymbol{\Lambda}_{kk}^{1/2} \boldsymbol{\Pi}_{2k}^{\text{T}} \boldsymbol{\Lambda}_{22}^{1/2} & \cdots & \boldsymbol{\Lambda}_{kk}
\end{pmatrix},
\end{equation}
which we can find explicitly. Let $\mathbf{e}_{r,\ell}$ be the $r$-th canonical unit (column) vector in $\mathbb{R}^{\ell}$ and put $d_{0} = 0$. For ease of notation, we also write $\boldsymbol{\Pi}_{r\ell} = \boldsymbol{\Pi}_{\ell r}^{\text{T}}$ if $r > \ell$. We now list the eigenvalues and eigenvectors of the matrix $\boldsymbol{\Lambda}_{m}$.
\begin{itemize}
    \item For all $i = 0,\dots,k-2$ and all $j = d_{i}+1,\dots,d_{i+1}$, the vector
    \begin{equation}\label{eq: eig1}
        \left ( \mathbf{0}_{d_{1}}^{\text{T}},\dots,\mathbf{0}_{d_{i}}^{\text{T}},\lambda_{j,(i+1)(i+1)}^{1/2}\mathbf{e}_{j,d_{i+1}}^{\text{T}},\dots,\lambda_{j,kk}^{1/2}\mathbf{e}_{j,d_{k}}^{\text{T}} \right )^{\text{T}},
    \end{equation}
    denoting $\mathbf{0}_{d_{i}} \in \mathbb{R}^{d_{i}}$ for a column vector of zeroes, is an eigenvector of $\boldsymbol{\Lambda}_{m}$ with eigenvalue
    \begin{equation}\label{eq: val1}
        \lambda_{j,(i+1)(i+1)} + \dots + \lambda_{j,kk}.
    \end{equation}
\end{itemize}
Indeed, when fixing a certain $i \in \{0,\dots,k-2\}$ and $j \in \{d_{i}+1,\dots,d_{i+1}\}$, the $r$-th block row of \eqref{eq: biglambda} for $r \in \{1,\dots,i\}$, multiplied with \eqref{eq: eig1} equals
\begin{equation*}
\begin{split}
    \boldsymbol{\Lambda}_{rr} \mathbf{0}_{d_{r}} + \sum_{\substack{\ell=1 \\ \ell\neq r}}^{i} \boldsymbol{\Lambda}_{rr}^{1/2}\boldsymbol{\Pi}_{r\ell} \boldsymbol{\Lambda}_{\ell\ell}^{1/2}\mathbf{0}_{d_{\ell}} + \sum_{\ell = i+1}^{k} \boldsymbol{\Lambda}_{rr}^{1/2}\boldsymbol{\Pi}_{r\ell} \boldsymbol{\Lambda}_{\ell\ell}^{1/2} \lambda_{j,\ell\ell}^{1/2}\mathbf{e}_{j,d_{\ell}} & = \sum_{\ell=i+1}^{k} \lambda_{j,\ell\ell}  \boldsymbol{\Lambda}_{rr}^{1/2}\boldsymbol{\Pi}_{r\ell} \mathbf{e}_{j,d_{\ell}} = \mathbf{0}_{d_{r}},
\end{split}
\end{equation*}
since $\boldsymbol{\Pi}_{r\ell}\mathbf{e}_{j,d_{\ell}} = \mathbf{0}_{d_{r}}$, as $j > d_{r}$. If for $r \in \{i+1,\dots,k\}$, we multiply the $r$-th block row of \eqref{eq: biglambda} with \eqref{eq: eig1}, we obtain
\begin{equation*} 
\begin{split}
    \sum_{\ell=1}^{i} \boldsymbol{\Lambda}_{rr}^{1/2}\boldsymbol{\Pi}_{r\ell} \boldsymbol{\Lambda}_{\ell\ell}^{1/2}\mathbf{0}_{d_{\ell}} + \boldsymbol{\Lambda}_{rr} \lambda_{j,rr}^{1/2}\mathbf{e}_{j,d_{r}} + \sum_{\substack{\ell = i+1 \\ \ell \neq r}}^{k}  \boldsymbol{\Lambda}_{rr}^{1/2}\boldsymbol{\Pi}_{r\ell} \boldsymbol{\Lambda}_{\ell\ell}^{1/2} \lambda_{j,\ell\ell}^{1/2}\mathbf{e}_{j,d_{\ell}} & = \lambda_{j,rr}^{3/2} \mathbf{e}_{j,d_{r}} + \sum_{\substack{\ell = i+1 \\ \ell \neq r}}^{k} \lambda_{j,\ell\ell}  \boldsymbol{\Lambda}_{rr}^{1/2}\boldsymbol{\Pi}_{r\ell} \mathbf{e}_{j,d_{\ell}} \\
    & = \lambda_{j,rr}^{3/2} \mathbf{e}_{j,d_{r}} + \lambda_{j,rr}^{1/2} \sum_{\substack{\ell = i+1 \\ \ell \neq r}}^{k} \lambda_{j,\ell\ell} \hspace{0.01cm} \mathbf{e}_{j,d_{r}} \\
    & = \lambda_{j,rr}^{1/2} \left (\sum_{\ell=i+1}^{k} \lambda_{j,\ell\ell} \right ) \mathbf{e}_{j,d_{r}},
\end{split}
\end{equation*} 
since $\boldsymbol{\Pi}_{rl}\mathbf{e}_{j,d_{\ell}} = \mathbf{e}_{j,d_{r}}$, as $j \leq d_{r}$. Hence, for every $r \in \{1,\dots,k\}$, the $r$-th block element (in $\mathbb{R}^{d_{r} \times 1}$) of the matrix multiplication of $\boldsymbol{\Lambda}_{m}$ with the vector \eqref{eq: eig1}, is equal to \eqref{eq: val1} multiplied with the vector \eqref{eq: eig1}. This shows what was desired, and delivers $d_{1} + (d_{2} - d_{1}) + \dots + (d_{k-1} - d_{k-2}) = d_{k-1}$ eigenvalues. 
\begin{itemize}
    \item For all $j = d_{k-1} + 1, \dots, d_{k}$, the vector
    \begin{equation}\label{eq: eig2}
        \left ( \mathbf{0}_{d_{1}}^{\text{T}},\dots,\mathbf{0}_{d_{k-1}}^{\text{T}},\mathbf{e}_{j,d_{k}}^{\text{T}} \right )^{\text{T}}
    \end{equation} 
    is an eigenvector with eigenvalue 
    \begin{equation*}
        \lambda_{j,kk}.
    \end{equation*}
\end{itemize}
Indeed, following the previous result, it is straightforward to see that for $r \in \{1,\dots,k-1\}$, the multiplication of the $r$-th block row of $\boldsymbol{\Lambda}_{m}$ with the vector \eqref{eq: eig2} equals $\mathbf{0}_{d_{r}}$, as $j > d_{r}$, while the last block row leads to $\boldsymbol{\Lambda}_{kk}\mathbf{e}_{j,d_{k}} = \lambda_{j,kk}\mathbf{e}_{j,d_{k}}$. This gives an additional amount of $(d_{k}-d_{k-1})$ eigenvalues, resulting in an intermediate total of $d_{k-1} + (d_{k} - d_{k-1}) = d_{k}$ eigenvalues.
\begin{itemize}
    \item For all $i = 1,\dots,k-1$ and all $j = 1,\dots,d_{i}$, the vector 
    \begin{equation}\label{eq: eig3}
        \left ( \mathbf{0}_{d_{1}}^{\text{T}},\dots,\mathbf{0}_{d_{i-1}}^{\text{T}},\lambda_{j,(i+1)(i+1)}^{1/2}\mathbf{e}_{j,d_{i}}^{\text{T}}, -\lambda_{j,ii}^{1/2} \mathbf{e}_{j,d_{i+1}}^{\text{T}},\mathbf{0}_{d_{i+2}}^{\text{T}}, \dots, \mathbf{0}_{d_{k}}^{\text{T}} \right )
    \end{equation}
    is an eigenvector with eigenvalue $0$.
\end{itemize}
Indeed, when fixing a certain $i \in \{1,\dots,k-1\}$ and $j \in \{1,\dots,d_{i}\}$, the $r$-th block row of \eqref{eq: biglambda} for $r \in \{1,\dots,i-1\}$ multiplied with \eqref{eq: eig3} equals
\begin{equation*}
\begin{split}
\boldsymbol{\Lambda}_{rr} \mathbf{0}_{d_{r}} + \sum_{\substack{\ell=1 \\ \ell\neq r}}^{i-1} \boldsymbol{\Lambda}_{rr}^{1/2}\boldsymbol{\Pi}_{r\ell} \boldsymbol{\Lambda}_{\ell\ell}^{1/2}\mathbf{0}_{d_{\ell}} & + \boldsymbol{\Lambda}_{rr}^{1/2} \boldsymbol{\Pi}_{ri} \boldsymbol{\Lambda}_{ii}^{1/2} \lambda_{j,(i+1)(i+1)}^{1/2} \mathbf{e}_{j,d_{i}} - \boldsymbol{\Lambda}_{rr}^{1/2} \boldsymbol{\Pi}_{r(i+1)} \boldsymbol{\Lambda}_{(i+1)(i+1)}^{1/2} \lambda_{j,ii}^{1/2} \mathbf{e}_{j,d_{i+1}} + \sum_{\ell= i+2}^{k} \boldsymbol{\Lambda}_{rr}^{1/2}\boldsymbol{\Pi}_{rl} \boldsymbol{\Lambda}_{\ell\ell}^{1/2}\mathbf{0}_{d_{\ell}} \\
    & \hspace{-4cm}  = \lambda_{j,(i+1)(i+1)}^{1/2} \lambda_{j,ii}^{1/2} \boldsymbol{\Lambda}_{rr}^{1/2} \boldsymbol{\Pi}_{ri}\mathbf{e}_{j,d_{i}} - \lambda_{j,ii}^{1/2} \lambda_{j,(i+1)(i+1)}^{1/2} \boldsymbol{\Lambda}_{rr}^{1/2} \boldsymbol{\Pi}_{r(i+1)} \mathbf{e}_{j,d_{i+1}}  \\ & \hspace{-4cm} = \mathbf{0}_{d_{r}}, 
\end{split}
\end{equation*} 
since
$\boldsymbol{\Pi}_{ri} \mathbf{e}_{j,d_{i}} = \boldsymbol{\Pi}_{r(i+1)} \mathbf{e}_{j,d_{i+1}}$. For $r = i$, we get
\begin{equation*}
\begin{split}
 \sum_{\ell=1}^{i-1} \boldsymbol{\Lambda}_{rr}^{1/2}\boldsymbol{\Pi}_{rl} \boldsymbol{\Lambda}_{\ell\ell}^{1/2}\mathbf{0}_{d_{\ell}} & + \lambda_{j,(i+1)(i+1)}^{1/2} \boldsymbol{\Lambda}_{ii} \mathbf{e}_{j,d_{i}} - \lambda_{j,ii}^{1/2} \boldsymbol{\Lambda}_{ii}^{1/2} \boldsymbol{\Pi}_{i(i+1)} \boldsymbol{\Lambda}_{(i+1)(i+1)}^{1/2} \mathbf{e}_{j,d_{i+1}}   + \sum_{\ell= i+2}^{k} \boldsymbol{\Lambda}_{rr}^{1/2}\boldsymbol{\Pi}_{rl} \boldsymbol{\Lambda}_{\ell\ell}^{1/2}\mathbf{0}_{d_{\ell}} \\
    & \hspace{-2.6cm} = \lambda_{j,(i+1)(i+1)}^{1/2} \lambda_{j,ii}\mathbf{e}_{j,d_{i}} - \lambda_{j,ii}^{1/2} \lambda_{j,(i+1)(i+1)}^{1/2} \boldsymbol{\Lambda}_{ii}^{1/2} \mathbf{e}_{j,d_{i}} \\
    & \hspace{-2.6cm}  = \lambda_{j,(i+1)(i+1)}^{1/2} \lambda_{j,ii} \mathbf{e}_{j,d_{i}}  - \lambda_{j,ii}^{1/2} \lambda_{j,(i+1)(i+1)}^{1/2} \lambda_{j,ii}^{1/2} \mathbf{e}_{j,d_{i}} \\  & \hspace{-2.6cm} = \mathbf{0}_{d_{i}},
\end{split}
\end{equation*}
since $\boldsymbol{\Pi}_{i(i+1)} \mathbf{e}_{j,d_{i+1}} = \mathbf{e}_{j,d_{i}}$, as $j \leq d_{i}$. For $r = i+1$, we get
\begin{equation*} 
\begin{split}
     \sum_{\ell=1}^{i-1} \boldsymbol{\Lambda}_{rr}^{1/2}\boldsymbol{\Pi}_{r\ell} \boldsymbol{\Lambda}_{\ell\ell}^{1/2}\mathbf{0}_{d_{\ell}} & + \lambda_{j,(i+1)(i+1)}^{1/2} \boldsymbol{\Lambda}_{(i+1)(i+1)}^{1/2}\boldsymbol{\Pi}_{(i+1)i} \boldsymbol{\Lambda}_{ii}^{1/2} \mathbf{e}_{j,d_{i}} - \lambda_{j,ii}^{1/2}  \boldsymbol{\Lambda}_{(i+1)(i+1)}\mathbf{e}_{j,d_{i+1}}  + \sum_{\ell= i+2}^{k} \boldsymbol{\Lambda}_{rr}^{1/2}\boldsymbol{\Pi}_{rl} \boldsymbol{\Lambda}_{\ell\ell}^{1/2}\mathbf{0}_{d_{\ell}} \\
    & \hspace{-2.7cm} = \lambda_{j,(i+1)(i+1)}^{1/2} \lambda_{j,ii}^{1/2} \boldsymbol{\Lambda}_{(i+1)(i+1)}^{1/2}\mathbf{e}_{j,d_{i+1}} - \lambda_{j,ii}^{1/2} \lambda_{j,(i+1)(i+1)} \mathbf{e}_{j,d_{i+1}} \\
    & \hspace{-2.7cm} = \lambda_{j,(i+1)(i+1)}^{1/2} \lambda_{j,ii}^{1/2} \lambda_{j,(i+1)(i+1)}^{1/2} \mathbf{e}_{j,d_{i+1}}  - \lambda_{j,ii}^{1/2} \lambda_{j,(i+1)(i+1)} \mathbf{e}_{j,d_{i+1}} \\  & \hspace{-2.7cm} = \mathbf{0}_{d_{i+1}},
\end{split}
\end{equation*}
since $\boldsymbol{\Pi}_{(i+1)i} \mathbf{e}_{j,d_{i}} = \mathbf{e}_{j,d_{i+1}}$, as $j \leq d_{i}$. Finally, for $r \in \{i+2,\dots,k\}$, we obtain
\begin{equation*} 
\begin{split}
  \sum_{\ell=1}^{i-1} \boldsymbol{\Lambda}_{rr}^{1/2}\boldsymbol{\Pi}_{rl} \boldsymbol{\Lambda}_{\ell\ell}^{1/2}\mathbf{0}_{d_{\ell}} & + \boldsymbol{\Lambda}_{rr}^{1/2} \boldsymbol{\Pi}_{ri} \boldsymbol{\Lambda}_{ii}^{1/2} \lambda_{j,(i+1)(i+1)}^{1/2} \mathbf{e}_{j,d_{i}} - \boldsymbol{\Lambda}_{rr}^{1/2} \boldsymbol{\Pi}_{r(i+1)} \boldsymbol{\Lambda}_{(i+1)(i+1)}^{1/2} \lambda_{j,ii}^{1/2} \mathbf{e}_{j,d_{i+1}} + \boldsymbol{\Lambda}_{rr} \mathbf{0}_{d_{r}} + \sum_{\substack{\ell= i+2 \\ l\neq r}}^{k} \boldsymbol{\Lambda}_{rr}^{1/2}\boldsymbol{\Pi}_{r\ell} \boldsymbol{\Lambda}_{\ell\ell}^{1/2}\mathbf{0}_{d_{\ell}} \\
    & \hspace{-2.7cm}  = \lambda_{j,(i+1)(i+1)}^{1/2} \lambda_{j,ii}^{1/2} \boldsymbol{\Lambda}_{rr}^{1/2} \boldsymbol{\Pi}_{ri}\mathbf{e}_{j,d_{i}} - \lambda_{j,ii}^{1/2} \lambda_{j,(i+1)(i+1)}^{1/2} \boldsymbol{\Lambda}_{rr}^{1/2} \boldsymbol{\Pi}_{r(i+1)} \mathbf{e}_{j,d_{i+1}}  \\ & \hspace{-2.7cm} = \mathbf{0}_{d_{r}}, 
\end{split}
\end{equation*} 
since
$\boldsymbol{\Pi}_{ri} \mathbf{e}_{j,d_{i}} = \boldsymbol{\Pi}_{r(i+1)} \mathbf{e}_{j,d_{i+1}}$, and we come by $d_{1} + \dots + d_{k-1}$ eigenvalues, resulting in the desired total of $d_{1} + \dots + d_{k} = q$ eigenvalues. These are indeed given by \begin{equation*}
    \boldsymbol{\lambda}(\mathbf{R}_{m}) = (\lambda_{j,11} + \lambda_{j,22} + \dots + \lambda_{j,kk})_{j=1}^{q}.
\end{equation*} 
Next, we need to show that for an arbitrary $\mathbf{A} \in \Gamma(\mathbf{R}_{11},\dots,\mathbf{R}_{kk})$ with eigenvalues $\lambda_{1} \geq \lambda_{2} \geq \dots \geq \lambda_{q}$, it holds that
\begin{equation}\label{eq: tproof} 
\begin{split}
    \sum_{\ell=1}^{n} \lambda_{\ell} & \leq \sum_{\ell=1}^{w}d_{\ell} + \sum_{\ell=1}^{n} (\lambda_{\ell,(w+1)(w+1)} + \dots + \lambda_{\ell,kk})  \hspace{0.2cm} \text{for all} \hspace{0.2cm} w  = 0, \dots, k-1,  n  = d_{w} + 1, \dots, d_{w+1}.
\end{split}
\end{equation} 
For this, we first introduce some notation. We know that $\mathbf{A}$ must be of the form
\begin{equation*}
\mathbf{A}= 
\begin{pmatrix}
\mathbf{R}_{11} & \mathbf{Q}_{12} & \cdots & \mathbf{Q}_{1k} \\
\mathbf{Q}_{12}^{\text{T}} & \mathbf{R}_{22} & \cdots & \mathbf{Q}_{2k} \\
\vdots & \vdots & \ddots & \vdots \\
\mathbf{Q}_{1k}^{\text{T}} & \mathbf{Q}_{2k}^{\text{T}} & \cdots & \mathbf{R}_{kk}
\end{pmatrix} \in \mathbb{R}^{q \times q},
\end{equation*}
with $\mathbf{Q}_{ij} \in \mathbb{R}^{d_{i} \times d_{j}}$. Based on this, define the matrices 
\begin{equation*}
    \mathbf{B}_{i} = \begin{pmatrix}
    \mathbf{R}_{(i+1)(i+1)} & \mathbf{Q}_{(i+1)(i+2)} & \cdots & \mathbf{Q}_{(i+1)k} \\
    \mathbf{Q}_{(i+1)(i+2)}^{\text{T}} & \mathbf{R}_{(i+2)(i+2)} & \cdots & \mathbf{Q}_{(i+2)k} \\
    \vdots & \vdots & \ddots & \vdots \\
    \mathbf{Q}_{(i+1)k}^{\text{T}} & \mathbf{Q}_{(i+2)k}^{\text{T}} & \cdots & \mathbf{R}_{kk}
    \end{pmatrix} \hspace{0.3cm} \text{and} \hspace{0.3cm}
    \mathbf{C}_{i} = \begin{pmatrix}
    \mathbf{R}_{ii} & \mathbf{Q}_{i} \\
    \mathbf{Q}_{i}^{\text{T}} & \mathbf{B}_{i}
    \end{pmatrix},
\end{equation*}
where $\mathbf{Q}_{i} = (\mathbf{Q}_{i(i+1)} \cdots \mathbf{Q}_{ik}) \in \mathbb{R}^{d_{i} \times (d_{i+1} \cdots + d_{k})}$, for all $i = 1,\dots,k-2$. We will approach $\mathbf{C}_{i}$ as a block matrix consisting of four blocks. Note that $\mathbf{B}_{i} \in \mathbb{R}^{(d_{i+1} + \cdots + d_{k}) \times(d_{i+1} + \cdots + d_{k})}$ and $\mathbf{C}_{i} \in \mathbb{R}^{(d_{i} + \cdots + d_{k}) \times (d_{i} + \cdots + d_{k})}$. Further, put 
\begin{equation*}
    \begin{split}
        \gamma_{1,i} & \geq \gamma_{2,i} \geq \cdots \geq \gamma_{m_{i},i} \hspace{0.35cm} \text{the eigenvalues of} \hspace{0.2cm} \mathbf{C}_{i} \hspace{0.2cm} \text{for} \hspace{0.2cm} i = 1,\dots,k-2, m_{i} = d_{i} + \dots + d_{k} \\
        \lambda_{1,ii} & \geq \lambda_{2,ii} \geq \cdots \geq \lambda_{d_{i},ii} \hspace{0.25cm} \text{the eigenvalues of} \hspace{0.2cm} \mathbf{R}_{ii} \hspace{0.11cm} \text{for} \hspace{0.2cm} i = 1, \dots, k \\
        \mu_{1,i} & \geq \mu_{2,i} \geq \cdots \geq \mu_{t_{i},i} \hspace{0.45cm} \text{the eigenvalues of} \hspace{0.2cm} \mathbf{B}_{i} \hspace{0.2cm} \text{for} \hspace{0.2cm} i = 1,\dots,k-2, t_{i} = d_{i+1} + \dots + d_{k}.
    \end{split}
\end{equation*}
Then, by applying Theorem 1 in \cite{Thompson1972} to $\mathbf{C}_{i}$ for an arbitrary $i \in \{1,\dots,k-2\}$ and also to $\mathbf{B}_{k-2}$, we can pick any $\alpha_{i}, \beta_{i}$ satisfying 
\begin{equation*}
 \begin{split}
     0 & \leq \alpha_{i} \leq d_{i} \hspace{0.2cm} \text{for} \hspace{0.2cm} i = 1,\dots,k \\
     0 & \leq \beta_{i} \leq t_{i} \hspace{0.28cm} \text{for} \hspace{0.2cm} i = 1,\dots,k-2
\end{split}
\end{equation*}
and integers 
\begin{equation*}
    \begin{split}
        1 & \leq r_{1,i} < \cdots < r_{\alpha_{i},i} \leq d_{i}, \hspace{1cm} r_{\ell,i} = d_{i} - \alpha_{i} + \ell \hspace{0.2cm} \text{for} \hspace{0.2cm} \ell > \alpha_{i}, \hspace{1cm} \text{for} \hspace{0.2cm} i = 1,\dots,k \\
        1 & \leq j_{1,i} < \cdots < j_{\beta_{i},i} \leq t_{i},  \hspace{1.05cm} j_{\ell,i} = t_{i} - \beta_{i} + \ell \hspace{0.2cm} \text{for} \hspace{0.21cm} \ell > \beta_{i}, \hspace{1.15cm} \text{for} \hspace{0.2cm} i = 1,\dots,k-2,
    \end{split}
\end{equation*} 
that, defined as such, will guarantee that
\begin{equation}\label{eq: eqhj}
\begin{split}
    \sum_{\ell=1}^{\alpha_{i} + \beta_{i}} \gamma_{r_{\ell,i} + j_{\ell,i} - \ell,i} & \leq \sum_{\ell=1}^{\alpha_{i}} \lambda_{r_{\ell,i},ii} + \sum_{\ell=1}^{\beta_{i}} \mu_{j_{\ell,i},i} \hspace{0.2cm} \text{for} \hspace{0.2cm} i = 1, \dots, k-2 \\
   \sum_{\ell=1}^{\alpha_{k-1}+\alpha_{k}} \mu_{r_{\ell,k-1}+r_{\ell,k}-\ell,k-2}  & \leq \sum_{\ell=1}^{\alpha_{k-1}} \lambda_{r_{\ell,k-1},(k-1)(k-1)} + \sum_{\ell=1}^{\alpha_{k}} \lambda_{r_{\ell,k},kk}. 
\end{split}
\end{equation}
Notice that \eqref{eq: eqhj} describes a total of $k-1$ inequalities. We will now prove \eqref{eq: tproof} by making an adequate choice for the $\alpha_{i}, \beta_{i}$ and corresponding indices $r_{\ell,i},j_{\ell,i}$. In particular, when considering a fixed $w \in \{0,\dots,k-1\}$ and fixed $n \in \{d_{w}+1,\dots,d_{w+1}\}$, take
\begin{equation*} 
\begin{split}
    & \hspace{4cm} \alpha_{i} = d_{i}, \hspace{0.1cm} r_{\ell,i} = \ell \hspace{0.2cm} \text{for} \hspace{0.2cm} \ell = 1,\dots,d_{i}, i = 1,\dots, w \\
    & \hspace{4.1cm} \alpha_{i} = n, \hspace{0.1cm} r_{\ell,i} = \ell \hspace{0.2cm} \text{for} \hspace{0.2cm} \ell = 1,\dots,n , i = w+1, \dots, k \\
    & \beta_{i} = t_{i} - t_{w} + (k-w)n, \hspace{0.1cm} j_{\ell,i} = \begin{cases} \ell & \mbox{if } \ell \in \{1,\dots,n\}\\ t_{w} - (k-w)n + \ell, & \mbox{if } \ell \in \{n+1,\dots,\beta_{i}\}  \end{cases} \hspace{0.2cm} \text{for} \hspace{0.2cm} i = 1,\dots,\min\{k-2,w\} \\
    & \hspace{1.4cm} \beta_{i} = (k-i)n, \hspace{0.1cm} j_{\ell,i} = \begin{cases} \ell & \mbox{if } \ell \in \{1,\dots,n\}\\ t_{i} - (k-i)n + \ell, & \mbox{if } \ell \in \{n+1,\dots,\beta_{i}\}  \end{cases} \hspace{0.2cm} \text{for} \hspace{0.2cm} i = w+1,\dots,k-2.
\end{split}
\end{equation*} 
Then, concerning the $k-2$' th and $k-1$' th inequality of \eqref{eq: eqhj}, we see that if $w = k-1$, we have
\begin{equation}\label{eq: eqq1}
\begin{split}
    \sum_{\ell=1}^{\beta_{k-2}} \mu_{j_{\ell,k-2},k-2} & = \sum_{\ell=1}^{n} \mu_{\ell,k-2} + \sum_{\ell = n+1}^{\beta_{k-2}} \mu_{d_{k}-n+\ell,k-2} \\
    & = \mu_{1,k-2} + \dots + \mu_{n,k-2} + \mu_{d_{k}+1,k-2} + \dots + \mu_{t_{k-2},k-2},
\end{split}
\end{equation}
and,
\begin{equation}\label{eq: eqq2}
    \begin{split}
        \sum_{\ell=1}^{\alpha_{k-1}+\alpha_{k}} \mu_{r_{\ell,k-1}+r_{\ell,k}-\ell,k-2} & = \sum_{\ell=1}^{d_{k-1}+n} \mu_{r_{\ell,k-1} + r_{\ell,k}-\ell,k-2} \\
        & = \sum_{\ell=1}^{d_{k-1}} \mu_{\ell,k-2} + \sum_{\ell=d_{k-1}+1}^{n}\mu_{\ell,k-2} + \sum_{\ell=n+1}^{n+d_{k-1}}\mu_{d_{k}-n+\ell,k-2} \\
        & = \mu_{1,k-2} + \dots + \mu_{n,k-2} + \mu_{d_{k}+1,k-2} + \dots + \mu_{d_{k-1}+d_{k},k-2}.
    \end{split}
\end{equation}
Since $t_{k-2} = d_{k-1} + d_{k}$, we see that \eqref{eq: eqq1} and \eqref{eq: eqq2} are equal. If on the other hand $w \in \{0,\dots,k-2\}$, we get
\begin{equation}\label{eq: eqq3}
    \begin{split}
        \sum_{\ell=1}^{\beta_{k-2}} \mu_{j_{\ell,k-2},k-2} & = \sum_{\ell=1}^{n} \mu_{\ell,k-2} + \sum_{\ell=n+1}^{\beta_{k-2}} \mu_{t_{k-2}-2n+\ell,k-2} \\
        & = \mu_{1,k-2} + \dots + \mu_{n,k-2} + \mu_{t_{k-2}-n+1,k-2} + \dots + \mu_{t_{k-2},k-2},
    \end{split}
\end{equation}
and,
\begin{equation}\label{eq: eqq4}
\begin{split}
    \hspace{-0.2cm} \sum_{\ell=1}^{\alpha_{k-1}+\alpha_{k}} \mu_{r_{\ell,k-1}+r_{\ell,k}-\ell,k-2} & = \sum_{\ell=1}^{2n} \mu_{r_{\ell,k-1}+r_{\ell,k}-\ell,k-2} \\
    & = \sum_{\ell=1}^{n}\mu_{\ell,k-2} + \sum_{\ell=n+1}^{2n} \mu_{d_{k-1}+d_{k}-2n+\ell,k-2} \\
    & = \mu_{1,k-2} + \dots + \mu_{n,k-2} + \mu_{d_{k-1}+d_{k}-n+1,k-2} + \dots + \mu_{d_{k-1} +d_{k},k-2}.
\end{split}    
\end{equation}
Again because $t_{k-2} = d_{k-1} + d_{k}$, we have equality between \eqref{eq: eqq3} and \eqref{eq: eqq4}. We thus have shown so far that the second term on the right side of the $k-2$' th inequality of \eqref{eq: eqhj} equals the term on the left side of the $k-1$' th inequality. Now, concerning the $i$' th and $i+1$' th inequality for $i \in \{1,\dots,k-3\}$, we must have either both $i, i+1 \in \{w+1,\dots,k-2\}$ or both $i,i+1 \in \{1,\dots,w\}$, or $i \in \{1,\dots,w\}, i+1 \in \{w+1,\dots,k-2\}$ meaning $i = w$. In the first case, it holds that
\begin{align*} 
    \sum_{\ell=1}^{\alpha_{i+1}+\beta_{i+1}} \gamma_{r_{\ell,i+1}+j_{\ell,i+1}-\ell,i+1} &  = \sum_{\ell=1}^{n+\beta_{i+1}} \gamma_{r_{\ell,i+1}+j_{\ell,i+1}-\ell,i+1} \notag \\
    & = \sum_{\ell=1}^{n} \gamma_{\ell,i+1} + \sum_{\ell=n+1}^{\beta_{i+1}} \gamma_{d_{i+1}-n+t_{i+1}-(k-i-1)n+\ell,i+1} \stepcounter{equation}\tag{\theequation}\label{eq: eqq5} + \sum_{\ell= \beta_{i+1}+1}^{n+\beta_{i+1}}\gamma_{d_{i+1}-n+t_{i+1}-(k-i-1)n+\ell,i+1} \notag \\
    & = \gamma_{1,i+1} + \dots + \gamma_{n,i+1} + \gamma_{d_{i+1}+t_{i+1}-(k-i-1)n+1,i+1} + \dots + \gamma_{d_{i+1}+t_{i+1},i+1} \notag,
\end{align*} 
and
\begin{equation}\label{eq: eqq6}
    \begin{split}
        \sum_{\ell=1}^{\beta_{i}} \mu_{j_{\ell,i},i} & = \sum_{\ell=1}^{n} \mu_{\ell,i} + \sum_{\ell=n+1}^{\beta_{i}} \mu_{t_{i} - (k-i)n+\ell,i} = \mu_{1,i} + \dots + \mu_{n,i} + \mu_{t_{i}-(k-i)n+n+1,i} + \dots + \mu_{t_{i},i}.
    \end{split}
\end{equation}
Since $t_{i} = d_{i+1} + t_{i+1}$ and $\mathbf{C}_{i+1} = \mathbf{B}_{i}$ (by construction, hence they also have the same eigenvalues), we see that \eqref{eq: eqq5} is equal to \eqref{eq: eqq6}. For the second case, i.e., both $i,i+1 \in \{1,\dots,w\}$, we have
\begin{align*}
        \sum_{\ell=1}^{\alpha_{i+1}+\beta_{i+1}} \gamma_{r_{\ell,i+1}+j_{\ell,i+1}-\ell,i+1} & = \sum_{\ell=1}^{d_{i+1}+\beta_{i+1}} \gamma_{r_{\ell,i+1}+j_{\ell,i+1}-\ell,i+1} \\
        & = \sum_{\ell=1}^{d_{i+1}} \gamma_{\ell,i+1} + \sum_{\ell=d_{i+1}+1}^{n} \gamma_{\ell,i+1} + \sum_{\ell=n+1}^{\beta_{i+1}} \gamma_{t_{w}-(k-w)n+\ell,i+1} \stepcounter{equation}\tag{\theequation}\label{eq: eqq7} + \sum_{\ell = \beta_{i+1}+1}^{d_{i+1}+\beta_{i+1}} \gamma_{t_{w}-(k-w)n+\ell,i+1}  \\
        & = \gamma_{1,i+1} + \dots + \gamma_{n,i+1} + \gamma_{t_{w}-(k-w-1)n+1,i+1} + \dots + \gamma_{d_{i+1}+t_{i+1},i+1},
\end{align*}
and
\begin{equation}\label{eq: eqq8}
    \begin{split}
        \sum_{\ell=1}^{\beta_{i}} \mu_{j_{\ell,i},i} & = \sum_{\ell=1}^{n} \mu_{\ell,i} + \sum_{\ell=n+1}^{\beta_{i}} \mu_{t_{w}-(k-w)n+\ell,i} = \mu_{1,i} + \dots + \mu_{n,i} + \mu_{t_{w}-(k-w-1)n+1,i} + \dots + \mu_{t_{i},i}.
    \end{split}
\end{equation}
For the same reasons as in the previous case, we clearly see that \eqref{eq: eqq7} and \eqref{eq: eqq8} are the same expressions. Finally, if $i = w$, 
\begin{align*} 
        \sum_{\ell=1}^{\alpha_{i+1}+\beta_{i+1}} \gamma_{r_{\ell,i+1}+j_{\ell,i+1}-\ell,i+1} & = \sum_{\ell=1}^{n+\beta_{w+1}} \gamma_{r_{\ell,w+1}+j_{\ell,w+1}-\ell,w+1} \stepcounter{equation}\tag{\theequation}\label{eq: eqq9} \\
        & = \gamma_{1,w+1} + \dots + \gamma_{n,w+1} + \gamma_{d_{w+1}+t_{w+1}-(k-w-1)n+1,w+1} + \dots + \gamma_{d_{w+1} + t_{w+1},w+1},
\end{align*}
and,
\begin{equation}\label{eq: eqq10}
    \begin{split}
        \sum_{\ell=1}^{\beta_{i}} \mu_{j_{\ell,i},i} & = \sum_{\ell=1}^{\beta_{w}} \mu_{j_{\ell,w},w} = \mu_{1,w} + \dots + \mu_{n,w} + \mu_{t_{w}-(k-w-1)n+1,w} + \dots + \mu_{t_{w},w}.
    \end{split}
\end{equation}
Since $t_{w} = d_{w+1} + t_{w+1}$ and $\mathbf{C}_{w+1} = \mathbf{B}_{w}$, we also have equality between \eqref{eq: eqq9} and \eqref{eq: eqq10}. All together, we have proven that the second term on the right hand side of the $i$' th inequality of \eqref{eq: eqhj} is equal to the term on the left hand side of the $i+1$' th inequality for all $i = 1,\dots,k-2$. Hence \eqref{eq: eqhj} can be reduced to
\begin{equation}\label{eq: eqq11}
    \sum_{\ell=1}^{\alpha_{1} + \beta_{1}} \gamma_{r_{\ell,1}+j_{\ell,1}-\ell,1} \leq \sum_{\ell=1}^{w} d_{\ell} + \sum_{\ell=1}^{n} (\lambda_{\ell,(w+1)(w+1)} + \dots + \lambda_{\ell,kk}).
\end{equation}
Moreover, since by choice $r_{\ell,1} = \ell$ for all $\ell = 1,\dots, n$ if $w = 0$, or $r_{\ell,1} = \ell$ for all $\ell = 1, \dots, d_{1}$ and $r_{\ell,1} = d_{1}-d_{1}+\ell = \ell$ for all $\ell > d_{1}$ if $w > 0$, and $j_{\ell,1} = \ell$ for all $\ell = 1,\dots,n$, we definitely have
\begin{equation}\label{eq: eqq12}
    \sum_{\ell=1}^{n} \gamma_{\ell,1} \leq \sum_{\ell=1}^{\alpha_{1} + \beta_{1}} \gamma_{r_{\ell,1}+j_{\ell,1}-\ell,1}.
\end{equation}
The inequality in \eqref{eq: eqq12} is an equality if there are $\alpha_{1} + \beta_{1} - n$ eigenvalues equal to zero. If for example $w = k-1$ and $n = d_{k}$, this means that $d_{1}+t_{1}-t_{k-1}+n-n = q-d_{k}$ eigenvalues should be zero, which is the case for the matrix $\mathbf{R}_{m}$. Combining \eqref{eq: eqq11} with \eqref{eq: eqq12} and realizing that the eigenvalues $\gamma_{\ell,1}$ are the same as the eigenvalues $\lambda_{\ell}$, since $\mathbf{A} = \mathbf{C}_{1}$, finishes the proof.\hfill \qedsymbol

\subsection{Proof of Proposition \ref{prop3}}\label{App C2}

Consider an arbitrary $\mathbf{A} \in \Gamma(\mathbf{R}_{11},\dots,\mathbf{R}_{kk})$ having ordered eigenvalues $\lambda_{1} \geq \dots \geq \lambda_{q}$. Then, by definition of the Bures-Wasserstein distance \eqref{eq: Bures},
\begin{equation}\label{eq: eqqq1}
    d_{W}^{2}(\mathbf{A},\mathbf{I}_{q}) = q + \text{tr}(\mathbf{A}) - 2\text{tr}\left(\mathbf{A}^{1/2}\right) = q + \sum_{j=1}^{q} \lambda_{j} - 2 \sum_{j=1}^{q} \lambda_{j}^{1/2}. 
\end{equation}
Notice that the function $\lambda \mapsto \lambda - 2 \lambda^{1/2}$ is convex for $\lambda \in [0,\infty)$. By Proposition \ref{prop2}, we know that the eigenvalues of $\mathbf{R}_{m}$ majorize those of $\mathbf{A}$. From Lemma \ref{lem2}, it follows that \eqref{eq: eqqq1} is maximal if $\mathbf{A} = \mathbf{R}_{m}$. Secondly, we have
\begin{equation*}
    \mathbf{R}_{0}^{1/2}\mathbf{A}\mathbf{R}_{0}^{1/2} = 
    \begin{pmatrix} \vspace{0.1cm}
    \mathbf{R}_{11}^{2} & \mathbf{R}_{11}^{1/2} \boldsymbol{\Psi}_{12} \mathbf{R}_{22}^{1/2} & \cdots & \mathbf{R}_{11}^{1/2} \boldsymbol{\Psi}_{1k} \mathbf{R}_{kk}^{1/2} \\
    \mathbf{R}_{22}^{1/2} \boldsymbol{\Psi}_{12}^{\text{T}} \mathbf{R}_{11}^{1/2} & \mathbf{R}_{22}^{2} & \cdots & \mathbf{R}_{22}^{1/2} \boldsymbol{\Psi}_{2k} \mathbf{R}_{kk}^{1/2} \\
    \vdots & \vdots & \ddots & \vdots \\
    \mathbf{R}_{kk}^{1/2} \boldsymbol{\Psi}_{1k}^{\text{T}} \mathbf{R}_{11}^{1/2} & \mathbf{R}_{kk}^{1/2} \boldsymbol{\Psi}_{2k}^{\text{T}} \mathbf{R}_{22}^{1/2} & \cdots & \mathbf{R}_{kk}^{2} 
    \end{pmatrix}.
\end{equation*}
Recall that $\mathbf{R}^{\alpha}_{ii} = \mathbf{U}_{ii}\boldsymbol{\Lambda}_{ii}^{\alpha} \mathbf{U}_{ii}^{\text{T}}$ is the eigendecomposition of $\mathbf{R}_{ii}^{\alpha}$ for $\alpha > 0$, and for $\mathbf{R}_{m}$ we have $\boldsymbol{\Psi}_{ij} = \mathbf{U}_{ii}\boldsymbol{\Lambda}_{ii}^{1/2}\boldsymbol{\Pi}_{ij}\boldsymbol{\Lambda}_{jj}^{1/2}\mathbf{U}_{jj}^{\text{T}}$. Hence, we get
\begin{equation*}
    \mathbf{R}_{ii}^{1/2} \boldsymbol{\Psi}_{ij} \mathbf{R}_{jj}^{1/2} = \left (\mathbf{U}_{ii}\boldsymbol{\Lambda}_{ii}^{1/2} \mathbf{U}_{ii}^{\text{T}} \right ) \left ( \mathbf{U}_{ii}\boldsymbol{\Lambda}_{ii}^{1/2}\boldsymbol{\Pi}_{ij}\boldsymbol{\Lambda}_{jj}^{1/2}\mathbf{U}_{jj}^{\text{T}} \right ) \left (\mathbf{U}_{jj}\boldsymbol{\Lambda}_{jj}^{1/2} \mathbf{U}_{jj}^{\text{T}} \right ) = \mathbf{U}_{ii} \boldsymbol{\Lambda}_{ii} \boldsymbol{\Pi}_{ij} \boldsymbol{\Lambda}_{jj} \mathbf{U}_{jj}^{\text{T}},
\end{equation*}
which is of the same form as $\boldsymbol{\Psi}_{ij}$, but with $\boldsymbol{\Lambda}_{ii}^{1/2}$ and $\boldsymbol{\Lambda}_{jj}^{1/2}$ replaced by $\boldsymbol{\Lambda}_{ii}$ and $\boldsymbol{\Lambda}_{jj}$. Applying Proposition \ref{prop2} with $\mathbf{R}_{ii}^{2}$ instead of $\mathbf{R}_{ii}$ for $i = 1,\dots,k$, it follows that the eigenvalues of an arbitrary matrix in $\Gamma(\mathbf{R}_{11}^{2},\dots,\mathbf{R}_{kk}^{2})$ are majorized by those of the matrix $\mathbf{R}_{0}^{1/2} \mathbf{R}_{m} \mathbf{R}_{0}^{1/2}$. Hence, 
\begin{equation}\label{eq: eqqq2}
    d_{W}^{2}(\mathbf{A},\mathbf{R}_{0}) = 2 \text{tr}(\mathbf{A}) - 2 \text{tr} \left \{ \left (\mathbf{R}_{0}^{1/2} \mathbf{A} \mathbf{R}_{0}^{1/2} \right )^{1/2} \right \} = 2 \text{tr}(\mathbf{A}) - 2 \sum_{j=1}^{q} \kappa_{j}^{1/2}
\end{equation}
with $\kappa_{1},\dots,\kappa_{q}$ the eigenvalues of $\mathbf{R}_{0}^{1/2}\mathbf{A}\mathbf{R}_{0}^{1/2}$. The fact that \eqref{eq: eqqq2} is maximal if $\mathbf{A} = \mathbf{R}_{m}$ now follows from the convexity of the function $\kappa \mapsto - \kappa^{1/2}$ on $[0,\infty)$ and Lemma \ref{lem2}. \hfill \qedsymbol 

\subsection{Additional lemmas and proof of Theorem \ref{thm1}}\label{App C3}

\begin{lemma}\label{lem3}
Under the conditions of Theorem \ref{thm1}, it holds that, for $\mathbf{H}_{t} \in \mathbb{S}^{q}$ for $t > 0$ and $\mathbf{H} \in \mathbb{S}^{q}$ such that $||\mathbf{H}_{t} - \mathbf{H}||_{\text{F}} \to 0$ as $t \to 0$, 
\begin{equation*}
    \lim_{t\,\underset{>}{\to}\,0} \frac{\text{tr} \left \{ (\mathbf{R} + t \mathbf{H}_{t})_{m}^{1/2} \right \} - \text{tr} \left (\mathbf{R}_{m}^{1/2} \right )}{t} = \frac{1}{2} \text{tr} \left (\boldsymbol{\Upsilon}_{1}\mathbf{H} \right ),
\end{equation*}
with $(\mathbf{R} + t \mathbf{H}_{t})_{m}$ the matrix in \eqref{eq: Rmhelp} for $\mathbf{R}$ replaced by $\mathbf{R} + t \mathbf{H}_{t}$ and $\boldsymbol{\Upsilon}_{1}$ defined in \eqref{eq: ups1}.
\end{lemma}

\begin{proof}
Define the map $L : (\mathbb{S}^{q},||\cdot||_{\text{F}}) \rightarrow (\mathbb{S}^{q},||\cdot||_{\text{F}}) : \mathbf{A} \mapsto L(\mathbf{A})$ via $L(\mathbf{A})$ being the diagonal matrix whose diagonal is equal to the $q$ eigenvalues (counting multiplicities) of $\mathbf{A}$ in decreasing order. Then, we have $L(\mathbf{R}_{ii}) = \boldsymbol{\Lambda}_{ii} = \text{diag}(\lambda_{1,ii},\dots,\lambda_{d_{i},ii})$ for $i = 1,\dots,k$. Also define the map
\begin{equation*}
    M : (\mathbb{S}^{d_{1}}_{>} \times \cdots \times \mathbb{S}^{d_{k}}_{>},||\cdot||_{\text{F}}) \rightarrow (\mathbb{S}^{d_{1}}_{>} \times \cdots \times \mathbb{S}^{d_{k}}_{>},||\cdot||_{\text{F}}) : (\mathbf{A}_{1},\dots,\mathbf{A}_{k}) \mapsto \left (L(\mathbf{A}_{1}),\dots,L(\mathbf{A}_{k}) \right ),
\end{equation*}
where the Frobenius norm is naturally defined for a certain $(\mathbf{A}_{1},\dots,\mathbf{A}_{k}) \in \mathbb{S}^{d_{1}} \times \cdots \times \mathbb{S}^{d_{k}}$ through $||(\mathbf{A}_{1},\dots,\mathbf{A}_{k})||_{\text{F}} =  (||\mathbf{A}_{1}||_{\text{F}}^{2} + \cdots + ||\mathbf{A}_{k}||_{\text{F}}^{2} )^{1/2}$. Using Lemma 3 of \cite{Mordant2022}, it is then quickly seen that the Fr\'echet derivative of $M$ at $(\mathbf{R}_{11},\dots,\mathbf{R}_{kk})$ in the direction of a certain $\mathbf{H} = (\mathbf{H}_{11},\dots,\mathbf{H}_{kk}) \in \mathbb{S}^{d_{1}}_{>} \times \cdots \times \mathbb{S}^{d_{k}}_{>}$ is $\boldsymbol{\Delta}(\mathbf{H}) = (\mathbf{D}_{\mathbf{U}_{11}^{\text{T}}\mathbf{H}_{11}\mathbf{U}_{11}}, \dots,\mathbf{D}_{\mathbf{U}_{kk}^{\text{T}}\mathbf{H}_{kk}\mathbf{U}_{kk}})$, with $\mathbf{D}_{\mathbf{U}_{ii}^{\text{T}}\mathbf{H}_{ii}\mathbf{U}_{ii}}$ the diagonal matrix containing the diagonal of $\mathbf{U}_{ii}^{\text{T}}\mathbf{H}_{ii}\mathbf{U}_{ii}$. Consider now (recalling that $d_{1} \leq \cdots \leq d_{k}$)
\begin{equation*}
    g : (\mathbb{S}^{d_{1}}_{>} \times \cdots \times \mathbb{S}^{d_{k}}_{>},||\cdot||_{\text{F}}) \rightarrow (\mathbb{R},|\cdot|) : (\mathbf{B}_{1},\dots,\mathbf{B}_{k}) \mapsto \sum_{j=1}^{d_{k}} (\eta_{j,11} + \cdots + \eta_{j,kk})^{1/2},
\end{equation*}
with $\eta_{j,ii}$ the eigenvalues of $\mathbf{B}_{i}$ in decreasing order for $i = 1,\dots,k$ and $j = 1,\dots,d_{i}$, and putting $\eta_{j,ii} = 0$ for $j = d_{i}+1,\dots,d_{k}$. We see that $\text{tr}(\mathbf{R}_{m}^{1/2}) = g(L(\mathbf{R}_{11}),\dots,L(\mathbf{R}_{kk})) = g(\boldsymbol{\Lambda}_{11},\dots,\boldsymbol{\Lambda}_{kk})$, also putting $\lambda_{j,ii} = 0$ for $j = d_{i} + 1, \dots, d_{k}$. The Fr\'echet derivative of $g$ at $(\boldsymbol{\Lambda}_{11},\dots,\boldsymbol{\Lambda}_{kk})$ in the direction of $\boldsymbol{\Delta}(\mathbf{H})$ equals $\Psi(\boldsymbol{\Delta}(\mathbf{H}))$, where $\Psi : \mathbb{S}^{d_{1}}_{>} \times \cdots \times \mathbb{S}^{d_{k}}_{>} \rightarrow \mathbb{R}$ satisfies 
\begin{equation}\label{eq: lmp} 
    \lim_{||\boldsymbol{\Delta}(\mathbf{H})||_{\text{F}} \to 0} \frac{|g(\boldsymbol{\Lambda}_{11} + \mathbf{D}_{\mathbf{U}_{11}^{\text{T}}\mathbf{H}_{11}\mathbf{U}_{11}}, \dots, \boldsymbol{\Lambda}_{kk} + \mathbf{D}_{\mathbf{U}_{kk}^{\text{T}}\mathbf{H}_{kk}\mathbf{U}_{kk}}) - g(\boldsymbol{\Lambda}_{11},\dots,\boldsymbol{\Lambda}_{kk})-\Psi(\boldsymbol{\Delta}(\mathbf{H}))|}{||\boldsymbol{\Delta}(\mathbf{H})||_{\text{F}}} = 0.
\end{equation}
Defining $t_{j,ii} = (\mathbf{U}_{ii}^{\text{T}} \mathbf{H}_{ii} \mathbf{U}_{ii})_{jj}$ for all $i = 1,\dots,k$ and $j = 1,\dots,d_{i}$, as well as $t_{j,ii} = 0$ for $j = d_{i} + 1,\dots,d_{k}$, it is obvious that the eigenvalues of $\boldsymbol{\Lambda}_{ii} + \mathbf{D}_{\mathbf{U}_{ii}^{\text{T}}\mathbf{H}_{ii}\mathbf{U}_{ii}}$ are $\lambda_{j,ii} + t_{j,ii}$. Moreover, since we assumed that the eigenvalues $\lambda_{j,ii}$ are distinct over $j$, it will hold, for $||\boldsymbol{\Delta}(\mathbf{H})||_{\text{F}}$ small enough, that $\lambda_{j,ii} + t_{j,ii}$ are in decreasing order again over $j$. Hence, using our definition of $g$, we see that \eqref{eq: lmp} becomes 
\begin{equation}\label{eq: hfy}
    \lim_{||\boldsymbol{\Delta}(\mathbf{H})||_{\text{F}} \to 0} \frac{\left |\sum_{j=1}^{d_{k}} \left \{ \sum_{i=1}^{k} (\lambda_{j,ii}+t_{j,ii}) \right \}^{1/2}- \sum_{j=1}^{d_{k}} \left (\sum_{i=1}^{k} \lambda_{j,ii} \right )^{1/2}  - \Psi(\boldsymbol{\Delta}(\mathbf{H}))\right |}{||\boldsymbol{\Delta}(\mathbf{H})||_{\text{F}}} = 0.
\end{equation}
From expression \eqref{eq: hfy}, we observe that $\Psi(\boldsymbol{\Delta}(\mathbf{H}))$ is nothing more than the total derivative of the function 
\begin{equation*}
    f : \mathbb{R}^{kd_{k}} \rightarrow \mathbb{R} : (x_{1,11},x_{1,22},\dots,x_{d_{k},kk}) \mapsto \sum_{j=1}^{d_{k}} \left (\sum_{i=1}^{k} x_{j,ii} \right )^{1/2}
\end{equation*}
at $(\lambda_{1,11},\lambda_{1,22},\dots,\lambda_{d_{k},kk})$ evaluated in $(t_{1,11},t_{1,22},\dots,t_{d_{k},kk})$. Computing the Jacobian matrix of $f$, we find
\begin{equation*}
\begin{split}
    \Psi(\boldsymbol{\Delta}(\mathbf{H})) & = \sum_{j=1}^{d_{k}} \left (\sum_{i=1}^{k} \frac{1}{2(\lambda_{j,11} + \cdots + \lambda_{j,kk})^{1/2}} t_{j,ii}\right ) \\
    & = \sum_{i=1}^{k} \left ( \sum_{j=1}^{d_{i}} \frac{1}{2(\lambda_{j,11} + \cdots + \lambda_{j,kk})^{1/2}} t_{j,ii} \right ) \\
    & = \sum_{i=1}^{k} \left ( \sum_{j=1}^{d_{i}} \frac{1}{2(\lambda_{j,11} + \cdots + \lambda_{j,kk})^{1/2}} \left ( \mathbf{U}_{ii}^{\text{T}}\mathbf{H}_{ii}\mathbf{U}_{ii}\right )_{jj} \right ).
\end{split}
\end{equation*}
By the chain rule, we have just shown that
\begin{equation*} 
\begin{split}
\lim_{t\,\underset{>}{\to}\,0} \frac{\text{tr} \left \{(\mathbf{R} + t \mathbf{H}_{t})_{m}^{1/2} \right \} - \text{tr}\left (\mathbf{R}_{m}^{1/2} \right )}{t} & = \lim_{t\,\underset{>}{\to}\,0} \frac{g(M(\mathbf{R}_{11} + t\mathbf{H}_{t,11}, \dots , \mathbf{R}_{kk} + t\mathbf{H}_{t,kk})) - g(M(\mathbf{R}_{11},\dots,\mathbf{R}_{kk}))}{t} \\
& =  \sum_{i=1}^{k} \left ( \sum_{j=1}^{d_{i}} \frac{1}{2(\lambda_{j,11} + \cdots + \lambda_{j,kk})^{1/2}} \left ( \mathbf{U}_{ii}^{\text{T}}\mathbf{H}_{ii}\mathbf{U}_{ii}\right )_{jj} \right ),
\end{split}
\end{equation*}
where $\mathbf{H}_{t,ii}$ is the $d_{i} \times d_{i}$ diagonal block of $\mathbf{H}_{t}$. Using the notations \eqref{eq: pp1}, \eqref{eq: pp2} and \eqref{eq: pp3}, we can further simplify the above expression into 
\begin{equation*}
\begin{split}
    \frac{1}{2} \sum_{i=1}^{k} \text{tr} \left (\boldsymbol{\Delta}_{i}\mathbf{U}_{ii}^{\text{T}} \mathbf{H}_{ii} \mathbf{U}_{ii} \right ) & = \frac{1}{2} \sum_{i=1}^{k} \text{tr} \left (\mathbf{P}_{i}^{\text{T}} \mathbf{U}_{ii} \boldsymbol{\Delta}_{i} \mathbf{U}_{ii}^{\text{T}} \mathbf{P}_{i} \mathbf{H} \right ) = \frac{1}{2} \text{tr}(\boldsymbol{\Upsilon}_{1}\mathbf{H}),
\end{split}
\end{equation*}
where we used that $\text{tr}(\mathbf{A}\hspace{0.04cm}\text{diag}(\mathbf{B})) = \text{tr}(\text{diag}(\mathbf{A})\mathbf{B})$ for square matrices $\mathbf{A},\mathbf{B}$, the cyclic permutation property of the trace operator, the fact that $\mathbf{H}_{ii} = \mathbf{P}_{i}\mathbf{H}\mathbf{P}_{i}^{\text{T}}$ and finally the identity $\sum_{i=1}^{k} \mathbf{P}_{i}^{\text{T}} \mathbf{A}_{i} \mathbf{P}_{i} = \text{diag}(\mathbf{A}_{1},\dots,\mathbf{A}_{k})$ for $(d_{i} \times d_{i})$ matrices $\mathbf{A}_{i}$.
\end{proof}

\begin{lemma}\label{lem4}
The Fr\'echet derivative of the map
\begin{equation*}
    \eta : \mathbb{S}^{q}_{>} \rightarrow \mathbb{R} : \mathbf{R} \mapsto \text{tr}\left \{ \left (\mathbf{R}_{0}^{1/2}\mathbf{R}\mathbf{R}_{0}^{1/2} \right )^{1/2} \right \}
\end{equation*}
is given by
\begin{equation*}
    \lim_{t\,\underset{>}{\to}\,0} \frac{\eta(\mathbf{R} + t\mathbf{H}_{t}) - \eta(\mathbf{R})}{t} = \frac{1}{2} \text{tr} \left \{ \left (\mathbf{J}_{0} + \mathbf{J}^{-1} \right )\mathbf{H} \right \},
\end{equation*}
for $\mathbf{H}_{t},\mathbf{H} \in \mathbb{S}^{q}$ with $||\mathbf{H}_{t} \to \mathbf{H}||_{\text{F}} \to 0$ as $t \to 0$ and $\mathbf{J}$ and $\mathbf{J}_{0}$ as in \eqref{eq: J} and \eqref{eq: J0} respectively. 
\end{lemma}
\begin{proof}
From Lemma 5 in \cite{Mordant2022}, the Fr\'echet derivative of the map
\begin{equation*}
    \phi : \mathbb{S}^{q}_{>} \times \mathbb{S}^{q}_{>} \rightarrow \mathbb{R} : (\mathbf{A},\mathbf{B}) \mapsto 2 \text{tr} \left \{ \left (\mathbf{A}^{1/2}\mathbf{B}\mathbf{A}^{1/2} \right )^{1/2} \right \}
\end{equation*}
at $(\mathbf{A},\mathbf{B}) \in \mathbb{S}^{q}_{>} \times \mathbb{S}^{q}_{>}$ in the direction of $(\mathbf{G},\mathbf{H}) \in \mathbb{S}^{q} \times \mathbb{S}^{q}$ equals $ \text{tr}(\mathbf{J}\mathbf{G}) + \text{tr}(\mathbf{J}^{-1}\mathbf{H})$ with
\begin{equation*}
    \begin{split}
        \mathbf{J} & = \mathbf{A}^{-1/2} \left (\mathbf{A}^{1/2}\mathbf{B}\mathbf{A}^{1/2} \right )^{1/2}\mathbf{A}^{-1/2} = \mathbf{B}^{1/2} \left (\mathbf{B}^{1/2}\mathbf{A}\mathbf{B}^{1/2} \right )^{-1/2}\mathbf{B}^{1/2} \\
        \mathbf{J}^{-1} & = \mathbf{A}^{1/2} \left (\mathbf{A}^{1/2}\mathbf{B}\mathbf{A}^{1/2} \right )^{-1/2}\mathbf{A}^{1/2} = \mathbf{B}^{-1/2}\left (\mathbf{B}^{1/2}\mathbf{A}\mathbf{B}^{1/2} \right )^{1/2}\mathbf{B}^{-1/2}.
    \end{split}
\end{equation*}
Applying this to $(\mathbf{A},\mathbf{B}) = (\mathbf{R}_{0},\mathbf{R})$ and $(\mathbf{G},\mathbf{H}) = (\mathbf{H}_{0},\mathbf{H})$ with $\mathbf{H}_{0}$ having the same  $d_{i} \times d_{i}$ diagonal blocks as $\mathbf{H}$, but zero off-diagonal blocks , yields $\frac{1}{2}(\text{tr}(\mathbf{J}\mathbf{H}_{0}) + \text{tr}(\mathbf{J}^{-1}\mathbf{H}))$ as the Fr\'echet derivative of $\eta$ at $\mathbf{R}$ in the direction of $\mathbf{H}$, with $\mathbf{J}$ given in \eqref{eq: J}. The result  follows from $\text{tr}(\mathbf{J}\mathbf{H}_{0}) = \text{tr}(\mathbf{J}_{0}\mathbf{H})$, with $\mathbf{J}_{0}$ given in \eqref{eq: J0}.
\end{proof}

\begin{lemma}\label{lem5}
Under the conditions of Theorem \ref{thm1}, it holds that, for $\mathbf{H}_{t} \in \mathbb{S}^{q}$ for $t > 0$ and $\mathbf{H} \in \mathbb{S}^{q}$ such that $||\mathbf{H}_{t} - \mathbf{H}||_{\text{F}} \to 0$ as $t \to 0$,
\begin{equation*}
    \lim_{t\,\underset{>}{\to}\,0} \frac{\text{tr}\left [\left \{(\mathbf{R} + t\mathbf{H}_{t})_{0}^{1/2}(\mathbf{R} + t \mathbf{H}_{t})_{m}(\mathbf{R}+t\mathbf{H}_{t})_{0}^{1/2} \right \}^{1/2} - \left (\mathbf{R}_{0}^{1/2}\mathbf{R}_{m}\mathbf{R}_{0}^{1/2} \right )^{1/2} \right ]}{t} = \text{tr}(\boldsymbol{\Upsilon_{2}\mathbf{H}}),
\end{equation*}
with $(\mathbf{R} + t \mathbf{H}_{t})_{0}$ having the same $d_{i} \times d_{i}$ diagonal blocks as $\mathbf{R} + t \mathbf{H}_{t}$ but zero off-diagonal blocks, $(\mathbf{R} + t\mathbf{H}_{t})_{m}$ the matrix in \eqref{eq: Rmhelp} with $\mathbf{R}$ replaced by $\mathbf{R} + t \mathbf{H}_{t}$, and with $\boldsymbol{\Upsilon}_{2}$ defined in \eqref{eq: ups2}.
\end{lemma}
\begin{proof}
Recall the proof of Lemma \ref{lem3}. Keep the same definition for the map $M$ and define
\begin{equation*}
    g : \left ( \mathbb{S}^{d_{1}}_{>} \times \cdots \times \mathbb{S}^{d_{k}}_{>},||\cdot||_{\text{F}} \right ) \rightarrow (\mathbb{R},|\cdot|) : (\mathbf{B}_{1},\dots,\mathbf{B}_{k}) \mapsto \sum_{j=1}^{d_{k}} \left (\eta_{j,11}^{2} + \cdots + \eta_{j,kk}^{2} \right )^{1/2},
\end{equation*}
with again $\eta_{j,ii}$ the eigenvalues of $\mathbf{B}_{i}$ in decreasing order for $i = 1,\dots,k$ and $j = 1,\dots,d_{i}$, and putting $\eta_{j,ii} = 0$ if $j = d_{i}+1,\dots,d_{k}$. Again applying the chain rule yields, in a very similar way and using the same notations,
\begin{equation*} 
\begin{split}
\lim_{t\,\underset{>}{\to}\,0} \frac{g(M(\mathbf{R}_{11} + t\mathbf{H}_{t,11}, \dots , \mathbf{R}_{kk} + t\mathbf{H}_{t,kk})) - g(M(\mathbf{R}_{11},\dots,\mathbf{R}_{kk}))}{t} 
& =  \sum_{i=1}^{k} \left ( \sum_{j=1}^{d_{i}} \frac{\lambda_{j,ii}}{\left (\lambda_{j,11}^{2} + \cdots + \lambda_{j,kk}^{2} \right )^{1/2}} \left ( \mathbf{U}_{ii}^{\text{T}}\mathbf{H}_{ii}\mathbf{U}_{ii}\right )_{jj} \right ) \\ & = \sum_{i=1}^{k} \text{tr}\left (\widetilde{\boldsymbol{\Delta}}_{ii}\mathbf{U}_{ii}^{\text{T}}\mathbf{H}_{ii}\mathbf{U}_{ii} \right ) = \text{tr} \left (\boldsymbol{\Upsilon}_{2}\mathbf{H} \right ).
\end{split}
\end{equation*}
\end{proof}

\noindent 
\textit{Proof of Theorem \ref{thm1}} \newline 
 
We start by showing the Fr\'{e}chet differentiability of $\mathcal{D}_{1}$ and $\mathcal{D}_{2}$. To this end, we prove that they are Hadamard differentiable on $\mathbb{S}^{q}_{>}$ with
\begin{equation*}
    \lim_{t\,\underset{>}{\to}\,0} \frac{\mathcal{D}_{r}(\mathbf{R} + t \mathbf{H}_{t}) - \mathcal{D}_{r}(\mathbf{R})}{t} = \text{tr}(\mathbf{M}_{r}\mathbf{H}),
\end{equation*}
for $r \in \{1,2\}$ and $\mathbf{H}_{t} \in \mathbb{S}^{q}$ for $t > 0$ and $\mathbf{H} \in \mathbb{S}^{q}$ such that $||\mathbf{H}_{t} - \mathbf{H}||_{\text{F}} \to 0$ as $t \to 0$. \newline \\ \noindent \underline{Differentiability of $\mathcal{D}_{1}$} \newline \\ 
Consider the function
\begin{equation*}
    f : (0,\infty)^{k+2} \to \mathbb{R} : \left (\overline{x}_{1},\dots,\overline{x}_{k},\overline{y},\overline{z} \right ) \mapsto \frac{\sum_{i=1}^{k}\overline{x}_{i}-\overline{y}}{\sum_{i=1}^{k}\overline{x}_{i}-\overline{z}}.
\end{equation*}
It is then quickly seen that $\mathcal{D}_{1}(\mathbf{R}) = f(x_{1},\dots,x_{k},y,z)$ and $\mathcal{D}_{1}(\mathbf{R} + t\mathbf{H}_{t}) = f(x_{1}^{t},\dots,{x}_{k}^{t},y^{t},z^{t})$ with
\begin{equation*}
    \begin{split}
        & \hspace{2cm} y = \text{tr}\left (\mathbf{R}^{1/2} \right ), \hspace{0.5cm} x_{i} = \text{tr}\left (\mathbf{R}_{ii}^{1/2} \right ), \hspace{0.5cm} z = \text{tr}\left (\mathbf{R}_{m}^{1/2} \right ), \\
        \hspace{-2cm} y^{t} & = \text{tr}\left \{(\mathbf{R} + t \mathbf{H}_{t})^{1/2}\right \}, \hspace{0.5cm} x_{i}^{t} = \text{tr}\left \{ (\mathbf{R} + t \mathbf{H}_{t})_{ii}^{1/2}\right \}, \hspace{0.5cm} z^{t} = \text{tr}\left \{(\mathbf{R} + t \mathbf{H}_{t})_{m}^{1/2} \right \},
    \end{split}
\end{equation*}
for $i = 1,\dots,k$. Here, $(\mathbf{R} + t\mathbf{H}_{t})_{ii}$ is the $d_{i} \times d_{i}$ diagonal block of $\mathbf{R} + t\mathbf{H}_{t}$ and $(\mathbf{R} + t\mathbf{H}_{t})_{m}$ the matrix in \eqref{eq: Rmhelp} with $\mathbf{R}$ replaced by $\mathbf{R} + t\mathbf{H}_{t}$.
Since the matrix $\mathbf{R}$ is positive definite, there exists a unique square root matrix $\mathbf{R}^{1/2}$ having eigenvalues on the sector $\{z \in \mathbb{C}: -\pi/2 < \text{arg}(z) < \pi/2\}$ of the complex plane, where $\text{arg}(z)$  denotes the argument of the complex number $z$. Moreover, since the mapping $z \mapsto z^{1/2}$ is infinitely differentiable on this sector, it follows from Theorem 3.8 of \cite{Higham2008} that the Fréchet derivative of the matrix function $\mathbf{R} \mapsto \mathbf{R}^{1/2}$ exists. Furthermore, Lemma 2.2 in \cite{Cardoso2011} tells us that this derivative in the direction of $\mathbf{H}$ is equal to the unique solution $\mathbf{Y}$ of the Sylvester equation
\begin{equation}\label{eq: hllp}
    \mathbf{R}^{1/2}\mathbf{Y} + \mathbf{Y}\mathbf{R}^{1/2} = \mathbf{H}.
\end{equation}
Hence, the Hadamard derivative in the direction of $\mathbf{H}$ also equals $\mathbf{Y}$ and from \eqref{eq: hllp} we obtain
\begin{equation*}
    \text{tr}(\mathbf{Y}) = \text{tr} \left ( \mathbf{R}^{-1/2}\mathbf{H} \right ) - \text{tr} \left (\mathbf{R}^{-1/2}\mathbf{Y}\mathbf{R}^{1/2} \right ) = \text{tr} \left (\mathbf{R}^{-1/2}\mathbf{H} \right ) - \text{tr}(\mathbf{Y}),
\end{equation*}
i.e., $\text{tr}(\mathbf{Y}) = \text{tr} \left(\mathbf{R}^{-1/2}\mathbf{H} \right )/2$. This (and a similar reasoning with $\mathbf{R}_{ii}$ instead of $\mathbf{R}$) shows that
\begin{equation*}
    \lim_{t\,\underset{>}{\to}\,0} \frac{y^{t}-y}{t} = \frac{1}{2} \text{tr} \left (\mathbf{R}^{-1/2}\mathbf{H} \right ) \hspace{0.2cm} \text{and} \hspace{0.2cm} \lim_{t\,\underset{>}{\to}\,0} \frac{x_{i}^{t} - x_{i}}{t} = \frac{1}{2}\text{tr} \left (\mathbf{R}_{ii}^{-1/2}\mathbf{H}_{ii} \right )
\end{equation*}
for $i = 1,\dots,k$, where we used the notation $\mathbf{H}_{ii}$ for the $d_{i} \times d_{i}$ diagonal block of $\mathbf{H}$. Also, Lemma \ref{lem3} guarantees that 
\begin{equation*}
    \lim_{t\,\underset{>}{\to}\,0} \frac{z^{t}-z}{t} = \frac{1}{2} \text{tr}(\boldsymbol{\Upsilon}_{1}\mathbf{H}).
\end{equation*}
Putting this together, we see that the Fr\'echet derivative of the map
\begin{equation*}
    g : \mathbb{S}^{q} \rightarrow (0,\infty)^{k+2} : \mathbf{A} \mapsto \left (\text{tr} \left (\mathbf{A}_{11}^{1/2} \right ),\dots,\text{tr} \left (\mathbf{A}_{kk}^{1/2} \right ),\text{tr} \left (\mathbf{A}^{1/2} \right ),\text{tr}\left (\mathbf{A}_{m}^{1/2} \right ) \right )
\end{equation*}
at $\mathbf{R}$ in the direction of $\mathbf{H}$ equals $$\left (\frac{1}{2}\text{tr}\left(\mathbf{R}_{11}^{-1/2}\mathbf{H}_{11}\right),\dots,\frac{1}{2}\text{tr}\left(\mathbf{R}_{kk}^{-1/2}\mathbf{H}_{kk}\right),\frac{1}{2}\text{tr}\left(\mathbf{R}^{-1/2}\mathbf{H}\right),\frac{1}{2}\text{tr}\left(\boldsymbol{\Upsilon}_{1}\mathbf{H}\right) \right ),$$ which we will call $\boldsymbol{\Delta}(\mathbf{H})$.
Next, the Jacobian matrix of the function $f$ is given by 
\begin{equation*} 
\begin{split}
    \mathbf{J}_{f} & = \begin{pmatrix}
    \frac{\partial f}{\partial \overline{x}_{1}} & \cdots & \frac{\partial f}{\partial \overline{x}_{k}} & \frac{\partial f}{\partial \overline{y}} & \frac{\partial f}{\partial \overline{z}}
    \end{pmatrix} = 
    \begin{pmatrix}
    \frac{\overline{y}-\overline{z}}{(\overline{x}_{1}+\cdots+\overline{x}_{k} - \overline{z})^{2}} & \cdots & \frac{\overline{y}-\overline{z}}{(\overline{x}_{1}+\cdots+\overline{x}_{k} - \overline{z})^{2}} & \frac{-1}{\overline{x}_{1} + \cdots + \overline{x}_{k} - \overline{z}} & \frac{\overline{x}_{1}+\cdots \overline{x}_{k} - \overline{y}}{(\overline{x}_{1}+\cdots + \overline{x}_{k}-\overline{z})^{2}}
    \end{pmatrix},
\end{split}
\end{equation*}
such that the Fr\'echet derivative (being nothing more than a total derivative) of the function $f$ at $g(\mathbf{R})$ in the direction of $\boldsymbol{\Delta}(\mathbf{H})$ is equal to 
\begin{equation*} 
\begin{split}
    \mathbf{J}_{f}|_{g(\mathbf{R})}  \boldsymbol{\Delta}(\mathbf{H})^{\text{T}} = \sum_{\ell=1}^{k} \left [\frac{\text{tr}\left (\mathbf{R}^{1/2} \right )-\text{tr} \left (\mathbf{R}_{m}^{1/2} \right )}{ \left \{ \sum_{i=1}^{k}\text{tr} \left (\mathbf{R}_{ii}^{1/2} \right )-\text{tr} \left (\mathbf{R}_{m}^{1/2} \right )\right \}^{2}} \frac{1}{2} \text{tr}\left (\mathbf{R}_{\ell\ell}^{-1/2}\mathbf{H}_{\ell\ell}\right ) \right ] & - \frac{1}{\sum_{i=1}^{k}\text{tr} \left (\mathbf{R}_{ii}^{1/2} \right ) - \text{tr} \left (\mathbf{R}_{m}^{1/2} \right )} \frac{1}{2}\text{tr}\left (\mathbf{R}^{-1/2}\mathbf{H}\right ) \\  &  + \frac{\sum_{i=1}^{k}\text{tr}\left (\mathbf{R}_{ii}^{1/2} \right )-\text{tr} \left (\mathbf{R}^{1/2} \right )}{\left \{ \sum_{i=1}^{k}\text{tr} \left (\mathbf{R}_{ii}^{1/2} \right ) - \text{tr}\left (\mathbf{R}_{m}^{1/2} \right ) \right \}^{2}} \frac{1}{2} \text{tr}(\boldsymbol{\Upsilon}_{1} \mathbf{H}) \\
    & \hspace{-7.5cm} = \frac{1}{2C_{1}} (1-\mathcal{D}_{1}(\mathbf{R})) \text{tr} \left (\mathbf{R}_{0}^{-1/2} \mathbf{H} \right ) - \frac{1}{2C_{1}} \text{tr}\left (\mathbf{R}^{-1/2} \mathbf{H} \right ) + \frac{1}{2C_{1}} \mathcal{D}_{1}(\mathbf{R}) \text{tr}(\boldsymbol{\Upsilon}_{1}\mathbf{H}) \\
    & \hspace{-7.5cm} = \frac{1}{2C_{1}} \left (-\text{tr} \left (\mathbf{R}^{-1/2} \mathbf{H} \right ) + (1-\mathcal{D}_{1}(\mathbf{R})) \text{tr} \left (\mathbf{R}_{0}^{-1/2} \mathbf{H} \right ) + \mathcal{D}_{1}(\mathbf{R}) \text{tr}(\boldsymbol{\Upsilon}_{1}\mathbf{H}) \right ),
\end{split}
\end{equation*}
where we used the definition of $\mathcal{D}_{1}$ and the fact that $\sum_{\ell=1}^{k} \text{tr} (\mathbf{R}_{\ell\ell}^{-1/2}\mathbf{H}_{\ell\ell} ) = \text{tr}  (\mathbf{R}_{0}^{-1/2}\mathbf{H}  )$. The result follows from the linearity of the trace operator and applying the chain rule to $f(g(\mathbf{R})) = \mathcal{D}_{1}(\mathbf{R})$.
\newline \\ \noindent
\underline{Differentiability of $\mathcal{D}_{2}$} \newline \\ \noindent
Consider the function
\begin{equation*}
    f : (0,\infty)^{3} \rightarrow \mathbb{R} : (\overline{x},\overline{y},\overline{z}) \mapsto \frac{\overline{z}-\overline{x}}{\overline{z}-\overline{y}}.
\end{equation*}
Then, $\mathcal{D}_{2}(\mathbf{R}) = f(x,y,z)$ and $\mathcal{D}_{2}(\mathbf{R}+t\mathbf{H}_{t}) = f(x^{t},y^{t},z^{t})$, where
\begin{equation*}
\begin{split}
x = \text{tr} \left \{ \left (\mathbf{R}_{0}^{1/2} \mathbf{R} \mathbf{R}_{0}^{1/2} \right )^{1/2} \right \}, \hspace{0.2cm} y = \text{tr} \left \{ \left (\mathbf{R}_{0}^{1/2} \mathbf{R}_{m} \mathbf{R}_{0}^{1/2} \right )^{1/2} \right \}, \hspace{0.2cm} z = \text{tr}(\mathbf{R}),
\end{split}
\end{equation*}
and similarly for $x^{t}$, $y^{t}$ and $z^{t}$ with $\mathbf{R} + t \mathbf{H}_{t}$ instead of $\mathbf{R}$. From Lemma \ref{lem4} and Lemma \ref{lem5}, we have
\begin{equation*}
     \lim_{t\,\underset{>}{\to}\,0} \frac{x^{t}-x}{t} = \frac{1}{2}\text{tr} \left \{ \left (\mathbf{J}_{0} + \mathbf{J}^{-1} \right )\mathbf{H} \right \} \hspace{0.2cm} \text{and} \hspace{0.2cm}  \lim_{t\,\underset{>}{\to}\,0} \frac{y^{t}-y}{t} = \text{tr}(\boldsymbol{\Upsilon}_{2} \mathbf{H}).
\end{equation*}
Evidently
\begin{equation*}
    \lim_{t\,\underset{>}{\to}\,0} \frac{z^{t}-z}{t} = \text{tr}(\mathbf{H}).
\end{equation*}
Hence, the Fr\'echet derivative of the map
\begin{equation*}
    g : \mathbb{S}^{q} \rightarrow (0,\infty)^{3} : \mathbf{A} \mapsto \left (\text{tr}\left \{ \left (\mathbf{A}_{0}^{1/2}\mathbf{A}\mathbf{A}_{0}^{1/2} \right )^{1/2} \right \}, \text{tr}\left \{ \left(\mathbf{A}_{0}^{1/2}\mathbf{A}_{m}\mathbf{A}_{0}^{1/2} \right )^{1/2} \right \}, \text{tr}(\mathbf{A}) \right )
\end{equation*}
at $\mathbf{R}$ in the direction of $\mathbf{H}$ equals 
\begin{equation*}
    \left (\frac{1}{2}\text{tr} \left \{ \left (\mathbf{J}_{0} + \mathbf{J}^{-1} \right )\mathbf{H} \right \}, \text{tr}\left (\boldsymbol{\Upsilon}_{2} \mathbf{H} \right ), \text{tr}(\mathbf{H}) \right ),
\end{equation*}
which we will call $\boldsymbol{\Delta}(\mathbf{H})$. 
The Jacobian matrix of $f$ is given by 
\begin{equation*}
    \mathbf{J}_{f} = \begin{pmatrix}
    \frac{\partial f}{\partial \overline{x}} & \frac{\partial f}{\partial \overline{y}} & \frac{\partial f}{\partial \overline{z}} 
    \end{pmatrix}
    = 
    \begin{pmatrix}
    \frac{-1}{\overline{z}-\overline{y}} & \frac{\overline{z}-\overline{x}}{(\overline{z}-\overline{y})^{2}} & \frac{\overline{x}-\overline{y}}{(\overline{z}-\overline{y})^{2}}
    \end{pmatrix},
\end{equation*}
such that the total derivative of $f$ at $g(\mathbf{R})$ in the direction of $\boldsymbol{\Delta}(\mathbf{H})$ becomes
\begin{equation*}
\begin{split}
    & \mathbf{J}_{f}|_{g(\mathbf{R})} \boldsymbol{\Delta}(\mathbf{H})^{\text{T}} =  \frac{\frac{1}{2} \text{tr} \left \{ \left (\mathbf{J}_{0} + \mathbf{J}^{-1} \right )\mathbf{H} \right \}}{\text{tr}\left \{ \left (\mathbf{R}_{0}^{1/2}\mathbf{R}_{m}\mathbf{R}_{0}^{1/2} \right )^{1/2} \right \} - \text{tr}(\mathbf{R})} + \frac{\text{tr}(\mathbf{R}) - \text{tr} \left \{ \left (\mathbf{R}_{0}^{1/2}\mathbf{R}\mathbf{R}_{0}^{1/2} \right )^{1/2} \right \}}{\left [\text{tr}(\mathbf{R}) - \text{tr}\left \{ \left (\mathbf{R}_{0}^{1/2}\mathbf{R}_{m}\mathbf{R}_{0}^{1/2} \right )^{1/2} \right \} \right ]^{2}} \text{tr}(\boldsymbol{\Upsilon}_{2}\mathbf{H}) \\
    & \hspace{6.4cm} + \frac{\text{tr} \left \{ \left (\mathbf{R}_{0}^{1/2}\mathbf{R}\mathbf{R}_{0}^{1/2} \right )^{1/2} \right \} - \text{tr}\left \{ \left(\mathbf{R}_{0}^{1/2}\mathbf{R}_{m}\mathbf{R}_{0}^{1/2} \right )^{1/2} \right \}}{\left [\text{tr}(\mathbf{R}) - \text{tr}\left \{ \left (\mathbf{R}_{0}^{1/2}\mathbf{R}_{m}\mathbf{R}_{0}^{1/2} \right )^{1/2} \right \} \right ]^{2}} \text{tr}(\mathbf{H})
  \\
    & \hspace{1.7cm} = \frac{1}{C_{2}} \left (-\frac{1}{2}\text{tr} \left \{ \left (\mathbf{J}_{0} + \mathbf{J}^{-1} \right )\mathbf{H} \right \} + \mathcal{D}_{2}(\mathbf{R}) \text{tr}(\boldsymbol{\Upsilon_{2}} \mathbf{H}) + (1-\mathcal{D}_{2}(\mathbf{R})) \text{tr}(\mathbf{H})  \right ).  
\end{split}
\end{equation*}
\newline \\ \noindent
\underline{Applying the delta method} \newline \\ \noindent 
Next, we consider the estimator $\widehat{\mathbf{R}}_{n}$. Theorem 3.1 in \cite{Klaassen1997} tells us that
\begin{equation*} 
    \sqrt{n}(\widehat{\mathbf{R}}_{n} - \mathbf{R}) - \frac{1}{\sqrt{n}} \sum_{\ell=1}^{n} \left [\mathbf{Z}^{(\ell)} \left (\mathbf{Z}^{(\ell)} \right )^{\text{T}} - \frac{1}{2} \left \{ \mbox{diag} \left (\mathbf{Z}^{(\ell)} \left (\mathbf{Z}^{(\ell)} \right )^{\text{T}} \right )\mathbf{R} + \mathbf{R} \hspace{0.05cm} \text{diag}\left (\mathbf{Z}^{(\ell)} \left (\mathbf{Z}^{(\ell)} \right )^{\text{T}} \right )  \right \} \right ] \xrightarrow{p} \mathbf{0}_{q \times q}
\end{equation*}
as $n \to \infty$, where $\mathbf{Z}^{(\ell)} = (\mathbf{Z}_{1}^{(\ell)},\dots,\mathbf{Z}_{k}^{(\ell)})^{\text{T}}$, with $\mathbf{Z}_{i}^{(\ell)} = (Z_{i1}^{(\ell)},\dots,Z_{id_{i}}^{(\ell)})$ for $i \in \{1,\dots,k\}$, for $\ell \in \{1,\dots,n\}$ is a sample from the $\mathcal{N}_{q}(\mathbf{0}_{q},\mathbf{R})$ distribution. The same expansion holds when $\widehat{\mathbf{R}}_{n}$ is the empirical correlation matrix of $\mathbf{Z}^{(1)},\dots,\mathbf{Z}^{(n)}$, see, e.g., Lemma 8 in \cite{Mordant2022}. Hence 
\begin{equation*}
    \sqrt{n}(\widehat{\mathbf{R}}_{n} - \mathbf{R}) - \sqrt{n} \left (\varphi \left (\frac{1}{n} \sum_{\ell=1}^{n} \mathbf{Z}^{(\ell)} \left (\mathbf{Z}^{(\ell)} \right )^{\text{T}} \right ) - \mathbf{R} \right  ) \xrightarrow{p} \mathbf{0}_{q \times q}
\end{equation*}
as $n \to \infty$, i.e., making use of the empirical correlation matrix based on a true Gaussian sample or based on a pseudo Gaussian sample, results in the same asymptotic expansion. Suppose further that $\mathbf{R} = \mathbf{U} \boldsymbol{\Lambda} \mathbf{U}^{\text{T}}$ is the eigendecomposition of $\mathbf{R}$. Then $\mathbf{Z}^{(\ell)} = \mathbf{U}\boldsymbol{\Lambda}^{1/2} \boldsymbol{\epsilon}^{(\ell)}$ for $\ell \in \{1,\dots,n\}$ and $\boldsymbol{\epsilon}^{(1)},\dots,\boldsymbol{\epsilon}^{(n)}$ a sample from $\mathcal{N}_{q}(\mathbf{0}_{q},\mathbf{I}_{q})$. From Lemma 7 of \cite{Mordant2022}, we have
\begin{equation*}
    \mathbf{W}_{n} = \frac{1}{\sqrt{n}} \sum_{\ell=1}^{n} \left \{ \boldsymbol{\epsilon}^{(\ell)} \left (\boldsymbol{\epsilon}^{(\ell)} \right )^{\text{T}} - \mathbf{I}_{q} \right \} \xrightarrow{d} \mathbf{W},
\end{equation*}
as $n \to \infty$, where $\mathbf{W}$ is a random symmetric matrix with $\mathbf{W}_{jk} \sim \mathcal{N}(0,2)$ if $j = k \in \{1,\dots,q\}$ and $\mathbf{W}_{jk} \sim \mathcal{N}(0,1)$ if $1 \leq j < k \leq q$ independently (and similarly for $k < j$). Moreover, for $\mathbf{A},\mathbf{B} \in \mathbb{S}^{q}$, it holds that
\begin{equation*}
    \text{E} \left (\text{tr}(\mathbf{A}\mathbf{W})\text{tr}(\mathbf{B}\mathbf{W}) \right ) = 2 \text{tr}(\mathbf{A}\mathbf{B}).
\end{equation*}
We find 
\begin{equation*}
    \mathbf{U}\boldsymbol{\Lambda}^{1/2} \mathbf{W}_{n} \boldsymbol{\Lambda}^{1/2} \mathbf{U}^{\text{T}} = \sqrt{n} \left ( \frac{1}{n} \sum_{\ell=1}^{n} \left (\mathbf{Z}^{(\ell)}\left (\mathbf{Z}^{(\ell)} \right )^{\text{T}} \right ) - \mathbf{R} \right ) \xrightarrow{d} \mathbf{U}\boldsymbol{\Lambda}^{1/2} \mathbf{W} \boldsymbol{\Lambda}^{1/2} \mathbf{U}^{\text{T}},
\end{equation*}
as $n \to \infty$.
Applying the delta method (and using that $\varphi(\mathbf{R}) = \mathbf{D}_{\mathbf{R}}^{-1/2}\mathbf{R}\mathbf{D}_{\mathbf{R}}^{-1/2} = \mathbf{R}$), we obtain
\begin{equation*}
\begin{split}
    \sqrt{n} \left (\mathcal{D}_{r}(\widehat{\mathbf{R}}_{n}) - \mathcal{D}_{r}(\mathbf{R}) \right ) & \xrightarrow{d} \text{tr} \left \{ (\mathbf{M}_{r} - \mathbf{D}_{\mathbf{M}_{r}\mathbf{R}})\mathbf{U}\boldsymbol{\Lambda}^{1/2} \mathbf{W} \boldsymbol{\Lambda}^{1/2} \mathbf{U}^{\text{T}} \right \} \\ & \hspace{0.11cm} = \text{tr} \left \{ \boldsymbol{\Lambda}^{1/2} \mathbf{U}^{\text{T}} (\mathbf{M}_{r} - \mathbf{D}_{\mathbf{M}_{r} \mathbf{R}}) \mathbf{U} \boldsymbol{\Lambda}^{1/2} \mathbf{W} \right \},
\end{split}
\end{equation*}
as $n \to \infty$. The latter asymptotic expression is centered Gaussian with asymptotic variance 
\begin{equation*}
    2 \text{tr} \left [ \left \{ \boldsymbol{\Lambda}^{1/2} \mathbf{U}^{\text{T}} \left (\mathbf{M}_{r} - \mathbf{D}_{\mathbf{M}_{r} \mathbf{R}} \right ) \mathbf{U} \boldsymbol{\Lambda}^{1/2} \right \}^{2} \right ] = 2 \text{tr} \left [ \left \{ \mathbf{R}  \left (\mathbf{M}_{r}-\mathbf{D}_{\mathbf{M}_{r}\mathbf{R}} \right ) \right \}^{2} \right ],
\end{equation*}
using the trace cyclical property and finishing the proof. \hfill \qedsymbol{}

\section{Proof of Proposition \ref{prop4}}\label{App D}

Consider the estimator given in \eqref{eq: est cor matrix}. Notice first of all that the deterministic correction 
\begin{equation*}
\left [ \frac{1}{n} \sum_{\ell = 1}^{n} \left \{\Phi^{-1} \left (\frac{\ell}{n+1} \right ) \right \}^{2} \right ]^{-1} = 1 + \mathcal{O} \left (n^{-1} \ln(n) \right )
\end{equation*}
is asymptotically insignificant. Suppose that $\sup_{n} \lambda_{\max}(\mathbf{R}_{n}) \leq \epsilon_{0}^{-1}$ for a certain $\epsilon_{0} > 0$. Fix now $i,m \in \{1,\dots,k\}$, $j \in \{1,\dots,d_{i}\}$ and $t \in \{1,\dots,d_{m} \}$. We desire to find a non-asymptotic deviation inequality
\begin{equation}\label{eq: devineq}
\mathbb{P} \left [\left |\frac{1}{n} \sum_{\ell = 1}^{n} \left ( \widehat{Z}_{ij}^{(\ell)} \widehat{Z}_{mt}^{(\ell)} - \rho_{ij,mt} \right ) \right | \geq  \epsilon \right ] \leq f(\epsilon,n)
\end{equation}
holding for $0 < \epsilon \leq \delta$, where $\delta > 0$ is a positive constant only depending on $\epsilon_{0}$. Denote by $\widehat{F}^{*}_{ij} = n/(n+1) \widehat{F}_{ij}$ the rescaled empirical cdf of $X_{ij}$, and similarly for $X_{mt}$, en let $\widehat{H}$ be the joint empirical cdf of $(X_{ij}^{(1)},X_{mt}^{(1)}),\dots, (X_{ij}^{(n)},X_{mt}^{(n)})$, and $H$ be the true cdf of $(X_{ij},X_{mt})$. Recall that $(X_{ij},X_{mt})$ has a Gaussian copula, meaning that
\begin{equation*}
\left (\left (\Phi^{-1} \circ F_{ij} \right )(X_{ij}), \left (\Phi^{-1} \circ F_{mt} \right )(X_{mt})\right ) \sim \Phi_{\text{G}},
\end{equation*}
where $\Phi_{\text{G}}$ stands for a bivariate normal distribution with means zero, unit variances and correlation $\rho_{ij,mt}$. We also need the Dvoretzky-Kiefer-Wolfowitz inequality (see, e.g., \cite{Massart1990}):
\begin{equation}\label{eq: DKW}
\mathbb{P} \left [\sup_{x \in \mathbb{R}} \left |\widehat{F}_{ij}(x) - F_{ij}(x) \right | \geq \epsilon \right ] \leq 2 \exp \left (-2n\epsilon^{2} \right )
\end{equation}
for any $\epsilon > 0$. Now, consider the decomposition 
\begin{equation*}
\begin{split}
\frac{1}{n} \sum_{\ell = 1}^{n} \left ( \widehat{Z}_{ij}^{(\ell)} \widehat{Z}_{mt}^{(\ell)} - \rho_{ij,mt} \right ) & = \int_{\mathbb{R}^{2}} \Phi^{-1} \left (\widehat{F}_{ij}^{*}(x) \right ) \Phi^{-1} \left (\widehat{F}_{mt}^{*}(y) \right ) d\widehat{H}(x,y) - \int_{\mathbb{R}^{2}} \Phi^{-1} \left (F_{ij}(x) \right ) \Phi^{-1} \left (F_{mt}(y) \right ) dH(x,y) \\
& = A_{1n} + A_{2n} + A_{3n} + B_{1n} + B_{2n} + R_{n},
\end{split}
\end{equation*}
where (denoting $\phi$ for the standard normal density function)
\begin{equation*}
    \begin{split}
        A_{1n} & = \int_{\mathbb{R}^{2}} \Phi^{-1}\left (F_{ij}(x) \right )\Phi^{-1}\left (F_{mt}(y) \right ) d \left (\widehat{H}-H \right )(x,y) \\
        A_{2n} & = \int_{\mathbb{R}^{2}} \frac{\widehat{F}_{ij}^{*}(x)-F_{ij}(x)}{\phi \left (\Phi^{-1}\left (F_{ij}(x) \right ) \right )} \Phi^{-1}\left (F_{mt}(y) \right ) dH(x,y) \\
        A_{3n} & = \int_{\mathbb{R}^{2}} \frac{\widehat{F}_{mt}^{*}(y)-F_{mt}(y)}{\phi \left (\Phi^{-1}\left (F_{mt}(y) \right ) \right )} \Phi^{-1}\left (F_{ij}(x) \right ) dH(x,y) \\
        B_{1n} & = \int_{\mathbb{R}^{2}} \left \{ \Phi^{-1} \left (\widehat{F}_{ij}^{*}(x) \right ) - \Phi^{-1} \left (F_{ij}(x) \right ) \right \} \Phi^{-1} \left (F_{mt}(y) \right ) d\widehat{H}(x,y) - A_{2n} \\
        B_{2n} & = \int_{\mathbb{R}^{2}} \left \{ \Phi^{-1} \left (\widehat{F}_{mt}^{*}(y) \right ) - \Phi^{-1} \left (F_{mt}(y) \right ) \right \} \Phi^{-1} \left (F_{ij}(x) \right ) d\widehat{H}(x,y) - A_{3n} \\
        R_{n} & = \int_{\mathbb{R}^{2}} \left \{ \Phi^{-1} \left (\widehat{F}_{ij}^{*}(x) \right ) - \Phi^{-1} \left (F_{ij}(x) \right ) \right \} \left \{ \Phi^{-1} \left (\widehat{F}_{mt}^{*}(y) \right ) - \Phi^{-1} \left (F_{mt}(y) \right ) \right \} d\widehat{H}(x,y).
    \end{split}
\end{equation*}
Lemma A.3. in \cite{Bickel2008} tells us that 
\begin{equation*}
\mathbb{P} \left [\left |A_{1n} \right | \geq \epsilon \right ] \leq C_{1} \exp \left (-C_{2} n \epsilon^{2} \right )
\end{equation*}
for $\epsilon \leq \delta_{1}$, where $C_{1},C_{2} > 0$ and $\delta_{1} > 0$ only depend on $\epsilon_{0}$.  Next, since 
\begin{equation*}
    dH(x,y) = \phi_{\text{G}} \left (\Phi^{-1} \left (F_{ij}(x) \right ), \Phi^{-1} \left (F_{mt}(y) \right ) \right ) d\Phi^{-1} \left (F_{ij}(x) \right )d\Phi^{-1} \left ( F_{mt}(y) \right ),
\end{equation*}
with $\phi_{\text{G}}$ the density of $\Phi_{\text{G}}$, we see that 
\begin{equation*}
    \begin{split}
        A_{2n} = \int_{\mathbb{R}^{2}} \frac{\widehat{F}_{ij}^{*}(x)-F_{ij}(x)}{\phi \left (\Phi^{-1} \left (F_{ij}(x) \right ) \right )} \Phi^{-1} \left (F_{mt}(y) \right ) \phi_{\text{G}} \left (\Phi^{-1} \left (F_{ij}(x) \right ), \Phi^{-1} \left (F_{mt}(y) \right ) \right ) d\Phi^{-1} \left (F_{ij}(x) \right )d\Phi^{-1} \left ( F_{mt}(y) \right ),
    \end{split}
\end{equation*}
where 
\begin{equation*}
    \begin{split}
        \int_{\mathbb{R}} \Phi^{-1} \left (F_{mt}(y) \right ) \phi_{\text{G}} \left (\Phi^{-1} \left (F_{ij}(x) \right ), \Phi^{-1} \left (F_{mt}(y) \right ) \right ) d\Phi^{-1} \left ( F_{mt}(y) \right ) & = \int_{\mathbb{R}} \phi \left (\Phi^{-1} \left (F_{ij}(x) \right ) \right ) \frac{\phi_{\text{G}} \left (\Phi^{-1} \left (F_{ij}(x) \right ),s \right )}{\phi \left (\Phi^{-1} \left (F_{ij}(x) \right ) \right )} \hspace{0.02cm} s \hspace{0.02cm} ds \\ & = \rho_{ij,mt} \hspace{0.02cm} \phi \left (\Phi^{-1} \left (F_{ij}(x) \right ) \right ) \Phi^{-1} \left (F_{ij}(x) \right ),
    \end{split}
\end{equation*}
since $\mathbb{E}(X_{2} | X_{1} = x_{1}) = \rho \hspace{0.02cm} x_{1}$ when $(X_{1},X_{2})$ follows a bivariate standard normal distribution with correlation $\rho$. Hence
\begin{equation*}
    \begin{split}
        A_{2n} & = \rho_{ij,mt} \int_{\mathbb{R}} \Phi^{-1} \left (F_{ij}(x) \right ) \left (\widehat{F}_{ij}^{*}(x)-F_{ij}(x) \right ) d\Phi^{-1} \left (F_{ij}(x) \right ) \\
        & = -\frac{\rho_{ij,mt}}{2} \int_{\mathbb{R}} \left \{ \Phi^{-1} \left (F_{ij}(x) \right ) \right \}^{2} d \left (\widehat{F}_{ij}^{*} - F_{ij} \right )(x) \\
        & = - \frac{\rho_{ij,mt}}{2} \frac{1}{n} \sum_{\ell=1}^{n} \left (\left \{ \Phi^{-1} \left (F_{ij} \left (X_{ij}^{(\ell)} \right ) \right ) \right \}^{2} - 1 \right ) \\
        & = - \frac{\rho_{ij,mt}}{2} \frac{1}{n} \sum_{\ell=1}^{n} \left (V_{\ell}^{2} - 1 \right ),
    \end{split}
\end{equation*}
where we did partial integration, and $V_{\ell} \sim \chi^{2}_{1}$ for $\ell = 1,\dots,n$. Hence, using Theorem 3.2 on page 45 of \cite{Saulis1991} (noting that condition (P)(3.12) on page 45 holds for $\chi^{2}_{1}$, see also the proof of Lemma A.3. in \cite{Bickel2008}), 
we have 
\begin{equation*}
\mathbb{P} \left [\left |A_{2n} \right | \geq \epsilon \right ] \leq C_{3} \exp \left (- C_{4} n \epsilon^{2} \right )
\end{equation*}
for $\epsilon \leq \delta_{2}$, where $C_{3},C_{4}$ and $\delta_{2} > 0$ only depend on $\epsilon_{0}$. A very similar argument holds for $A_{3n}$. 

Next, we deal with the term $B_{1n}$. By symmetry, the term $B_{2n}$ can be dealt with similarly. We use the mean value theorem, giving that
\begin{equation*}
    \Phi^{-1} \left (\widehat{F}_{ij}^{*}(x) \right ) - \Phi^{-1} \left (F_{ij}(x) \right ) = \frac{\widehat{F}_{ij}^{*}(x)-F_{ij}(x)}{\phi \left ( \Phi^{-1} \left (\widetilde{F}_{ij}^{*}(x) \right ) \right )}
\end{equation*}
for a certain $\widetilde{F}_{ij}^{*}(x)$ satisfying $|\widetilde{F}_{ij}^{*}(x)-F_{ij}(x)| \leq |\widehat{F}_{ij}^{*}(x) - F_{ij}(x)|$. Of course, $|\Phi^{-1}(t)| \to \infty$ when $t \to 0$ or $t \to 1$, and $\phi(t) \to 0$ when $|t| \to \infty$, so we will need to further split up $B_{1n}$ by defining the set $M_{\eta} = M_{\eta 1} \times M_{\eta 2}$, where 
\begin{equation*}
    \begin{split}
        M_{\eta 1} & = \left [F_{ij}^{-1}(\eta), F_{ij}^{-1}(1-\eta) \right ] \\
        M_{\eta 2} & = \left [F_{mt}^{-1}(\eta), F_{mt}^{-1}(1-\eta) \right ]
    \end{split}
\end{equation*}
for a certain small $\eta > 0$. Doing so, we can write 
\begin{equation*}
    B_{1n} = B_{\eta11n} + B_{\eta21n} + B_{\eta31n} + B_{\eta41n},
\end{equation*}
where 
\begin{equation*}
    \begin{split}
        B_{\eta11n} & = \int_{M_{\eta}^{\text{C}}} \left \{ \Phi^{-1} \left (\widehat{F}_{ij}^{*}(x) \right ) - \Phi^{-1} \left (F_{ij}(x) \right ) \right \} \Phi^{-1} \left (F_{mt}(y) \right ) d\widehat{H}(x,y) \\
        B_{\eta21n} & = \int_{M_{\eta}} \left \{ \frac{\left (\widehat{F}_{ij}^{*}(x) - F_{ij}(x) \right )}{\phi \left ( \Phi^{-1} \left (\widetilde{F}_{ij}^{*}(x) \right ) \right )} - \frac{ \left (\widehat{F}_{ij}^{*}(x) - F_{ij}(x) \right )}{\phi \left ( \Phi^{-1} \left (F_{ij}(x) \right ) \right )} \right \} \Phi^{-1} \left (F_{mt}(y) \right ) d\widehat{H}(x,y) 
\\
        B_{\eta31n} & =  \int_{M_{\eta}} \frac{\left (\widehat{F}_{ij}^{*}(x) - F_{ij}(x) \right )}{\phi \left ( \Phi^{-1} \left (F_{ij}(x) \right ) \right )} \Phi^{-1} \left (F_{mt}(y) \right ) d \left (\widehat{H}-H \right )(x,y) \\
        B_{\eta41n} & =  \int_{M_{\eta}} \frac{ \left (\widehat{F}_{ij}^{*}(x) - F_{ij}(x) \right )}{\phi \left ( \Phi^{-1} \left (F_{ij}(x) \right ) \right )} \Phi^{-1} \left (F_{mt}(y) \right ) dH(x,y) - A_{2n}.
    \end{split}
\end{equation*}
Notice that we did not apply the mean value theorem in the term $B_{\eta11n}$, because we will decompose $R_{n}$ as
\begin{equation*}
    R_{n} = R_{\eta 1n} + R_{\eta 2n} + R_{\eta 3n},
\end{equation*}
where 
\begin{equation*}
    \begin{split}
        R_{\eta 1n} & =  \int_{M_{\eta}^{\text{C}}} \left \{ \Phi^{-1} \left (\widehat{F}_{ij}^{*}(x) \right ) - \Phi^{-1} \left (F_{ij}(x) \right ) \right \} \Phi^{-1} \left (\widehat{F}_{mt}^{*}(y) \right ) d\widehat{H}(x,y) \\
        R_{\eta 2n} & = - \int_{M_{\eta}^{\text{C}}} \left \{ \Phi^{-1} \left (\widehat{F}_{ij}^{*}(x) \right ) - \Phi^{-1} \left (F_{ij}(x) \right ) \right \} \Phi^{-1} \left (F_{mt}(y) \right ) d\widehat{H}(x,y) \\
        R_{\eta 3n} & = \int_{M_{\eta}} \frac{\widehat{F}_{ij}^{*}(x)-F_{ij}(x)}{\phi \left (\Phi^{-1} \left (\widetilde{F}_{ij}^{*}(x) \right ) \right )} \left \{ \Phi^{-1} \left (F_{mt}^{*}(y) \right ) - \Phi^{-1} \left (F_{mt}(y) \right ) \right \} d\widehat{H}(x,y),
    \end{split}
\end{equation*}
such that $B_{\eta 11n}$ and $R_{\eta 2 n}$ cancel each other out. As for $B_{\eta 21n}$, it holds almost surely that 
\begin{equation*}
    \left | B_{\eta 21n} \right | \leq \sup_{x \in \mathbb{R}} \left |\widehat{F}_{ij}^{*}(x)-F_{ij}(x) \right | \sup_{x \in M_{\eta 1}} \left |\frac{1}{\phi \left (\Phi^{-1} \left (\widetilde{F}_{ij}^{*}(x) \right ) \right )} - \frac{1}{\phi \left (\Phi^{-1} \left (F_{ij}(x) \right ) \right )} \right | \sup_{y \in M_{\eta 2}} \Phi^{-1} \left (F_{mt}(y) \right ).
\end{equation*}
 The function $\Phi^{-1} \circ F_{mt}$ is uniformly bounded on $M_{\eta 2}$ and since $\phi \circ \Phi^{-1}$ is continuous, it is uniformly continuous on $M_{\eta 1}$. This, and the fact that $|\widetilde{F}_{ij}^{*} - F_{ij}| \leq |\widehat{F}_{ij}^{*} - F_{ij}|$, results in 
\begin{equation*}
    \left |B_{\eta 21 n} \right | \leq K \sup_{x \in \mathbb{R}} \left |\widehat{F}_{ij}^{*}(x)-F_{ij}(x) \right |
\end{equation*}
almost surely for a certain $K > 0$. Hence, using \eqref{eq: DKW}, we obtain
\begin{equation*}
\begin{split}
    \mathbb{P} \left [\left |B_{\eta21n} \right | \geq \epsilon \right ] \leq \mathbb{P} \left [\sup_{x \in \mathbb{R}} \left |\widehat{F}_{ij}^{*}(x)-F_{ij}(x) \right | \geq \frac{\epsilon}{K} \right ] & \leq \mathbb{P} \left [\sup_{x \in \mathbb{R}} \left |\widehat{F}_{ij}(x)-F_{ij}(x) \right | \geq \frac{\epsilon}{2K} \right ] + \mathbb{P} \left [\sup_{x \in \mathbb{R}} \left |\widehat{F}_{ij}^{*}(x)-\widehat{F}_{ij}(x) \right | \geq \frac{\epsilon}{2K} \right ] \\ & \leq 2 \exp \left (\frac{-n \epsilon^{2}}{2K^{2}}  \right ) + \mathbb{P} \left [\frac{1}{n+1} \geq \frac{\epsilon}{2K} \right ].
\end{split}
\end{equation*}
Similar arguments hold for the terms $B_{\eta 31n},B_{\eta 41n}$ and $R_{\eta 3n}$. 

The only term that is not discussed yet, is $R_{\eta 1 n}$. We want this term to converge to zero in probability when $\eta \to 0$ for every $n$. We can deal with this in a similar way as in Corollary 5.6. of \cite{Ruymgaart1972} (in fact, we are even in the context of Section 6 on page $1133$ since Assumption 2.3(b) holds), since the corresponding assumptions hold here (they are verified in the proof of Theorem 3.1 of \cite{Klaassen1997}). 
In our setting, the key ingredient is that there exists $a = b = (1/2 - \xi )/2$ for a certain $0 < \xi < 1/2$ such that 
\begin{equation*}
    \left |\Phi^{-1}(t)\right | \leq M_{1} r^{a}(t), \hspace{1cm} \text{with} \hspace{1cm} r(t) = \frac{1}{t(1-t)}
\end{equation*}
for all $t \in (0,1)$ and a certain $M_{1} > 0$, i.e., $|\Phi^{-1}(t)| \to \infty$ as $t \to 0$ or $t \to 1$, but not too fast. Applying the mean value theorem again, we have
\begin{equation*}
\begin{split}
    \left |\Phi^{-1} \left (\widehat{F}_{ij}^{*}(x) \right ) - \Phi^{-1} \left (F_{ij}(x) \right ) \right | & = \left |\left (\widehat {F}_{ij}^{*}(x) - F_{ij}(x) \right ) \left [\Phi^{-1} \left (\widetilde{F}_{ij}^{*}(x) \right ) \right ]^{\prime} \right | \\ & \leq M_{1} \left |\widehat{F}_{ij}^{*}(x)-F_{ij}(x) \right | r^{a+1} \left (\widetilde{F}_{ij}^{*}(x) \right ) .
\end{split}
\end{equation*}
Moreover, if we take $\widetilde{\epsilon}_{1} > 0$ arbitrary, there exist $M_{2},M_{3},M_{4} > 0$ such that, uniformly in $n$, the events
\begin{equation*}
\begin{split}
    E_{1n} & = \left \{\left |\widehat{F}_{ij}^{*} - F_{ij} \right | \leq M_{2} r^{-1/2 + \xi / 4}(F_{ij}) \right \} \\
    E_{2n} & = \left \{r^{a} \left (\widehat{F}_{mt}^{*} \right ) \leq M_{3} r^{a}(F_{mt}) \right \} \\
    E_{3n} & = \left \{r^{a+1} \left (\widetilde{F}_{ij}^{*} \right )  \leq M_{4} r^{a+1}(F_{ij}) \right \}
\end{split}
\end{equation*}
satisfy $\mathbb{P}(E^{\text{C}}) < \widetilde{\epsilon}_{1}/2$, for $E = \cap_{i=1}^{3} E_{in}$. Combining the inequalities, we also see that (similarly to expression (5.2) in \cite{Ruymgaart1972})
\begin{equation*}
    \mathbb{E} \left (\mathds{1} \left (E \right ) \left |R_{\eta 1 n} \right | \right ) \leq M_{1}^{2} M_{2} M_{3} M_{4} \int_{M_{\eta}^{\text{C}}} r^{a + 1/2 + \xi / 4} \left (F_{ij}(x) \right )  r^{a} \left (F_{mt}(y) \right ) dH(x,y),
\end{equation*}
where $\mathds{1}$ is the indicator function.
By H\"{o}lder's inequality, the above integral is bounded by (noting that $M_{\eta}^{\text{C}} \subset (M_{\eta 1}^{\text{C}} \times \mathbb{R}) \cup (\mathbb{R} \times M_{\eta 2}^{\text{C}})$)
\begin{equation*}
\begin{split}
    & \left (\int_{(0,\eta) \cup (1-\eta,1)} r^{p_{1}(a + 1/2 + \xi / 4)}(t) dt \right )^{1/p_{1}} \left (\int_{0}^{1} r^{q_{1}a}(t) dt \right )^{1/q_{1}} \\ & \hspace{4cm} + \left ( \int_{0}^{1} r^{p_{1}(a+1/2+\xi/4)}(t) dt \right )^{1/p_{1}} \left (\int_{(0,\eta) \cup (1-\eta,1)} r^{q_{1}a}(t) dt \right )^{1/q_{1}} < \infty, 
\end{split}
\end{equation*}
where $p_{1}$ and $q_{1}$ are chosen such that $p_{1}^{-1} + q_{1}^{-1} = 1$, $p_{1}(a+1/2+\xi/4) < 1$ and $q_{1}a < 1$ (see equation (3.5) in \cite{Ruymgaart1972} and corresponding explanation), making the integral bounded. Hence, by dominated convergence, we have that $\mathds{1}(E)|R_{\eta 1 n}| \xrightarrow{p} 0$ as $\eta \to 0$ for each $n$. Hence, for our arbitrarily chosen $\widetilde{\epsilon}_{1} > 0$, and another arbitrary $\widetilde{\epsilon}_{2} > 0$, we see that
\begin{equation*}
    \begin{split}
        \mathbb{P} \left [\left |R_{\eta 1 n} \right | > \widetilde{\epsilon}_{2}  \right ] & = \mathbb{P} \left [\left |R_{\eta 1 n} \right | > \widetilde{\epsilon}_{2}, \mathds{1}(E) = 1 \right ] + \mathbb{P} \left [\left |R_{\eta 1 n} \right | > \widetilde{\epsilon}_{2}, \mathds{1}(E) = 0 \right ] \\
        & = \mathbb{P} \left [\mathds{1} \left (E \right ) \left |R_{\eta 1 n} \right | > \widetilde{\epsilon}_{2} \right ] + \mathbb{P} \left [\left |R_{\eta 1 n} \right | > \widetilde{\epsilon}_{2}, \mathds{1}(E) = 0 \right ] \\ 
        & \leq  \mathbb{P} \left [\mathds{1} \left (E \right ) \left |R_{\eta 1 n} \right | > \widetilde{\epsilon}_{2} \right ] + \mathbb{P} \left [\mathds{1}(E) = 0 \right ] \\
        & = \mathbb{P} \left [\mathds{1} \left (E \right ) \left |R_{\eta 1 n} \right | > \widetilde{\epsilon}_{2} \right ] + \mathbb{P} \left (E^{\text{C}} \right ) \\
        & < \frac{\widetilde{\epsilon}_{1}}{2} + \frac{\widetilde{\epsilon}_{1}}{2} = \widetilde{\epsilon}_{1}
    \end{split}
\end{equation*}
for $\eta$ small enough, i.e., $|R_{\eta 1 n}| \xrightarrow{p} 0$ as $\eta \to 0$ for all $n$. This means that $\mathbb{P} [|R_{\eta 1 n}| \geq \epsilon]$ can be made zero for every $n$ by letting $\eta \to 0$, and, combining everything, we have shown that a concentration inequality \eqref{eq: devineq} holds. In simplified form, we can say that \eqref{eq: devineq} holds with
\begin{equation*}
f(\epsilon,n) = K_{1} \exp \left (-K_{2} n \epsilon^{2} \right ) + K_{3} \mathbb{P} \left [\frac{1}{n+1} \geq \frac{\epsilon}{K_{4}} \right ],
\end{equation*}
for $0 < \epsilon \leq \delta$, where $K_{1},K_{2},K_{3},K_{4},\delta > 0$ possibly only depend on $\epsilon_{0}$. 

We now show that $||\widehat{\mathbf{R}}_{n} - \mathbf{R}_{n}||_{\infty} = \mathcal{O}_{p}(\{\ln(q_{n})/n\}^{1/2})$. Take an arbitrary $\varepsilon > 0$. Let $M,N > 0$ be such that $2 - K_{2} M^{2} < 0$, $K_{1} q_{n}^{2-K_{2}M^{2}} < \varepsilon$, $M n^{-1/2} \ln(q_{n})^{1/2} \leq \delta$ and $1/(n+1) < M n^{-1/2} \ln(q_{n})^{1/2} / K_{4}$ for all $n > N$. Then, by the union bound, it holds that 
\begin{equation*}
\begin{split}
\mathbb{P} \left [\frac{\left | \left |\widehat{\mathbf{R}}_{n} - \mathbf{R}_{n} \right | \right |_{\infty}}{n^{-1/2} \ln(q_{n})^{1/2}} > M \right ] & = \mathbb{P} \left [\left | \left |\widehat{\mathbf{R}}_{n} - \mathbf{R}_{n} \right | \right |_{\infty} > M n^{-1/2} \ln(q_{n})^{1/2} \right ] \\
& \leq q_{n}^{2} K_{1} \exp \left \{-K_{2} n \left (M n^{-1/2} \ln(q_{n})^{1/2} \right )^{2} \right \} + q_{n}^{2} K_{3} \mathbb{P} \left [\frac{1}{n+1} \geq \frac{M n^{-1/2} \ln(q_{n})^{1/2}}{K_{4}} \right ] \\
& = K_{1} q_{n}^{2-K_{2}M^{2}} < \varepsilon,
\end{split}
\end{equation*}
for all $n > N$, proving the desired result. \hfill \qedsymbol

\bibliographystyle{myjmva}
\bibliography{Bibliography}

\begin{thebibliography}{52}
\expandafter\ifx\csname natexlab\endcsname\relax\def\natexlab#1{#1}\fi
\providecommand{\bibinfo}[2]{#2}
\ifx\xfnm\relax \def\xfnm[#1]{\unskip,\space#1}\fi
\bibitem[{Ansari and Fuchs(2023)}]{Ansari2023}
\bibinfo{author}{J.~Ansari}, \bibinfo{author}{S.~Fuchs}, \bibinfo{title}{A
  simple extension of {A}zadkia \& {C}hatterjee's rank correlation to a vector
  of enogenous variables, preprint arxiv.2212.01621}  (\bibinfo{year}{2023}).
\bibitem[{Azadkia and Chatterjee(2021)}]{Azadkia2021}
\bibinfo{author}{M.~Azadkia}, \bibinfo{author}{S.~Chatterjee},
  \bibinfo{title}{A simple measure of conditional dependence},
  \bibinfo{journal}{Ann. Stat.} \bibinfo{volume}{49} (\bibinfo{year}{2021})
  \bibinfo{pages}{3070--3102}.
\bibitem[{Bickel and Levina(2008)}]{Bickel2008}
\bibinfo{author}{P.~J. Bickel}, \bibinfo{author}{E.~Levina},
  \bibinfo{title}{Regularized estimation of large covariance matrices},
  \bibinfo{journal}{Ann. Stat.} \bibinfo{volume}{36} (\bibinfo{year}{2008})
  \bibinfo{pages}{199--227}.
\bibitem[{Bien and Tibshirani(2011)}]{Bien2011}
\bibinfo{author}{J.~Bien}, \bibinfo{author}{R.~J. Tibshirani},
  \bibinfo{title}{Sparse estimation of a covariance matrix},
  \bibinfo{journal}{Biometrika} \bibinfo{volume}{98} (\bibinfo{year}{2011})
  \bibinfo{pages}{807--820}.
\bibitem[{Bien and Tibshirani(2022)}]{Bien2022}
\bibinfo{author}{J.~Bien}, \bibinfo{author}{R.~J. Tibshirani},
  \bibinfo{title}{Spcov: Sparse Estimation of a Covariance Matrix},
  \bibinfo{year}{2022}. \bibinfo{note}{R package version 1.3}.
\bibitem[{Bigot et~al.(2011)Bigot, Biscay, Loubes and
  Mu{\~n}iz-Alvarez}]{Bigot2011}
\bibinfo{author}{J.~Bigot}, \bibinfo{author}{R.~J. Biscay},
  \bibinfo{author}{J.-M. Loubes}, \bibinfo{author}{L.~Mu{\~n}iz-Alvarez},
  \bibinfo{title}{Group lasso estimation of high-dimensional covariance
  matrices}, \bibinfo{journal}{J. Mach. Learn. Res.} \bibinfo{volume}{12}
  (\bibinfo{year}{2011}) \bibinfo{pages}{3187--3225}.
\bibitem[{Cardoso(2011)}]{Cardoso2011}
\bibinfo{author}{J.~R. Cardoso}, \bibinfo{title}{Evaluating the {F}r\'{e}chet
  derivative of the matrix pth root}, \bibinfo{journal}{Electron. Trans. Numer.
  Anal.} \bibinfo{volume}{38} (\bibinfo{year}{2011}) \bibinfo{pages}{202--217}.
\bibitem[{Chatterjee(2021)}]{Chatterjee2021}
\bibinfo{author}{S.~Chatterjee}, \bibinfo{title}{A new coefficient of
  correlation}, \bibinfo{journal}{J. Am. Stat. Assoc.} \bibinfo{volume}{116}
  (\bibinfo{year}{2021}) \bibinfo{pages}{2009--2022}.
\bibitem[{Chiquet et~al.(2012)Chiquet, Grandvalet and
  Charbonnier}]{Chiquet2012}
\bibinfo{author}{J.~Chiquet}, \bibinfo{author}{Y.~Grandvalet},
  \bibinfo{author}{C.~Charbonnier}, \bibinfo{title}{Sparsity with sign-coherent
  groups of variables via the cooperative-lasso}, \bibinfo{journal}{Ann. Appl.
  Stat.} \bibinfo{volume}{6} (\bibinfo{year}{2012}) \bibinfo{pages}{795--860}.
\bibitem[{Clement and Desch(2008)}]{Clement2008}
\bibinfo{author}{P.~Clement}, \bibinfo{author}{W.~Desch}, \bibinfo{title}{An
  elementary proof of the triangle inequality for the {W}asserstein metric},
  \bibinfo{journal}{Proc. Am. Math. Soc.} \bibinfo{volume}{136}
  (\bibinfo{year}{2008}) \bibinfo{pages}{333--339}.
\bibitem[{De~Keyser and Gijbels(2023{\natexlab{a}})}]{Gijbels2023}
\bibinfo{author}{S.~De~Keyser}, \bibinfo{author}{I.~Gijbels},
  \bibinfo{title}{Copula-based divergence measures for dependence between
  random vectors}, in: \bibinfo{editor}{L.~A. Garc\'{i}a-Escudero},
  \bibinfo{editor}{A.~Gordaliza}, \bibinfo{editor}{A.~Mayo},
  \bibinfo{editor}{M.~A.~L. Gomez}, \bibinfo{editor}{M.~A. Gil},
  \bibinfo{editor}{P.~Grzegorzewski}, \bibinfo{editor}{O.~Hryniewicz} (Eds.),
  \bibinfo{booktitle}{Advances in Intelligent Systems and Computing, Vol. 1433,
  Building Bridges between Soft and Statistical Methodologies for Data
  Science}, \bibinfo{publisher}{Springer}, \bibinfo{year}{2023}{\natexlab{a}},
  pp. \bibinfo{pages}{104--111}.
\bibitem[{De~Keyser and Gijbels(2023{\natexlab{b}})}]{Gijbels2023b}
\bibinfo{author}{S.~De~Keyser}, \bibinfo{author}{I.~Gijbels},
  \bibinfo{title}{Parametric dependence between random vectors via copula-based
  divergence measures, preprint arxiv:2302.13611}
  (\bibinfo{year}{2023}{\natexlab{b}}).
\bibitem[{Fan et~al.(2009)Fan, Feng and Wu}]{Fan2009}
\bibinfo{author}{J.~Fan}, \bibinfo{author}{Y.~Feng}, \bibinfo{author}{Y.~Wu},
  \bibinfo{title}{Network exploration via the adaptive lasso and scad
  penalties}, \bibinfo{journal}{Ann. Appl. Stat.} \bibinfo{volume}{3}
  (\bibinfo{year}{2009}) \bibinfo{pages}{521--541}.
\bibitem[{Fan and Li(2001)}]{Fan2001}
\bibinfo{author}{J.~Fan}, \bibinfo{author}{R.~Li}, \bibinfo{title}{Variable
  selection via nonconcave penalized likelihood and its oracle properties},
  \bibinfo{journal}{J. Am. Stat. Assoc.} \bibinfo{volume}{96}
  (\bibinfo{year}{2001}) \bibinfo{pages}{1348--1360}.
\bibitem[{Fermanian(2024)}]{Fermanian2023}
\bibinfo{author}{J.-D. Fermanian}, \bibinfo{title}{Sparse {M}-estimators in
  semi-parametric copula models, to appear in {B}ernoulli, url:
  https://www.bernoullisociety.org/publications/bernoulli-journal/bernoulli-journal-papers}
   (\bibinfo{year}{2024}).
\bibitem[{Fop(2021)}]{Fop2021}
\bibinfo{author}{M.~Fop}, \bibinfo{title}{Covglasso: sparse covariance matrix
  estimation}, \bibinfo{year}{2021}. \bibinfo{note}{R package version 1.0.3}.
\bibitem[{Foygel and Drton(2010)}]{Foygel2010}
\bibinfo{author}{R.~Foygel}, \bibinfo{author}{M.~Drton},
  \bibinfo{title}{Extended {B}ayesian information criteria for {G}aussian
  graphical models}, in: \bibinfo{editor}{J.~Lafferty},
  \bibinfo{editor}{C.~Williams}, \bibinfo{editor}{J.~Shawe-Taylor},
  \bibinfo{editor}{R.~Zemel}, \bibinfo{editor}{A.~Culotta} (Eds.),
  \bibinfo{booktitle}{Advances in Neural Information Processing Systems 23
  (NIPS 2010)}, \bibinfo{year}{2010}, pp. \bibinfo{pages}{604--612}.
\bibitem[{Fuchs(2023)}]{Fuchs2022}
\bibinfo{author}{S.~Fuchs}, \bibinfo{title}{Quantifying directed dependence via
  dimension reduction}, \bibinfo{journal}{J. Multivar. Anal., in press,
  available online, 105266}  (\bibinfo{year}{2023}).
\bibitem[{Fuchs et~al.(2021)Fuchs, Di~Lascio and Durante}]{Fuchs2021}
\bibinfo{author}{S.~Fuchs}, \bibinfo{author}{F.~M.~L. Di~Lascio},
  \bibinfo{author}{F.~Durante}, \bibinfo{title}{Dissimilarity functions for
  rank-invariant hierarchical clustering of continuous variables},
  \bibinfo{journal}{Comput. Stat. Data Anal.} \bibinfo{volume}{159}
  (\bibinfo{year}{2021}) \bibinfo{pages}{107201}.
\bibitem[{Geenens and Lafaye~de Micheaux(2022)}]{Geenens2022}
\bibinfo{author}{G.~Geenens}, \bibinfo{author}{P.~Lafaye~de Micheaux},
  \bibinfo{title}{The {H}ellinger correlation}, \bibinfo{journal}{J. Am. Stat.
  Assoc.} \bibinfo{volume}{117} (\bibinfo{year}{2022})
  \bibinfo{pages}{639--653}.
\bibitem[{Genest et~al.(2009)Genest, R\'{e}millard and Beaudoin}]{Genest2009}
\bibinfo{author}{C.~Genest}, \bibinfo{author}{B.~R\'{e}millard},
  \bibinfo{author}{D.~Beaudoin}, \bibinfo{title}{Goodness-of-fit tests for
  copulas: A review and a power study}, \bibinfo{journal}{Insur. Math. Econ.}
  \bibinfo{volume}{44} (\bibinfo{year}{2009}) \bibinfo{pages}{199--213}.
\bibitem[{Gijbels et~al.(2021)Gijbels, Kika and Omelka}]{Gijbels2021}
\bibinfo{author}{I.~Gijbels}, \bibinfo{author}{V.~Kika},
  \bibinfo{author}{M.~Omelka}, \bibinfo{title}{On the specification of
  multivariate association measures and their behaviour with increasing
  dimension}, \bibinfo{journal}{J. Multivar. Anal.} \bibinfo{volume}{182}
  (\bibinfo{year}{2021}) \bibinfo{pages}{104704}.
\bibitem[{Grothe et~al.(2014)Grothe, Schnieders and Segers}]{Grothe2014}
\bibinfo{author}{O.~Grothe}, \bibinfo{author}{J.~Schnieders},
  \bibinfo{author}{J.~Segers}, \bibinfo{title}{Measuring association and
  dependence between random vectors}, \bibinfo{journal}{J. Multivar. Anal.}
  \bibinfo{volume}{123} (\bibinfo{year}{2014}) \bibinfo{pages}{96--110}.
\bibitem[{H\'{a}jek and \v{S}id\'ak(1967)}]{Hajek1967}
\bibinfo{author}{J.~H\'{a}jek}, \bibinfo{author}{Z.~\v{S}id\'ak},
  \bibinfo{title}{Theory of Rank Tests}, \bibinfo{publisher}{Academia},
  \bibinfo{address}{Prague}, \bibinfo{year}{1967}.
\bibitem[{Higham(2008)}]{Higham2008}
\bibinfo{author}{N.~J. Higham}, \bibinfo{title}{Functions of Matrices},
  \bibinfo{publisher}{SIAM}, \bibinfo{year}{2008}.
\bibitem[{Hofert et~al.(2019)Hofert, Oldford, Prasad and Zhu}]{Hofert2019}
\bibinfo{author}{M.~Hofert}, \bibinfo{author}{W.~Oldford},
  \bibinfo{author}{A.~Prasad}, \bibinfo{author}{M.~Zhu}, \bibinfo{title}{A
  framework for measuring association of random vectors via collapsed random
  variables}, \bibinfo{journal}{J. Multivar. Anal.} \bibinfo{volume}{172}
  (\bibinfo{year}{2019}) \bibinfo{pages}{5--27}.
\bibitem[{Hotelling(1936)}]{Hotelling1936}
\bibinfo{author}{H.~Hotelling}, \bibinfo{title}{Relations between two sets of
  variates}, \bibinfo{journal}{Biometrika} \bibinfo{volume}{28}
  (\bibinfo{year}{1936}) \bibinfo{pages}{321--377}.
\bibitem[{Klaassen and Wellner(1997)}]{Klaassen1997}
\bibinfo{author}{C.~A.~J. Klaassen}, \bibinfo{author}{J.~A. Wellner},
  \bibinfo{title}{Efficient estimation in the bivariate normal copula model:
  normal margins are least favourable}, \bibinfo{journal}{Bernoulli}
  \bibinfo{volume}{3} (\bibinfo{year}{1997}) \bibinfo{pages}{55--77}.
\bibitem[{Lam and Fan(2009)}]{Lam2009}
\bibinfo{author}{C.~Lam}, \bibinfo{author}{J.~Fan},
  \bibinfo{title}{Sparsistency and rates of convergence in large covariance
  matrix estimation}, \bibinfo{journal}{Ann. Stat.} \bibinfo{volume}{37}
  (\bibinfo{year}{2009}) \bibinfo{pages}{4254--4278}.
\bibitem[{L\^{e} and Husson(2008)}]{Le2008}
\bibinfo{author}{S.~L\^{e}}, \bibinfo{author}{F.~Husson},
  \bibinfo{title}{Senso{M}ine{R}: A package for sensory data analysis},
  \bibinfo{journal}{J. Sens. Stud.} \bibinfo{volume}{23} (\bibinfo{year}{2008})
  \bibinfo{pages}{14--25}.
\bibitem[{Ledoit and Wolf(2004)}]{Ledoit2004}
\bibinfo{author}{O.~Ledoit}, \bibinfo{author}{M.~Wolf}, \bibinfo{title}{A
  well-conditioned estimator for large-dimensional covariance matrices},
  \bibinfo{journal}{J. Multivar. Anal.} \bibinfo{volume}{88}
  (\bibinfo{year}{2004}) \bibinfo{pages}{365--411}.
\bibitem[{Llobell et~al.(2020)Llobell, Cariou, Vigneau, Labenne and
  Qannari}]{Llobell2020}
\bibinfo{author}{F.~Llobell}, \bibinfo{author}{V.~Cariou},
  \bibinfo{author}{E.~Vigneau}, \bibinfo{author}{A.~Labenne},
  \bibinfo{author}{E.~M. Qannari}, \bibinfo{title}{Analysis and clustering of
  multiblock datasets by means of the {STATIS} and {CLUSTATIS} methods.
  {A}pplication to sensometrics}, \bibinfo{journal}{Food Qual. Prefer.}
  \bibinfo{volume}{79} (\bibinfo{year}{2020}) \bibinfo{pages}{103520}.
\bibitem[{Marshall et~al.(2011)Marshall, Olkin and Arnold}]{Marshall2011}
\bibinfo{author}{A.~Marshall}, \bibinfo{author}{I.~Olkin},
  \bibinfo{author}{B.~Arnold}, \bibinfo{title}{Inequalities: theory of
  majorization and its applications}, volume~\bibinfo{volume}{2},
  \bibinfo{publisher}{Springer Series in Statistics, Springer},
  \bibinfo{year}{2011}.
\bibitem[{Massart(1990)}]{Massart1990}
\bibinfo{author}{P.~Massart}, \bibinfo{title}{The tight constant in the
  {D}voretzky-{K}iefer-{W}olfowitz inequality}, \bibinfo{journal}{Ann. Probab.}
  \bibinfo{volume}{18} (\bibinfo{year}{1990}) \bibinfo{pages}{1269--1283}.
\bibitem[{Medovikov and Prokhorov(2017)}]{Medovikov2017}
\bibinfo{author}{I.~Medovikov}, \bibinfo{author}{A.~Prokhorov},
  \bibinfo{title}{A new measure of vector dependence, with applications to
  financial risk and contagion}, \bibinfo{journal}{J. Financ. Econom.}
  \bibinfo{volume}{15} (\bibinfo{year}{2017}) \bibinfo{pages}{474--503}.
\bibitem[{Mordant and Segers(2022)}]{Mordant2022}
\bibinfo{author}{G.~Mordant}, \bibinfo{author}{J.~Segers},
  \bibinfo{title}{Measuring dependence between random vectors via optimal
  transport}, \bibinfo{journal}{J. Multivar. Anal.} \bibinfo{volume}{189}
  (\bibinfo{year}{2022}) \bibinfo{pages}{104912}.
\bibitem[{Nelsen(1996)}]{Nelsen1996}
\bibinfo{author}{R.~B. Nelsen}, \bibinfo{title}{Nonparametric measures of
  multivariate association}, in: \bibinfo{editor}{L.~R{\"u}schendorf},
  \bibinfo{editor}{B.~Schweizer}, \bibinfo{editor}{M.~D. Taylor} (Eds.),
  \bibinfo{booktitle}{IMS Lecture Notes - Monograph Series Vol. 28,
  Distributions with Fixed Marginals and Related Topics}, \bibinfo{year}{1996},
  pp. \bibinfo{pages}{223--232}.
\bibitem[{Nelsen(2006)}]{Nelsen2006}
\bibinfo{author}{R.~B. Nelsen}, \bibinfo{title}{An Introduction to Copulas},
  \bibinfo{publisher}{Springer Science and Business Media},
  \bibinfo{address}{New York}, \bibinfo{year}{2006}.
\bibitem[{Panaretos and Zemel(2019)}]{Panaretos2019}
\bibinfo{author}{V.~M. Panaretos}, \bibinfo{author}{Y.~Zemel},
  \bibinfo{title}{Statistical aspects of {W}asserstein distances},
  \bibinfo{journal}{Annu. Rev. Stat. Appl.} \bibinfo{volume}{6}
  (\bibinfo{year}{2019}) \bibinfo{pages}{405--431}.
\bibitem[{Ruymgaart et~al.(1972)Ruymgaart, Shorack and van
  Zwet}]{Ruymgaart1972}
\bibinfo{author}{F.~H. Ruymgaart}, \bibinfo{author}{G.~R. Shorack},
  \bibinfo{author}{W.~R. van Zwet}, \bibinfo{title}{Asymptotic normality of
  nonparametric tests for independence}, \bibinfo{journal}{Ann. Math. Stat.}
  \bibinfo{volume}{43} (\bibinfo{year}{1972}) \bibinfo{pages}{1122--1135}.
\bibitem[{Saulis and Statulevi\v{c}ius(1991)}]{Saulis1991}
\bibinfo{author}{L.~Saulis}, \bibinfo{author}{V.~A. Statulevi\v{c}ius},
  \bibinfo{title}{Limit theorems for Large Deviations},
  \bibinfo{publisher}{Kluwer Academic Publishers},
  \bibinfo{address}{Dordrecht}, \bibinfo{year}{1991}.
\bibitem[{Schmid and Schmidt(2007)}]{Schmid2007}
\bibinfo{author}{F.~Schmid}, \bibinfo{author}{R.~Schmidt},
  \bibinfo{title}{Multivariate extensions of {S}pearman’s rho and related
  statistics}, \bibinfo{journal}{Stat. Probab. Lett.} \bibinfo{volume}{77}
  (\bibinfo{year}{2007}) \bibinfo{pages}{407--416}.
\bibitem[{Sklar(1959)}]{Sklar1959}
\bibinfo{author}{A.~Sklar}, \bibinfo{title}{Fonctions de repartition \`a n
  dimensions et leurs marges}, \bibinfo{journal}{Publications de l’Institut
  Statistique de l’Universit\'{e} de Paris} \bibinfo{volume}{8}
  (\bibinfo{year}{1959}) \bibinfo{pages}{229--231}.
\bibitem[{Sz\'{e}kely et~al.(2007)Sz\'{e}kely, Rizzo and Bakirov}]{Szekely2007}
\bibinfo{author}{G.~J. Sz\'{e}kely}, \bibinfo{author}{M.~L. Rizzo},
  \bibinfo{author}{N.~K. Bakirov}, \bibinfo{title}{Measuring and testing
  dependence by correlation of distances}, \bibinfo{journal}{Ann. Stat.}
  \bibinfo{volume}{35} (\bibinfo{year}{2007}) \bibinfo{pages}{2769--2794}.
\bibitem[{Takatsu(2011)}]{Takatsu2011}
\bibinfo{author}{A.~Takatsu}, \bibinfo{title}{{W}asserstein geometry of
  {G}aussian measures}, \bibinfo{journal}{Osaka J. Math.} \bibinfo{volume}{48}
  (\bibinfo{year}{2011}) \bibinfo{pages}{1005--1026}.
\bibitem[{Thompson and Therianos(1972)}]{Thompson1972}
\bibinfo{author}{R.~C. Thompson}, \bibinfo{author}{S.~Therianos},
  \bibinfo{title}{Inequalities connecting the eigenvalues of a hermitian matrix
  with the eigenvalues of complementary principal submatrices},
  \bibinfo{journal}{Bull. Aust. Math. Soc.} \bibinfo{volume}{6}
  (\bibinfo{year}{1972}) \bibinfo{pages}{117--132}.
\bibitem[{Tsanas et~al.(2014)Tsanas, Little, Fox and Ramig}]{Tsanas2014}
\bibinfo{author}{A.~Tsanas}, \bibinfo{author}{M.~A. Little},
  \bibinfo{author}{C.~Fox}, \bibinfo{author}{L.~O. Ramig},
  \bibinfo{title}{Objective automatic assessment of rehabilitative speech
  treatment in {P}arkinson’s disease}, \bibinfo{journal}{IEEE Trans. Neural
  Syst. Rehabil. Eng.} \bibinfo{volume}{22} (\bibinfo{year}{2014})
  \bibinfo{pages}{181--190}.
\bibitem[{Villani(2008)}]{Villani2008}
\bibinfo{author}{C.~Villani}, \bibinfo{title}{Optimal Transport: Old and New},
  volume \bibinfo{volume}{338}, \bibinfo{publisher}{Springer Science \&
  Business Media}, \bibinfo{year}{2008}.
\bibitem[{Wang(2014)}]{Wang2014}
\bibinfo{author}{H.~Wang}, \bibinfo{title}{Coordinate descent algorithm for
  covariance graphical lasso}, \bibinfo{journal}{Stat. Comput.}
  \bibinfo{volume}{24} (\bibinfo{year}{2014}) \bibinfo{pages}{521--529}.
\bibitem[{Warton(2008)}]{Warton2008}
\bibinfo{author}{D.~I. Warton}, \bibinfo{title}{Penalized normal likelihood and
  ridge regularization of correlation and covariance matrices},
  \bibinfo{journal}{J. Am. Stat. Assoc.} \bibinfo{volume}{103}
  (\bibinfo{year}{2008}) \bibinfo{pages}{340--349}.
\bibitem[{Zhang et~al.(2020)Zhang, Lu and Wang}]{Zhang2020}
\bibinfo{author}{L.~Zhang}, \bibinfo{author}{D.~Lu}, \bibinfo{author}{X.~Wang},
  \bibinfo{title}{The essential dependence for a group of random vectors},
  \bibinfo{journal}{Commun. Stat. - Theory Methods} \bibinfo{volume}{50}
  (\bibinfo{year}{2020}) \bibinfo{pages}{1--37}.
\bibitem[{Zou and Li(2008)}]{Zou2008}
\bibinfo{author}{H.~Zou}, \bibinfo{author}{R.~Li}, \bibinfo{title}{One-step
  sparse estimates in nonconcave penalized likelihood models},
  \bibinfo{journal}{Ann. Stat.} \bibinfo{volume}{36} (\bibinfo{year}{2008})
  \bibinfo{pages}{1509--1533}.

\end{thebibliography}


\newpage
\noindent
\renewcommand{\theequation}{S\arabic{section}.\arabic{equation}}
\renewcommand{\thesection}{S\arabic{section}} 
\renewcommand{\theexample}{S\arabic{example}} 
\renewcommand{\thetable}{S\arabic{table}}  
\renewcommand{\thefigure}{S\arabic{figure}}
\setcounter{page}{1}
\setcounter{example}{0}
\setcounter{section}{0}
\setcounter{subsection}{0}
\setcounter{equation}{0}
\setcounter{figure}{0}
\setcounter{table}{0}
\begin{center}
{\large {Supplementary Material} \\
to }\\
{\Large{\bf High-dimensional copula-based Wasserstein dependence
}}
\vspace*{0.4 cm}

\noindent
by 
\vspace*{0.24 cm}

\noindent
{\large Steven De Keyser and Ir\`ene Gijbels}
\end{center}

\begin{figure}[h!] 
\includegraphics[scale = 0.4]{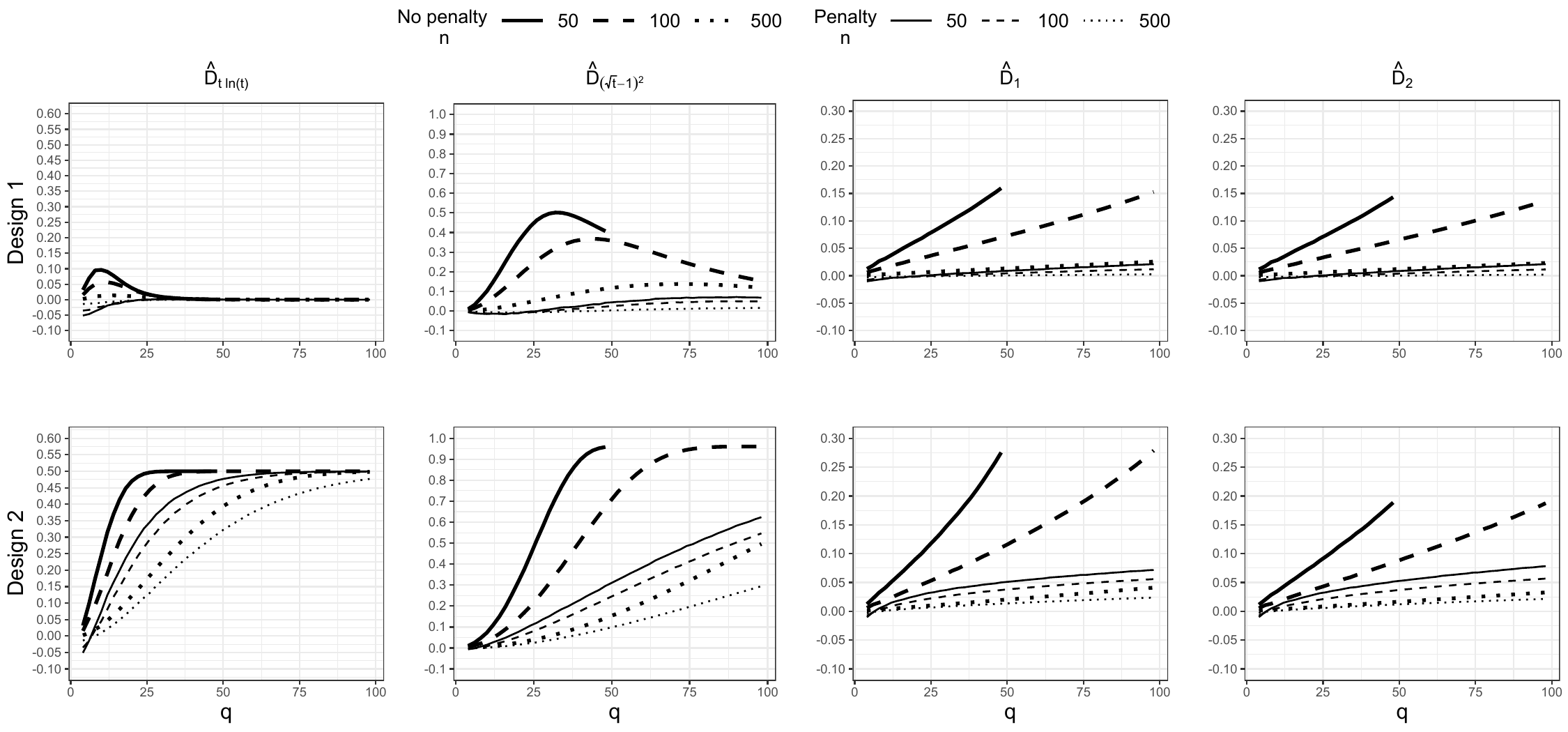}
\caption{Monte Carlo bias of the estimators $\mathcal{D}_{\bigcdot}(\widehat{\mathbf{R}}_{\text{R},n})$ (penalty) and $\mathcal{D}_{\bigcdot}(\widehat{\mathbf{R}}_{n})$ (no penalty) for $\mathcal{D}_{t \ln(t)}, \mathcal{D}_{(\sqrt{t}-1)^{2}}, \mathcal{D}_{1}$ and $\mathcal{D}_{2}$, based on $1000$ replications with sample sizes $n = 50,100,500$ as a function of $q$, in two different designs.}
\label{fig: BiasRidge}
\end{figure} 

\begin{figure}[h!] 
\includegraphics[scale = 0.4]{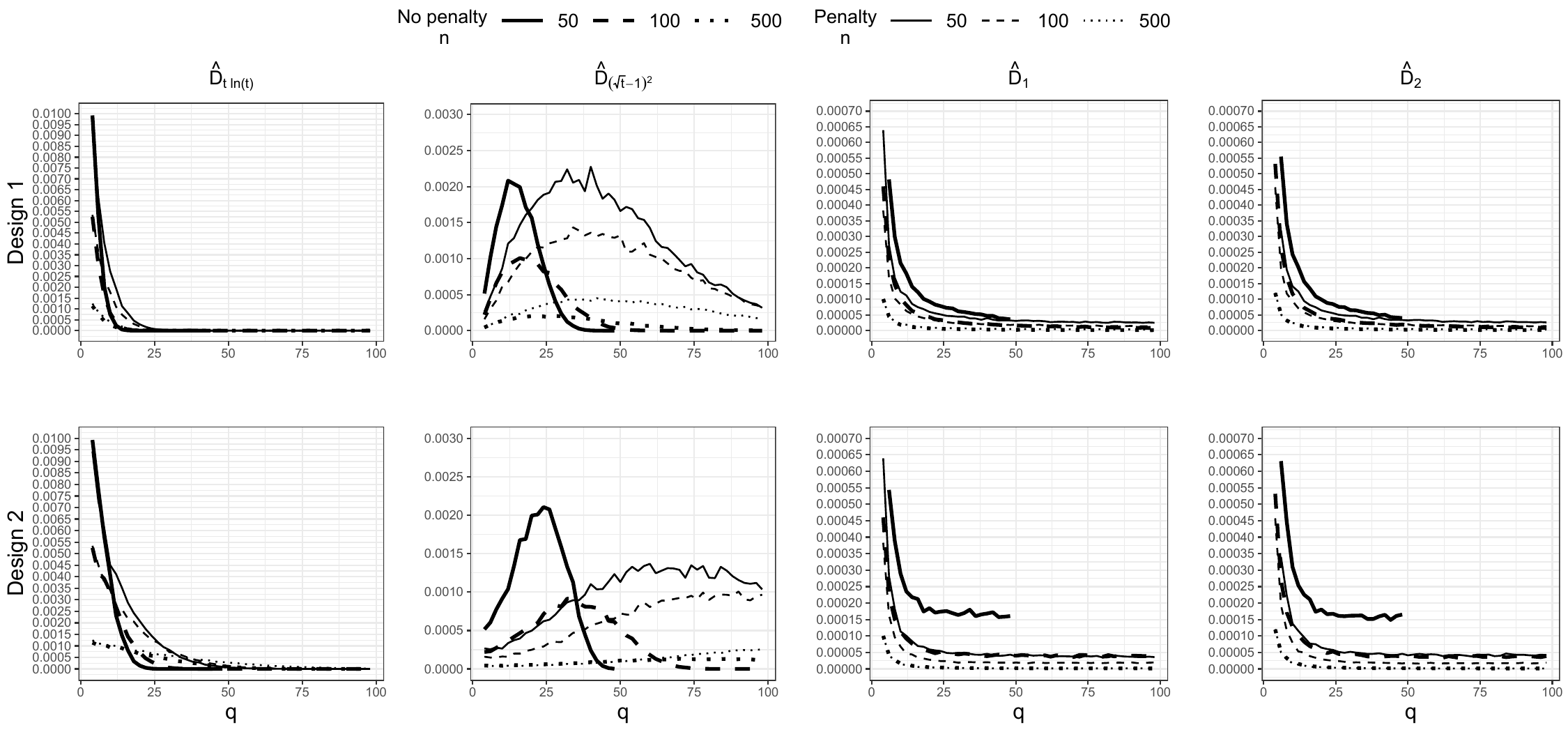}
\caption{Monte Carlo variance of the estimators $\mathcal{D}_{\bigcdot}(\widehat{\mathbf{R}}_{\text{R},n})$ (penalty) and $\mathcal{D}_{\bigcdot}(\widehat{\mathbf{R}}_{n})$ (no penalty) for $\mathcal{D}_{t \ln(t)}, \mathcal{D}_{(\sqrt{t}-1)^{2}}, \mathcal{D}_{1}$ and $\mathcal{D}_{2}$, based on $1000$ replications with sample sizes $n = 50,100,500$ as a function of $q$, in two different designs.}
\label{fig: VarianceRidge}
\end{figure} 

\begin{figure}[h!] 
\hspace*{-0.5cm}
\includegraphics[scale = 0.07]{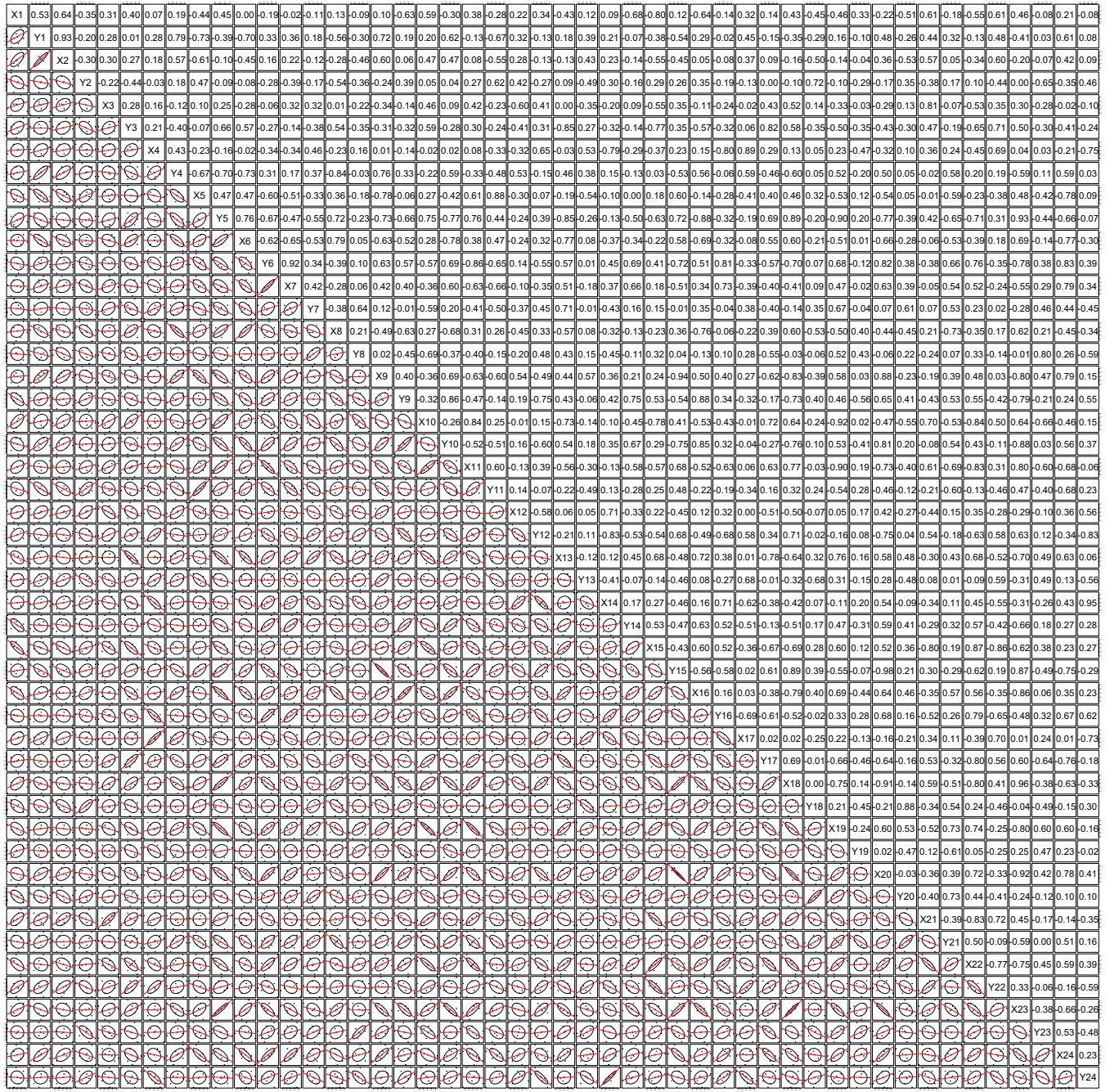}
\caption{Pairwise scatterplots of normal scores of $X$-$Y$ coordinates of $24$ people rating smoothies (rating of the $i$'th person is described by the coordinates $(X_{i},Y_{i})$).}
\end{figure} 

\begin{figure}[h!] \centering
\includegraphics[scale = 0.8]{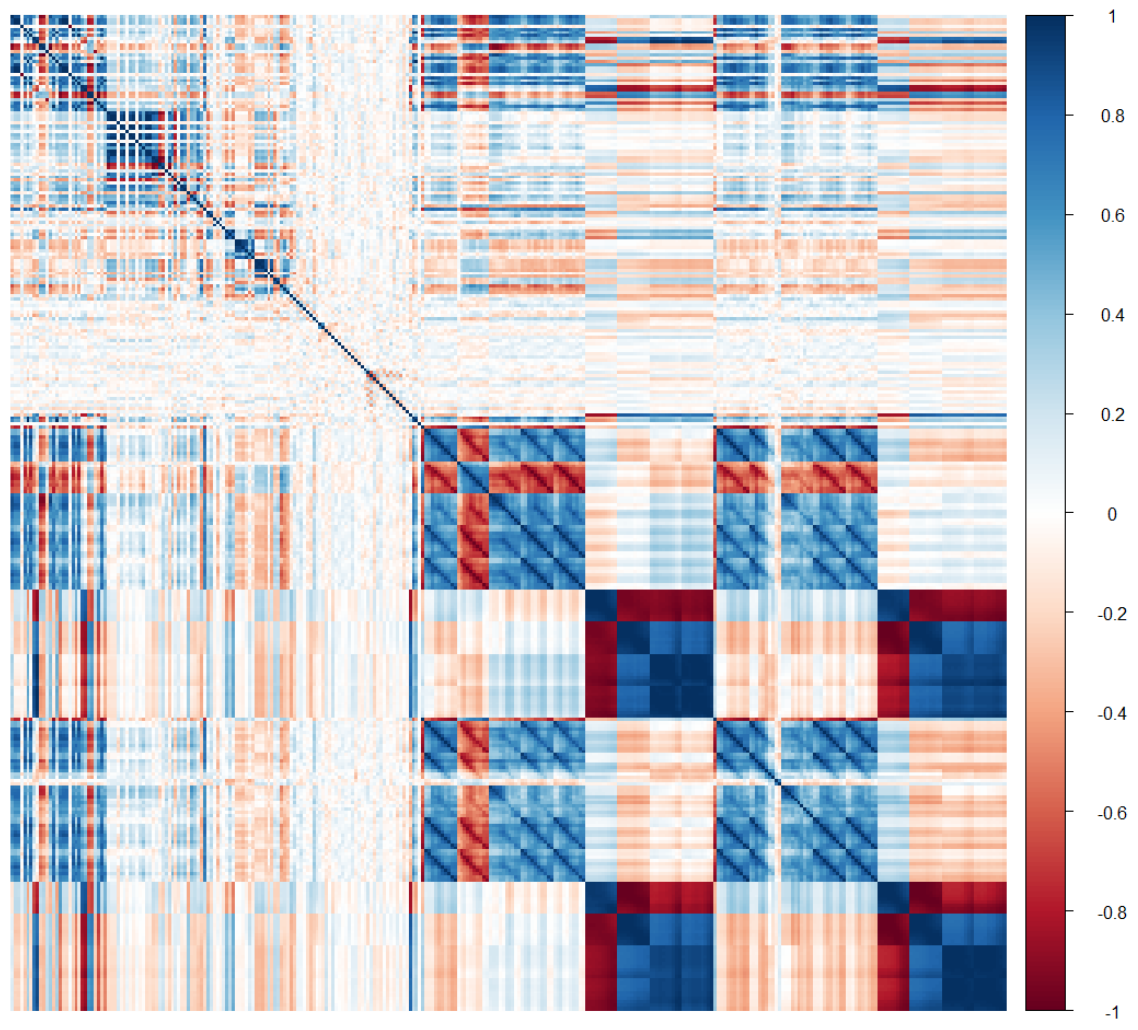}
\caption{Sample normal scores rank correlation matrix of the $310$ dysphonia measures in the LSVT voice rehabilitation dataset.}
\end{figure} 

\begin{figure}[h!] \centering
\includegraphics[scale = 0.4]{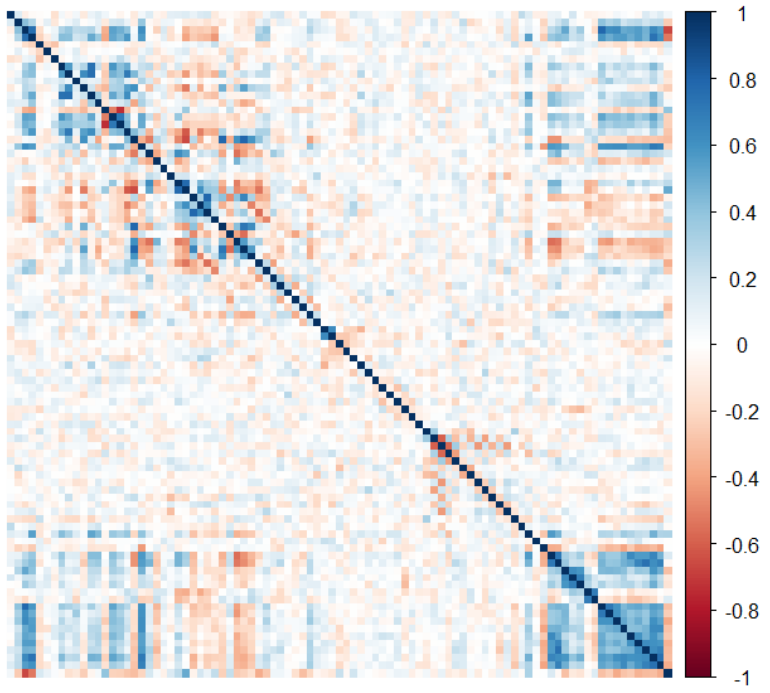}
\includegraphics[scale = 0.42]{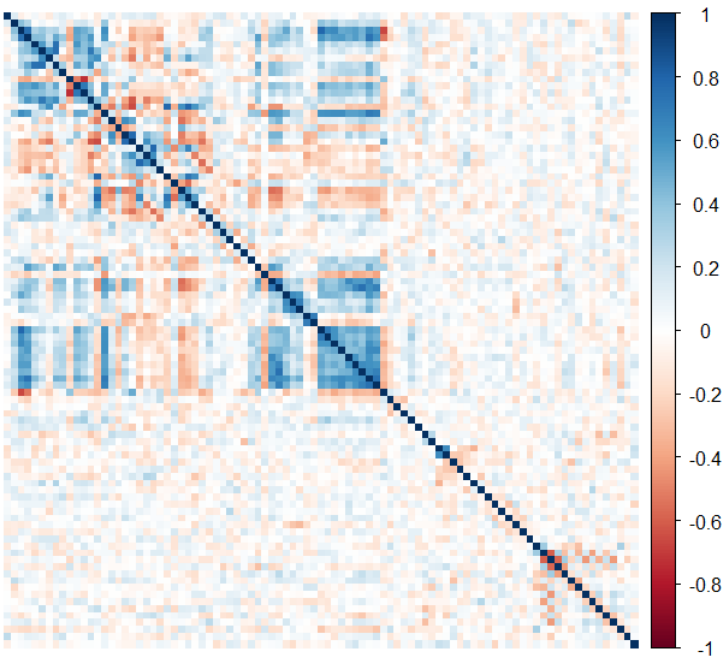}
\includegraphics[scale = 0.42]{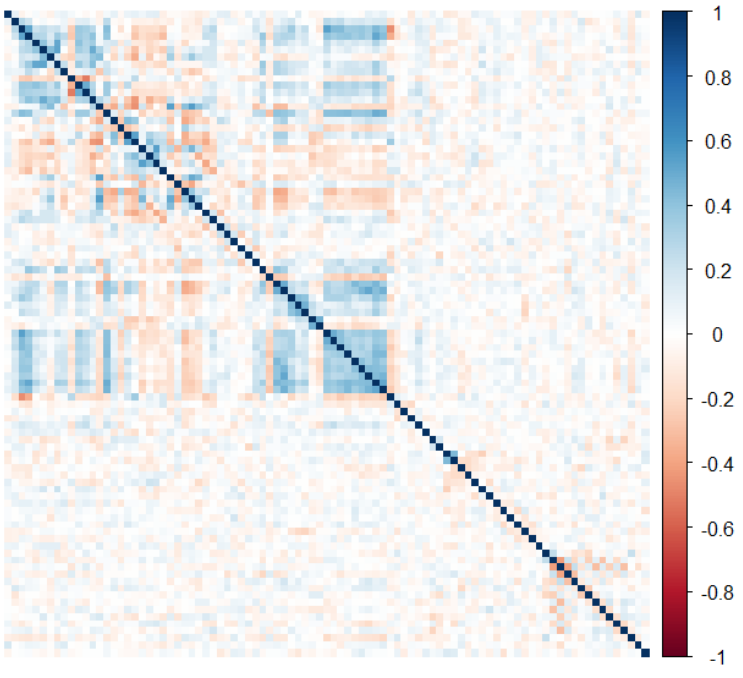}
\caption{Sample normal scores rank correlation matrix of the $91$ dysphonia measures before rearranging (left), after rearranging (middle) and after rearranging, estimated with a ridge penalty (right).}
\end{figure} 

\begin{figure}[h!] \centering
\includegraphics[scale = 0.67]{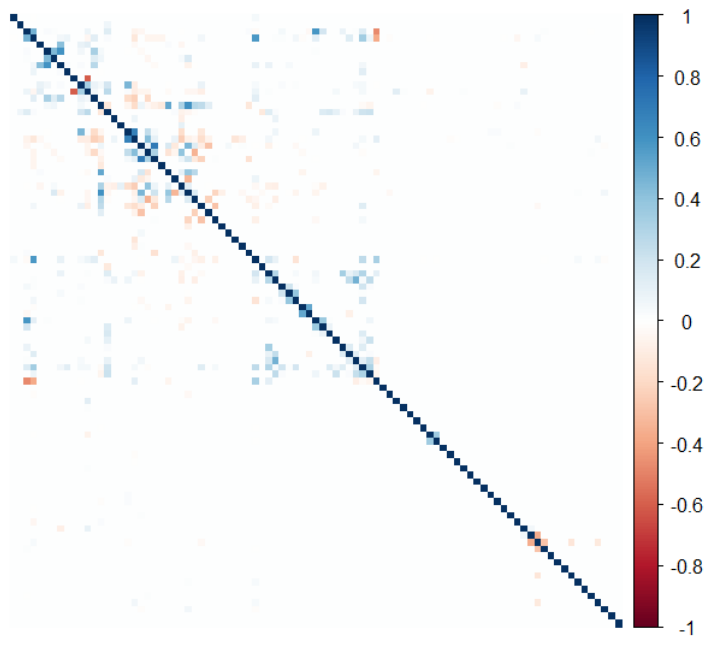}
\includegraphics[scale = 0.67]{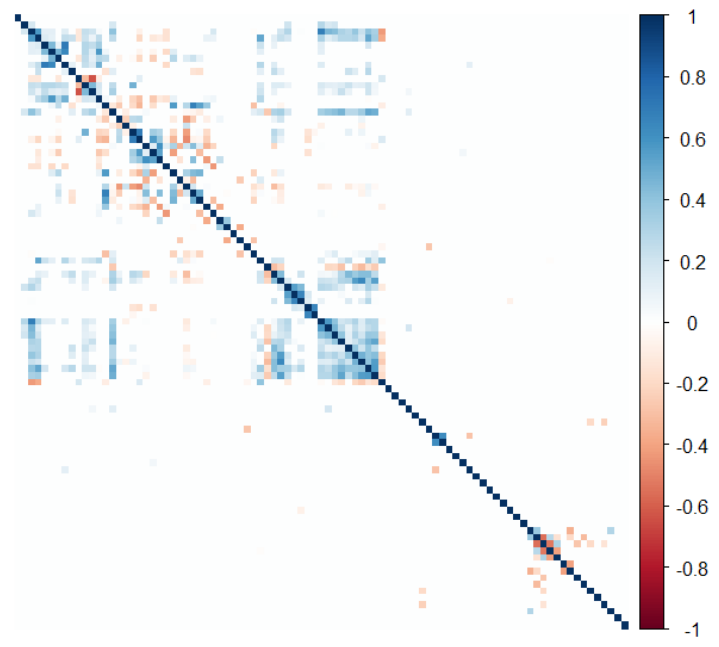}
\includegraphics[scale = 0.67]{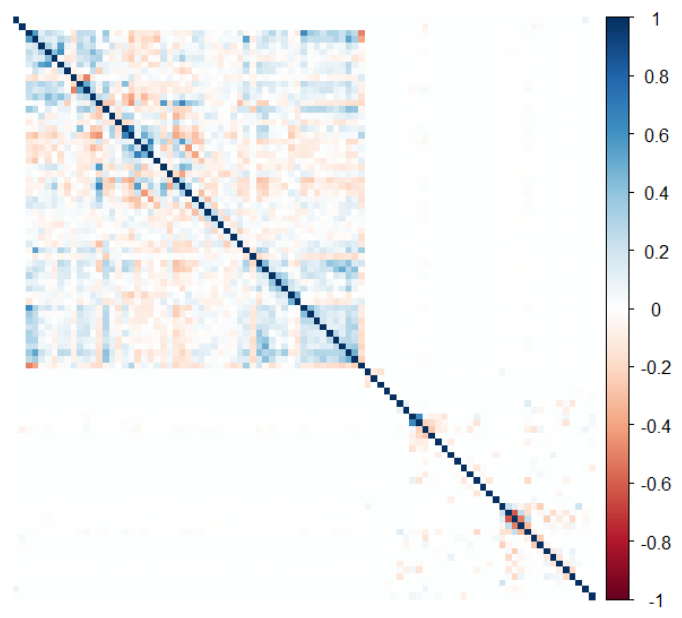}
\includegraphics[scale = 0.67]{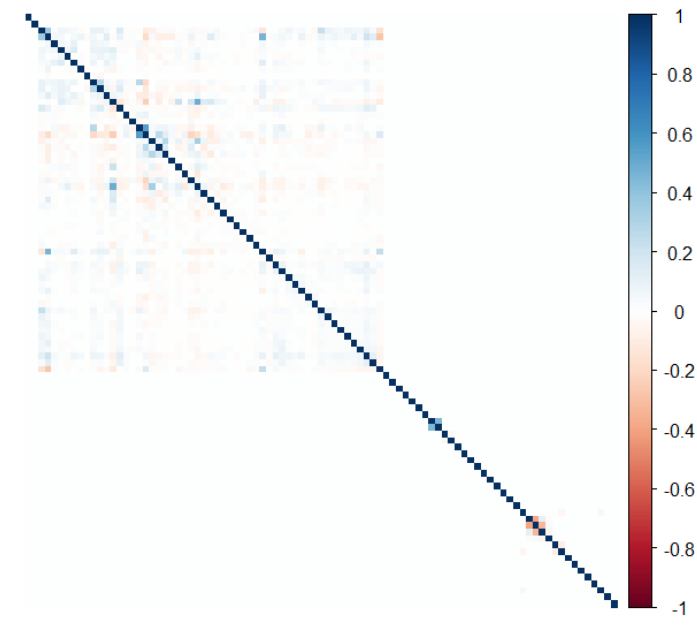}
\caption{Sample normal scores rank correlation matrix of the $91$ dysphonia measures after rearranging, estimated via lasso with $\omega_{n} = 0.623$ (top left), adaptive lasso with $\rho_{n} = 0.287$ (top right), group lasso with $\omega_{n} = 0.445$ (bottom left), and group lasso with $\omega_{n} = 0.670$ (bottom right). 
}
\end{figure} 
\end{document}